\documentclass{scrartcl}
\def\arXiv{}

\usepackage{basicstuff}

\usepackage{tikz}
\usepackage{calc}

\usepackage{xcolor}

\DeclareMathOperator{\pos}{pos} %
\DeclareMathOperator{\twin}{twin} %

\DeclareMathOperator{\pot}{pot} %
\DeclareMathOperator{\bridges}{bridges} %
\DeclareMathOperator{\expan}{exp} %

\newcommand{\ord}{\mathrm{ord}}

\newcommand{\eps}{\varepsilon}

\newcommand{\rst}[2]{\left.#1\right|_{#2}}

\newcommand*\circled[1]{ %
  \protect\tikz[baseline=(char.base)]{ %
    \protect\node[shape=circle,draw,inner sep=0.7pt,scale=0.85] (char) {\normalfont#1};}} %

\newcommand{\1}[1]{{\normalfont \ensuremath{#1^{\tiny\circled{1}}}}} %
\newcommand{\2}[1]{{\normalfont \ensuremath{#1^{\tiny\circled{2}}}}} %
\newcommand{\ii}[1]{{\normalfont \ensuremath{#1^{\tiny\circled{i}}}}} %
\newcommand{\E}{\mathcal E}
\newcommand{\T}{\mathcal T}

\DeclareMathOperator{\skel}{skel}

\DeclareMathOperator{\LCA}{LCA}

\ifdefined\arXiv

\usepackage{arXiv}

\newlength{\tmp}
\usepackage{thmtools}
\declaretheoremstyle[
shaded={rulecolor=black,bgcolor=white,rulewidth=0.5pt,
  margin=6pt, textwidth=\tmp},
notefont=\bfseries, notebraces={}{},
bodyfont=\normalfont,
postheadspace=0.5em,
headformat=\NAME~\NUMBER:\NOTE]{step}
\declaretheorem[style=step,name=Step]{sefe-bico-step}

\title{\Large Simultaneous Embedding:\\ Edge Orderings, Relative
  Positions, Cutvertices\thanks{Partly done within GRADR -- EUROGIGA
    project no. 10-EuroGIGA-OP-003.}}

\date{}

\author{Thomas Bläsius\thanks{Department of Theoretical Informatics,
    Karlsruhe Institute of Technology (KIT) \texttt{blaesius@kit.edu,
      annette.karrer@student.kit.edu}} \myand Annette Karrer$^\dagger$
  \myand Ignaz Rutter\thanks{Department of Theoretical Informatics,
    Karlsruhe Institute of Technology (KIT), Karlsruhe and Department
    of Applied Mathematics, Charles University, Prague.
    \texttt{rutter@kit.edu}.  Work was supported by a fellowship
    within the Postdoc-Program of the German Academic Exchange Service
    (DAAD).}}

\fi

\begin{document}

\maketitle

\begin{abstract}
  A simultaneous embedding (with fixed edges) of two graphs~$\1G$
  and~$\2G$ with common graph~$G=\1G \cap \2G$ is a pair of planar
  drawings of~$\1G$ and~$\2G$ that coincide on~$G$.  It is an open
  question whether there is a polynomial-time algorithm that decides
  whether two graphs admit a simultaneous embedding (problem
  \textsc{Sefe}).

  In this paper, we present two results.  First, a set of three
  linear-time preprocessing algorithms that remove certain
  substructures from a given \textsc{Sefe} instance, producing a set of
  equivalent \textsc{Sefe} instances without such substructures.  The
  structures we can remove are (1) cutvertices of the \emph{union
    graph}~$G^\cup = \1G \cup \2G$, (2) most separating pairs of
  $G^\cup$, and (3) connected components of~$G$ that are biconnected
  but not a cycle.

  Second, we give an~$O(n^3)$-time algorithm solving \textsc{Sefe} for
  instances with the following restriction.  Let $u$ be a pole of a
  P-node $\mu$ in the SPQR-tree of a block of $\1G$ or $\2G$.  Then at
  most three virtual edges of $\mu$ may contain common edges incident
  to $u$.  All algorithms extend to the sunflower case, i.e., to the
  case of more than three graphs pairwise intersecting in the same
  common graph.
\end{abstract}

\section{Introduction}
\label{sec:introduction}

A simultaneous embedding of two graphs~$\1G$ and~$\2G$ with common
graph~$G=\1G \cap \2G$ is a pair of planar drawings of~$\1G$
and~$\2G$, that coincide on~$G$.  The problem to decide whether a
simultaneous embedding exists is called \textsc{Sefe} (Simultaneous
Embedding with Fixed Edges).  This definition extends to more than two
graphs.  For three graphs \textsc{Sefe} is
NP-complete~\cite{gjp-sgefe-06}.  In the \emph{sunflower case} it is
required that every pair of input graphs has the same intersection.
See~\cite{bkr-sepg-13} for a survey on \textsc{Sefe} and related
problems.

There are two fundamental approaches to solving \textsc{Sefe} in the
literature.  The first approach is based on the characterization of
Jünger and Schulz~\cite{js-igsefe-09} stating that finding a
simultaneous embedding of two graphs $\1G$ and $\2G$ with common graph
$G$ is equivalent to finding planar embeddings of $\1G$ and $\2G$ that
induce the same embedding on $G$.  The second very recent approach by
Schaefer~\cite{s-ttp-13} is based on Hanani-Tutte-style redrawing
results.  One tries to characterize the existence of a \textsc{Sefe} via
the existence of drawings of the union graph $G^\cup$ where no two
independent edges of the same graph cross an odd number of times.  The
existence of such drawings can be expressed using a linear system of
boolean equations.

When following the first approach, we need two things to describe the
planar embedding of the common graph $G$.  First, for each vertex $v$,
a cyclic order of incident edges around $v$.  Second, for every pair
of connected components $H$ and $H'$ of $G$, the face $f$ of $H$
containing $H'$.  We call this relationship the \emph{relative
  position} of $H'$ with respect to $H$.  To find a simultaneous
embedding, one needs to find a pair of planar embeddings that induce
the same cyclic edge orderings (\emph{consistent edge orderings}) and
the same relative positions (\emph{consistent relative positions}) on
the common graph $G$.

Most previous results use the first approach but none of them
considers both consistent edge orderings and relative positions.  Most
of them assume the common graph to be connected or to contain no
cycles.  The strongest results of this type are the two linear-time
algorithms for the case that $G$ is biconnected by Haeupler et
al.~\cite{hjl-tspwc-13} and by Angelini et al.~\cite{adfpr-tsegi-12}
and a quadratic-time algorithm for the case where $\1G$ and $\2G$ are
biconnected and $G$ is connected~\cite{br-spoacep-13}.  In the latter
result, \textsc{Sefe} is modeled as an instance of the problem
\textsc{Simultaneous PQ-Ordering}.  On the other hand, there is a
linear-time algorithm for \textsc{Sefe} if the common graph consists
of disjoint cycles~\cite{br-drpse-15}, which requires to ensure
consistent relative positions but makes edge orderings trivially
consistent.

The advantage of the second approach (Hanani-Tutte) is that it
implicitly handles both, consistent edge orderings and consistent
relative positions, at the same time.  Thus, the results by
Schaefer~\cite{s-ttp-13} are the first that handle \textsc{Sefe}
instances where the common graph consists of several, non-trivial
connected components.  He gives a polynomial-time algorithm for the
cases where each connected component of the common graph is
biconnected or has maximum degree~3.  Although this approach is
conceptionally simple, very elegant, and combines several notions of
planarity within a common framework, it has two disadvantages.  The
running time of the algorithms are quite high and the high level of
abstraction makes it difficult to generalize the results.


\paragraph{Contribution \& Outline.}

In this paper, we follow the first approach and show how to enforce
consistent edge orderings and consistent relative positions at the
same time, by combining different recent approaches, namely the
algorithm by Angelini et al.~\cite{adfpr-tsegi-12}, the result on
\textsc{Simultaneous PQ-Ordering}~\cite{br-spoacep-13} for consistent
edge orderings, and the result on disjoint cycles~\cite{br-drpse-15}
for consistent relative positions.  Note that the relative positions
of connected components to each other are usually expressed in terms
of faces (containing the respective component).  This is no longer
possible if the embeddings, and thus the set of faces, of connected
components are not fixed.  To overcome this issue, we show that these
relative positions can be expressed in terms of relative positions
with respect to a cycle basis.  In addition to that, we are able to
handle certain cutvertices of $\1G$ and $\2G$.

More precisely, we classify a vertex $v$ to be a \emph{union
  cutvertex}, a \emph{simultaneous cutvertex}, and an \emph{exclusive
  cutvertex} if $v$ is a cutvertex of $G^\cup$, of $\1G$ and $\2G$ but
not of $G^\cup$, and of $\1G$ but not $\2G$ or the other way around,
respectively.  Similarly, we can define \textsc{union separating
  pairs} to be separating pairs in $G^\cup$.  We present several
preprocessing algorithms that simplify given instances of \textsc{Sefe};
see Section~\ref{sec:prepr-algor}.  Besides a very technical
preprocessing step (Section~\ref{sec:spec-bridg-comm-face-constr}),
they remove union cutvertices and most (but not all) union separating
pairs; see Theorem~\ref{thm:only-special-union-separating-pairs}, and
replace connected components of $G$ that are biconnected with cycles.
They run in linear time and can be applied independently.  The latter
algorithm together with the linear-time algorithm for disjoint
cycles~\cite{br-drpse-15} improves the result by
Schaefer~\cite{s-ttp-13} for instances where every connected component
of $G$ is biconnected to linear time.

In Section~\ref{sec:graph-with-common} we show how to solve instances
that have \emph{common P-node degree~3} and simultaneous cutvertices
of \emph{common degree} at most~3 in cubic time.  A vertex has common
degree~$k$ if it is a common vertex with degree~$k$ in $G$.  An
instance has common P-node degree~$k$ if, for each pole $v$ of a
P-node $\mu$ (of a block) of the input graphs, at most~$k$ virtual
edges of $\mu$ contain common edges incident to $v$.  This result
relies heavily on the preprocessing algorithms excluding certain
structures.  Together with the preprocessing steps, this includes the
case where every connected component of $G$ is biconnected, has
maximum degree~3, or is outerplanar with maximum degree~3 cutvertices.
As before, this approach also applies to the sunflower case.

We would like to point out that the conference version of this article
contained a flaw in the handling of relative positions that are
determined by P-nodes.  It was erroneously claimed that certain
constraints could be expressed in terms of linear Equations over
$\mathbb{F}_2$.  The issue has been fixed by an additional
preprocessing step (see
Section~\ref{sec:spec-bridg-comm-face-constr}), which allows to
exclude the problematic case when treating relative positions decided
by P-nodes.  See Section~\ref{sec:cons-relat-posit} for additional
details.

\section{Preliminaries}
\label{sec:preliminaries}

\paragraph{Connectivity \& SPQR-trees.}

A graph is \emph{connected} if there exists a path between any pair of
vertices.  A \emph{separating $k$-set} is a set of $k$ vertices whose
removal disconnects the graph.  Separating 1-sets and 2-sets are
\emph{cutvertices} and \emph{separating pairs}, respectively.  A
connected graph is \emph{biconnected} if it has no cut vertex and
\emph{triconnected} if it has no separating pair.  The maximal
biconnected components of a graph are called \emph{blocks}.  The
\emph{split components} with respect to a separating $k$-set are the
maximal subgraphs that are not disconnected by removing the separating
$k$-set.

The SPQR-tree $\T$ of a biconnected graph $G$ represents the
decomposition of $G$ along its \emph{split pairs}, where a split pair
is either a separating pair or a pair of adjacent
vertices~\cite{dt-omtc-96}.  It can be computed in linear
time~\cite{gm-lti-00}.

Let $\{s, t\}$ be a split pair and let~$H_1$ and~$H_2$ be two
subgraphs of $G$ such that $H_1 \cup H_2 = G$ and $H_1 \cap H_2 = \{s,
t\}$.  Consider the tree consisting of two nodes $\mu_1$ and $\mu_2$
associated with the graphs $H_1 + st$ and $H_2 + st$, respectively.
These graphs are called \emph{skeletons} of the nodes $\mu_i$, denoted
by $\skel(\mu_i)$, and the special edge $st$ is a \emph{virtual edge}.
Let $\eps_1 = st$ and $\eps_2 = st$ be the virtual edges connecting
$s$ and $t$ in $\skel(\mu_1)$ and $\skel(\mu_2)$, respectively.  The
two edges $\eps_1$ and $\eps_2$ are \emph{twins} and denote this
relationship by $\twin(\eps_1) = \eps_2$ (and vice versa).  We say
that the neighbor of $\mu_1$ \emph{corresponding} to the virtual edge
$\eps_1$ is $\mu_2$.  Conversely, $\mu_1$ corresponds to the virtual
edge in $\skel(\mu_2)$.  In this way, the edge between $\mu_1$ and
$\mu_2$ represents the twin-relationship between $\eps_1$ and
$\eps_2$; see Figure~\ref{fig:spqr-tree}a for an example.  We can
iterate this decomposition process on the graphs $\skel(\mu_1)$ and
$\skel(\mu_2)$; see Figure~\ref{fig:spqr-tree}b.

\begin{figure}
  \centering
  \includegraphics[page=1]{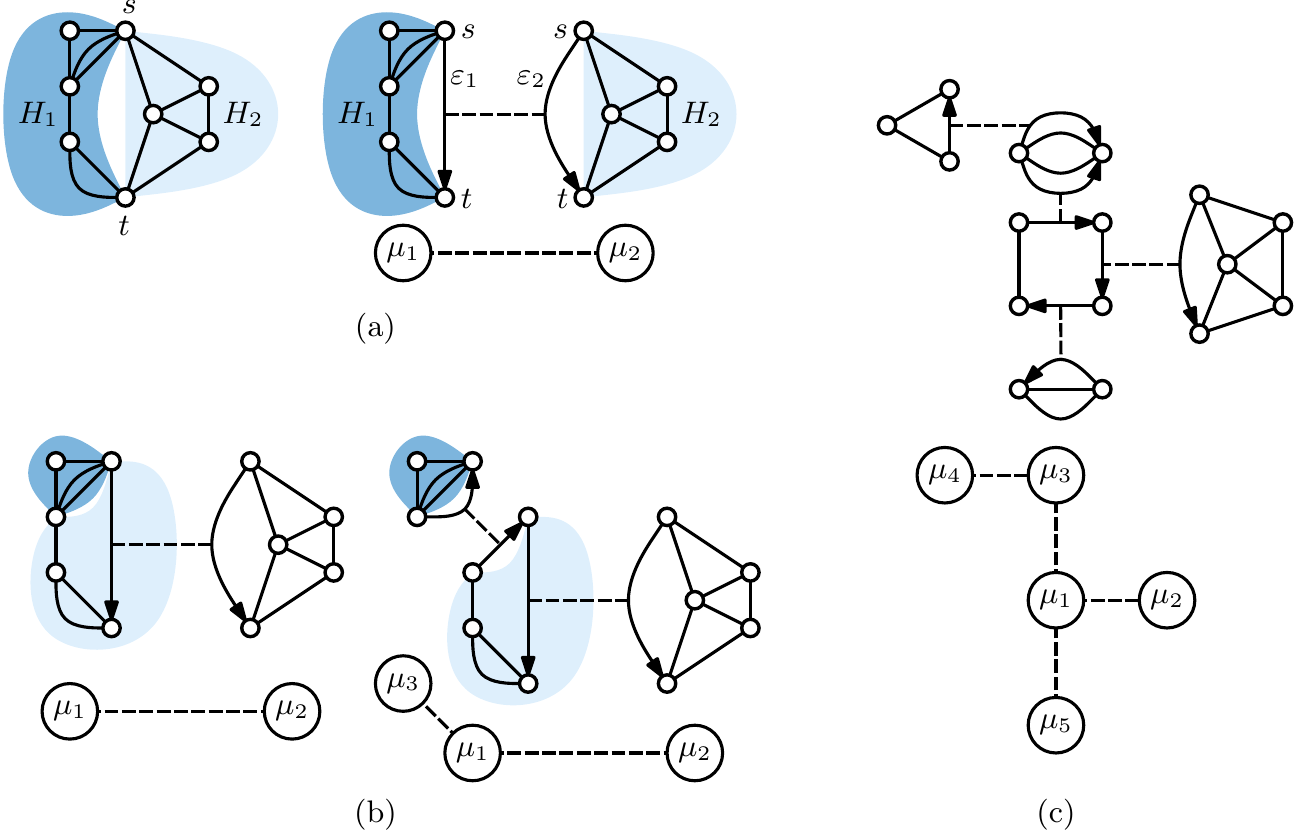}
  \caption{(a)~A single decomposition step with respect to the split
    pair $\{s, t\}$.  (b)~Continued decomposition.  (c)~The final
    SPQR-tree.}
  \label{fig:spqr-tree}
\end{figure}

Applying this kind of decomposition systematically yields the
SPQR-tree $\T$; see Figure~\ref{fig:spqr-tree}c.  The skeletons of the
internal nodes of $\T$ are either a cycle (S-node), a bunch of
parallel edges (P-node) or a triconnected planar graph (R-node).  All
edges in these skeletons are virtual edges.  The leaves are Q-nodes
and their skeleton consists of two vertices connected by a virtual and
a normal edge.  Sometimes it is more convenient to consider SPQR-trees
without their Q-nodes.  In this case, the P-, S-, and R-nodes can also
contain non-virtual edges (as in Figure~\ref{fig:spqr-tree}c).

When we choose a planar embedding for the skeleton of each node of the
SPQR-tree $\mathcal T$, this induces a planar embedding for $G$.
Conversely, fixing the planar embedding of $G$ determines the
embeddings of all skeletons.  Thus, the combination of all planar
embeddings of all skeletons is in one-to-one correspondence with the
planar embeddings of $G$.  Hence, the SPQR-tree breaks the complicated
embedding choices for $G$ (on the sphere, i.e., up to the choice of
the outer face) down to embedding choices of the skeletons.  These
remaining choices are very simple.  Skeletons of S-nodes are cycles
and thus have a unique planar embedding.  For P-nodes we can reorder
the parallel edges arbitrarily and the embedding of R-node skeletons
is fixed up to a flip (i.e., up to mirroring the embedding).

Assume $\mathcal T$ is rooted at an arbitrary node.  In this case, the
skeleton of every node (except for the root) has a unique virtual edge
corresponding to its parent in $\mathcal T$.  We call this virtual
edge the \emph{parent edge} and its endpoints the \emph{poles}.  We
recursively define the \emph{pertinent graph} of a node $\mu$ of
$\mathcal T$.  If $\mu$ is a Q-node, its pertinent graph is the
non-virtual edge in $\skel(\mu)$.  If $\mu$ is an inner node, the
pertinent graph of $\mu$ is obtained by deleting the parent edge in
$\skel(\mu)$ and replacing each remaining virtual edge with the
pertinent graph of the corresponding child.

Let $\eps$ be a virtual edge in $\skel(\mu)$ and let $\mu'$ be the
corresponding neighbor of $\mu$.  The \emph{expansion graph}
$\expan(\eps)$ of $\eps$ is the pertinent graph of $\mu'$ when
choosing $\mu$ as root.  Note that the expansion and pertinent graphs
are very similar concepts.  However, in most cases we use the
expansion graph as it is independent of the root of the SPQR-tree (and
is still defined if $\mathcal T$ is unrooted).  Intuitively, the
expansion graph of a virtual edge is the graph that is represented by
that virtual edge.  Note that replacing every virtual edge in
$\skel(\mu)$ (for any node $\mu$ of $\mathcal T$) with its expansion
graph yields the graph $G$.  A vertex in $\expan(\eps)$ is an
\emph{inner vertex} if it is not an endvertex of $\eps$.

\begin{figure}
  \centering
  \includegraphics[page=1]{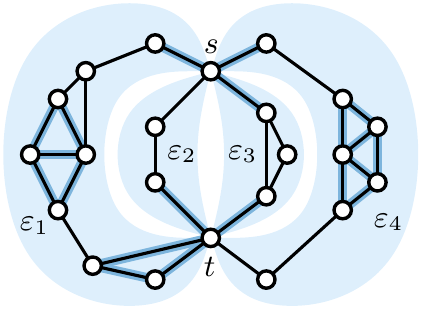}
  \caption{A P-node of $\1G$ with virtual edges $\eps_1, \dots,
    \eps_4$.  The node has common P-node degree $3$; for $s$ the
    virtual edges $\eps_1$, $\eps_3$, and $\eps_4$ count; for $t$ the
    virtual edges $\eps_1$, $\eps_2$, and $\eps_3$ count.}
  \label{fig:common-p-node-degree}
\end{figure}

Let $\1\T$ be the SPQR-tree of a block of $\1G$ in an instance of
\textsc{Sefe} and let $G$ be the common graph.  Let further $\mu$ be a
P-node of $\1\T$.  We say that $\mu$ has \emph{common P-node degree}
$k$ if both vertices in $\skel(\mu)$ are incident to common edges in
the expansion graphs of at most $k$ virtual edges (note that these can
be different edges for the two vertices); see
Figure~\ref{fig:common-p-node-degree} for an example.  We say that $\1G$
has common P-node degree $k$ if each P-node in the SPQR-tree of each
block of $\1G$ has common P-node degree $k$.  If this is the case for
$\1G$ and $\2G$, we say that the instance of \textsc{Sefe} has common
P-node degree~$k$.

We use the following conventions to make handling SPQR-trees more
convenient.  In most cases (as above) we explicitly name the
SPQR-trees we consider (e.g., $\mathcal T$, or $\1\T$).  However,
sometimes it is more convenient to write $\mathcal T(G)$ to denote the
SPQR-tree of a given graph $G$.  The SPQR-tree is only defined for
biconnected graphs.  With the SPQR-tree of a non-biconnected graph, we
implicitly mean a collection of SPQR-trees, one for each block.  For
an S-, P-, Q-, or R-node $\mu$ of the SPQR-tree of a graph $G$, we
also say that $\mu$ is an S-, P-, Q-, or R-node of $G$, respectively.
These conventions for example simplify the statement ``let $\mu$ be a
P-node of the SPQR-tree of a block of $G$'' to ``let $\mu$ be a P-node
of $G$''.

\paragraph{PQ-trees.}

A PQ-tree, originally introduced by Booth and
Lueker~\cite{bl-tcopi-76}, is a tree, whose inner nodes are either
P-nodes or Q-nodes (note that these P-nodes have nothing to do with
the P-nodes of the SPQR-tree).  The order of edges around a P-node can
be chosen arbitrarily, the edges around a Q-node are fixed up to a
flip.  In this way, a PQ-tree represents a set of orders on its
leaves.  A rooted PQ-tree represents linear orders, an unrooted
PQ-tree represents cyclic orders (in most cases we consider unrooted
PQ-trees).  Given a PQ-tree $T$ and a subset $S$ of its leaves, there
exists another PQ-tree $T'$ representing exactly the orders
represented by $T$ where the elements in $S$ are consecutive.  The
tree $T'$ is the \emph{reduction} of $T$ with respect to $S$.  The
\emph{projection} of $T$ to $S$ is a PQ-tree with leaves $S$
representing exactly the orders on $S$ that are represented by $T$.

The problem \textsc{Simultaneous PQ-Ordering} has several PQ-trees as
input that are related by identifying some of their
leaves~\cite{br-spoacep-13}.  More precisely, every instance is a
directed acyclic graph, where each node is a PQ-tree, and each arc
$(T, T')$ has the property that there is an injective map from the
leaves of the \emph{child} $T'$ to the leaves of the \emph{parent}
$T$.  For each PQ-tree in such an instance, one wants to find an order
of its leaves such that for every arc $(T, T')$ the order chosen for
the parent $T$ is an extension of the order chosen for the child $T'$
(with respect to the injective map).  We will later use instances of
\textsc{Simultaneous PQ-Ordering} to express relations between orderings
of edges around vertices.

\section{Preprocessing Algorithms}
\label{sec:prepr-algor}

In this section, we present several algorithms that can be used as a
preprocessing of a given \textsc{Sefe} instance.  The result is usually a
set of \textsc{Sefe} instances that admit a solution if and only if the
original instance admits one.  The running time of the preprocessing
algorithms is linear, and so is the total size of the equivalent set
of \textsc{Sefe} instances.  Each of the preprocessing algorithms removes
certain types of structures form the instance, in particular from the
common graph.  Namely, we show that we can eliminate union
cutvertices, simultaneous cutvertices with common-degree~3, and
connected components of $G$ that are biconnected but not a cycle.
None of these algorithms introduces new cutvertices in $G$ or
increases the degree of a vertex.  Thus, the preprocessing
algo\-rithms can be successively applied to a given instance, removing
all the claimed structures.

Let~$(\1G,\2G)$ be a \textsc{Sefe} instance with common graph~$G=\1G \cap
\2G$.  We can equivalently encode such an instance in terms of its
\emph{union graph}~$G^\cup = \1G \cup \2G$, whose edges are
labeled~$\{1\}$,~$\{2\}$, or~$\{1,2\}$, depending on whether they are
contained exclusively in~$\1G$, exclusively in~$\2G$, or in~$G$,
respectively.  Any graph with such an edge coloring can be considered
as a \textsc{Sefe} instance.  Since sometimes the coloring version is
more convenient, we use these notions interchangeably throughout this
section.

\subsection{Union Cutvertices}
\label{sec:union-cutv}

Recall that a union cutvertex of a \textsc{Sefe} instance~$(\1G,\2G)$ is
a cutvertex of the union graph~$G^\cup$.  The following theorem states
that the \textsc{Sefe} instances corresponding to the split components of
a cutvertex of~$G^\cup$ can be solved independently; see
Figure~\ref{fig:union-cutvertices}.

\begin{figure}
  \centering
  \includegraphics[page=1]{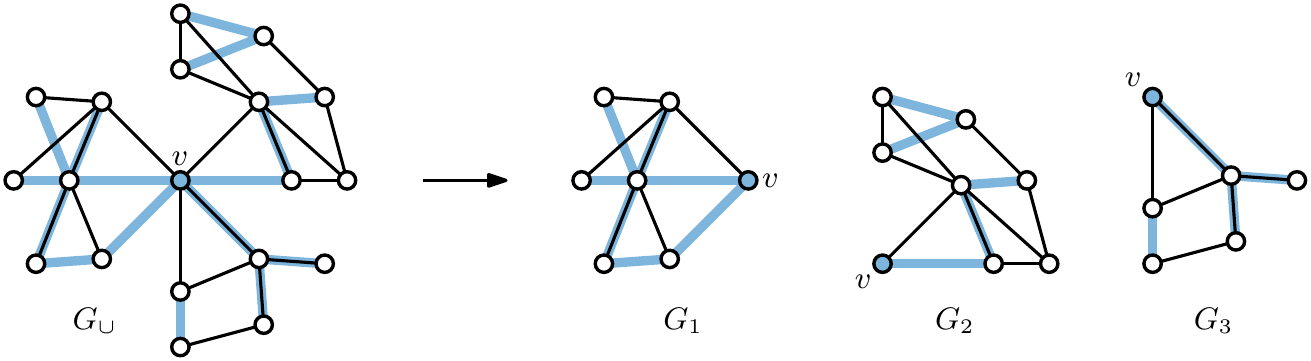}
  \caption{A union cutvertex separates a \textsc{Sefe} instance into
    independent subinstances.}
  \label{fig:union-cutvertices}
\end{figure}

\begin{lemma}
  \label{lem:union-cut}
  Let~$G^\cup$ be a \textsc{Sefe} instance and let~$v$ be a cutvertex
  of~$G^\cup$ with split components~$G_1, \dots, G_k$.  Then~$G^\cup$
  admits a \textsc{Sefe} if and only if~$G_i$ admits a \textsc{Sefe}
  for~$i=1,\dots,k$.
\end{lemma}
\begin{proof}
  Clearly, a \textsc{Sefe} of~$G^\cup$ contains a \textsc{Sefe}
  of~$G_1,\dots,G_k$.  Conversely, given a \textsc{Sefe}~$\E_i$ of~$G_i$
  for~$i=1,\dots,k$, we can assume without loss of generality that~$v$
  is incident to the outer face in each of the~$\E_i$.  Then these
  embeddings can be merged to a \textsc{Sefe}~$\E$ of~$G^\cup$.
\end{proof}

Due to Lemma~\ref{lem:union-cut}, it suffices to consider the blocks
of~$G^\cup$ of a \textsc{Sefe} instance independently.  Clearly, the
blocks can be computed in~$O(n)$ time, and, given a \textsc{Sefe} for
each block, a \textsc{Sefe} of the original instance can be computed
in~$O(n)$ time.

\begin{theorem}
  \label{thm:no-union-cutvertices}
  There is a linear-time algorithm that decomposes a \textsc{Sefe}
  instance into an equivalent set of \textsc{Sefe} instances that do not
  contain union cutvertices.
\end{theorem}

\subsection{Union Separating Pairs}
\label{sec:union-separ-pairs}

In analogy to a union cutvertex, we can define a \emph{union
  separating pair} to be a separating pair of the union graph
$G^\cup$.  It is tempting to proceed as for the union cutvertices:
separate $G^\cup$ according to a union separating pair, solve the
subinstances corresponding to the resulting subgraphs, and merge the
partial solutions.

However, this approach fails as merging the partial solutions may be
impossible; see Figure~\ref{fig:union-sep-pair}a.  Note that it is
easy to merge the partial solutions if all of them have $u$ and $v$ on
the outer face of their union graph.  One can enforce this kind of
behaviour by connecting $u$ and $v$ with a common edge in each
subinstance.  Unfortunately, this is too restrictive as the
subinstances may fail to have a \textsc{Sefe} with this additional
edge whereas the original instance has a solution; see
Figure~\ref{fig:union-sep-pair}b.

\begin{figure}
  \centering
  \includegraphics[page=1]{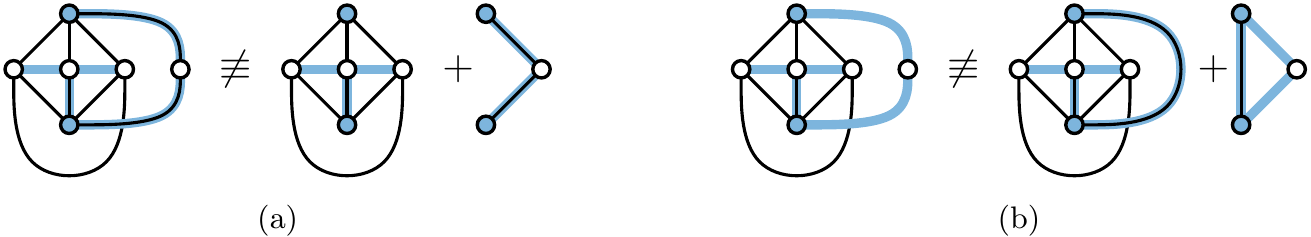}
  \caption{(a)~The original instance does not admit a \textsc{Sefe}
    but the split components with respect to the separating pair $\{u,
    v\}$ (marked vertices) do.  (b)~The original instance admits a
    \textsc{Sefe} but one of the split components does not when adding
    the common edge $uv$.}
  \label{fig:union-sep-pair}
\end{figure}

We can, however, use the idea of adding the common edge $uv$ in every
subinstance to get rid of most union separating pairs.  Throughout the
whole section, we restrict our considerations to the case that $u$ and
$v$ are vertices of the same block $B$ of the common graph and that
$\{u, v\}$ is a separating pair in $B$.  If $\{u, v\}$ separates $B$
into three or more split components, then $u$ and $v$ are poles of a
P-node of $\mathcal T(B)$.  The case when there are only two split
components is a somewhat special (less interesting) case.  To achieve
a more concise notation, we thus assume in the following that $u$ and
$v$ are the poles of a P-node.  However, all arguments extend to the
special case with two split components.

Let $\mu$ be the P-node of $\mathcal T(B)$ with poles $u$ and $v$.
Two virtual edges $\eps_1$ and $\eps_2$ of $\skel(\mu)$ are
\emph{linked in $\1G$} if $\1G$ contains a path from an inner vertex
in $\expan(\eps_1)$ to an inner vertex in $\expan(\eps_2)$ that is
disjoint from $B$ (except for the end vertices of the path).  The
\emph{\circled{1}-link graph} $\1{L_\mu}$ of $\mu$ has the virtual
edges of $\mu$ as nodes, with an edge between two nodes if and only if
the corresponding virtual edges are linked in $\1G$.  Analogously, we
can define the \emph{\circled{2}-link graph} $\2{L_\mu}$ and the
\emph{union-link graph} $L_\mu^\cup$.

\begin{figure}
  \centering
  \includegraphics[page=1]{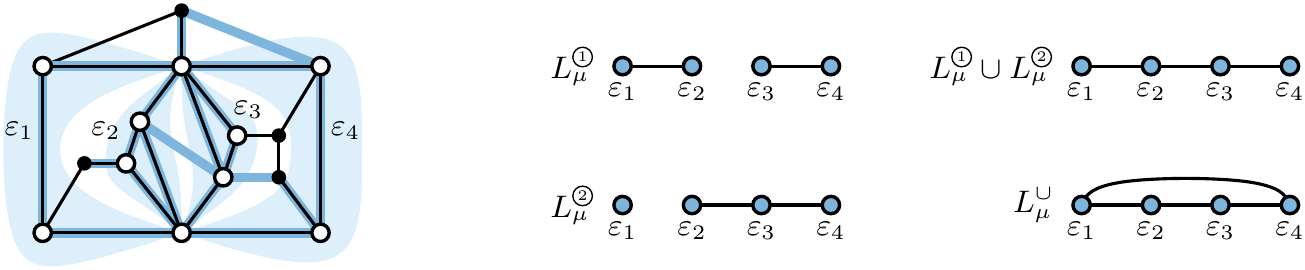}
  \caption{A P-node $\mu$ of the union graph with four virtual edges
    $\eps_1, \dots, \eps_4$ together with the link graphs $\1{L_\mu}$,
    $\2{L_\mu}$, $\1{L_\mu} \cup \2{L_\mu}$, and $L_\mu^\cup$.}
  \label{fig:link-graph}
\end{figure}

Note that the $\1{L_\mu}$ and $\2{L_\mu}$ are subgraphs of
$L_\mu^\cup$.  But $L_\mu^\cup$ is no the union of $\1{L_\mu}$ and
$\2{L_\mu}$, as two virtual edges may be linked in the union graph but
in none of the two exclusive graphs; see Figure~\ref{fig:link-graph}.
However, the union of $\1{L_\mu}$ and $\2{L_\mu}$ will also be of
interest later.  We call it the \emph{exclusive-link graph} and denote
it by $\1{L_\mu} \cup \2{L_\mu}$.  An edge in the exclusive-link graph
indicates that the two corresponding virtual edges are either linked
in $\1G$ or in $\2G$.

We note that the following two lemmas are neither entirely new (e.g.,
Angelini et al.~\cite{adfpr-tsegi-12} use a slightly weaker statement)
nor very surprising.

\begin{lemma}
  \label{lem:link-graph-adjacent}
  Let $L_\mu^\cup$ be a union-link graph of a given \textsc{Sefe}
  instance and let $\eps_1$ and $\eps_2$ be adjacent in $L_\mu^\cup$.
  In every simultaneous embedding, $\eps_1$ and $\eps_2$ are adjacent
  in the embedding of $\skel(\mu)$.
\end{lemma}
\begin{proof}
  First assume that $\eps_1$ and $\eps_2$ are already adjacent in
  $\1{L_\mu}$.  Then the expansion graphs of $\eps_1$ and $\eps_2$
  bound a face in every embedding of $G$ that extends to an embedding
  of $\1G$.  Thus, $\eps_1$ and $\eps_2$ must be adjacent in the
  embedding of $\skel(\mu)$.  The same holds if $\eps_1$ and $\eps_2$
  are adjacent in $\2G$.

  Otherwise, let $B$ be the block whose SPQR-tree contains $\mu$.  Let
  $\pi$ be a path in the union graph connecting inner vertices $v_1$
  and $v_2$ in the expansion graphs of $\eps_1$ and $\eps_2$,
  respectively, that is disjoint from $B$.  Clearly, then common
  vertices of $\pi$ must be embedded into a face $B$ that is incident
  to $v_1$ and to $v_2$.  Such a face only exists if $\eps_1$ and
  $\eps_2$ are adjacent in the embedding of $\skel(\mu)$.
\end{proof}

\begin{lemma}
  \label{lem:link-graph-cycle-or-paths}
  If $(\1G, \2G)$ admits a \textsc{Sefe}, then each union-link graph
  is either a cycle or a collection of paths.
\end{lemma}
\begin{proof}
  Let $B$ be a block of $G$ and let $\mu$ be a P-node of $\mathcal
  T(B)$.  Let $\skel(\mu)$ be embedded according to a simultaneous
  embedding of $(\1G, \2G)$ .  Let $\eps_1, \dots, \eps_k$ be the
  virtual edges of $\skel(\mu)$ embedded in this order.  Due to
  Lemma~\ref{lem:link-graph-adjacent}, two virtual edges $\eps_i$ and
  $\eps_j$ can be adjacent in $L_\mu^\cup$ only if $i + 1 = j$ or $i =
  k$ and $j = 1$.  Thus, $L_\mu^\cup$ is a subgraph of the cycle
  $\eps_1, \dots, \eps_k, \eps_1$.  Hence, $L_\mu^\cup$ is either a
  cycle or a collection of paths.
\end{proof}

Assume the union-link graph $L_\mu^\cup$ of a P-node $\mu$ is
connected (i.e., by Lemma~\ref{lem:link-graph-cycle-or-paths} a cycle
or a path containing all virtual edges).  Then
Lemma~\ref{lem:link-graph-adjacent} implies that the virtual edges in
$\skel(\mu)$ have to be embedded in a fixed order up to reversal.  In
this case, it remains to choose between two different embeddings,
although the $k$ virtual edges of $\skel(\mu)$ have $(k-1)!$ different
cyclic orders.  In the following we show that we can assume without
loss of generality that every union-link graph is connected.

Assume $L_\mu^\cup$ is not connected.  Then the poles $u$ and $v$ of
$\mu$ are a separating pair in the union graph.  Moreover, the
expansion graphs of two virtual edges from different connected
components of $L_\mu^\cup$ end up in different split components with
respect to $u$ and $v$.  Thus, we get at least two split components
with a common path from $u$ to $v$.  If this is the case, we say that
the separating pair $\{u, v\}$ \emph{separates a common cycle}.  We
obtain the following lemma; see
Figure~\ref{fig:union-sep-pair-sep-cycle}.

\begin{figure}
  \centering
  \includegraphics[page=2]{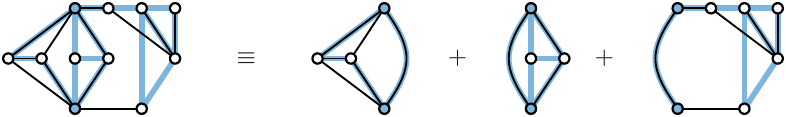}
  \caption{A union separating pair that separates a common cycle can
    be used to decompose the instances into simpler parts.}
  \label{fig:union-sep-pair-sep-cycle}
\end{figure}

\begin{lemma}
  \label{lem:union-separating-pair-separating-a-cycle}
  Let $\{u, v\}$ be a separating pair of the union graph $G^\cup$ that
  separates a common cycle and let $G_1^\cup, \dots, G_k^\cup$ be the
  split components.  Then $G^\cup$ admits a \textsc{Sefe} if and only
  if $G_i^\cup$ with the additional common edge $uv$ admits a
  \textsc{Sefe} for $i = 1, \dots, k$.
\end{lemma}
\begin{proof}
  Assume we have a solution for each subinstances $G_i^\cup + uv$.  As
  $uv$ is a common edge, we can assume without loss of generality that
  it lies in the boundary of the outer face.  It is thus easy to
  obtain a drawing of $G^\cup$ from these partial solutions without
  introducing any new crossings.  Thus, this yields a \textsc{Sefe} of
  $G^\cup$.

  Conversely, assume $G^\cup$ admits a \textsc{Sefe}.  As $\{u, v\}$
  separates a common cycle, we can assume that $G_1$ and $G_2$ both
  contain a path of common edges connecting $u$ and $v$.  We have to
  show that $G_i + uv$ admits a \textsc{Sefe} for every $i = 1, \dots,
  k$.  Assume that $i \not= 1$.  Let $\pi$ be the path of common edges
  connecting $u$ and $v$ in $G_1$.  The graph $G_i + \pi$ (which is a
  subgraph of $G^\cup$) admits a \textsc{Sefe} as the property of
  admitting a \textsc{Sefe} is closed under taking subgraphs.
  Moreover, it is also closed under contracting common edges.  Thus,
  we can assume that $\pi$ is actually the common edge $uv$.  This
  yields a \textsc{Sefe} of $G_i + uv$.  For $i = 1$ we can use the
  common path connecting $u$ and $v$ in $G_2$ instead.
\end{proof}

As argued above, a disconnected union-link graph implies the existence
of a separating pair that separates a common cycle.  We thus obtain
the following theorem.

\begin{theorem}
  \label{thm:all-link-graphs-are-connected}
  There is a linear-time algorithm that decomposes a \textsc{Sefe}
  instance into an equivalent set of \textsc{Sefe} instances of total
  linear size in which all union-link graphs are connected.
\end{theorem}
\begin{proof}
  Clearly, applying the decomposition implied by
  Lemma~\ref{lem:union-separating-pair-separating-a-cycle}
  exhaustively results in a set of instances of total linear size.  It
  remains to show that we can apply all decomposition steps in total
  linear time.  To this end, consider the SPQR-tree $\mathcal T$ of
  the union graph $G^\cup$.  Note that $G^\cup$ is non-planar in
  general and thus the R-nodes skeletons of $\mathcal T$ may be
  non-planar.  Nonetheless, $\mathcal T$ can be computed in linear
  time~\cite{gm-lti-00} and represents all separating pairs
  of~$G^\cup$.

  Let $\mu$ be an inner node of $\mathcal T$ and let $\eps = uv$ be a
  virtual edge in $\skel(\mu)$.  We say that $\eps$ is a \emph{common
    virtual edge} if the expansion graph of $\eps$ includes a common
  $uv$-path from.  Note that $\{u, v\}$ is a separating pair of
  $G^\cup$.  Moreover, if we know for each virtual edge whether it is
  a common virtual edge, we can determine whether $\{u, v\}$ separates
  a common cycle by only looking at $\skel(\mu)$.  More precisely, if
  $\mu$ is a P-node, then $\{u, v\}$ separates a common cycle if and
  only if two or more virtual edges are common virtual edges.  For S-
  and R-nodes, $\{u, v\}$ separates a common cycle if and only if the
  virtual edge $\eps$ is a common virtual edge and $\skel(\mu) - \eps$
  includes a path of common virtual edges from $u$ to~$v$.

  Let us assume, we know for each virtual edge, whether it is a common
  virtual edge.  Then we can easily compute the decomposition by
  rooting $\mathcal T$ and processing it bottom up.  Thus, it remains
  to compute the common virtual edges in linear time.  To this end,
  first root $\mathcal T$ at a Q-node.  By processing $\mathcal T$
  bottom up, one can easily compute for each virtual edge, except for
  the parent edges, whether it is a common virtual edge or not.

  It remains to deal with the parent edges.  We process $\mathcal T$
  top down.  When processing a node $\mu$, we assume that we know the
  common virtual edges of $\skel(\mu)$ (potentially including the
  parent edge).  We then compute in $O(|\skel(\mu)|)$ time for which
  children of $\mu$, the parent edge is a common virtual edge.  If
  $\mu$ is the root (i.e., a Q-node), then the only child of $\mu$ has
  a common virtual edge as parent edge if and only if the edge
  corresponding to the Q-node $\mu$ is a common edge.

  Let $\mu$ be a P-node and let $\eps$ be a virtual edge in
  $\skel(\mu)$.  Then $\twin(\eps)$ (which is the parent edge of the
  child corresponding to $\eps$) is a common virtual edge if and only
  if $\skel(\mu)$ includes a common virtual edge different from
  $\eps$.  Thus, $\mu$ can be processed in $O(|\skel(\mu)|)$ time.  If
  $\mu$ is an S-node, it similarly holds that $\twin(\eps)$ is a
  common virtual edge if and only if all virtual edges of $\skel(\mu)$
  except maybe $\eps$ are common virtual edges.

  Finally, if $\mu$ is an R-node, consider the graph $\skel'(\mu)$
  obtained from $\skel(\mu)$ by deleting all non-common virtual edges.
  Let $\eps$ be an arbitrary virtual edge of $\skel(\mu)$.  If $\eps$
  is non-common, then $\twin(\eps)$ is a common virtual edge if and
  only if the end vertices of $\eps$ lie in the same connected
  component of $\skel'(\mu)$.  If $\eps$ is a common virtual edge,
  then $\twin(\eps)$ is a common virtual edge if and only if $\eps$ is
  not a bridge in $\skel'(\mu)$.  Note that both of these properties
  can be checked in constant time for each virtual edge of
  $\skel(\mu)$ after $O(|\skel(\mu)|)$ preprocessing time.  Thus, we
  can also process R-nodes in $O(|\skel(\mu)|)$ time, which yields an
  overall linear running time.
\end{proof}

Let $B$ be a block of the common graph and let $\mu$ be a P-node of
$\mathcal T(B)$.  By Theorem~\ref{thm:all-link-graphs-are-connected},
we can assume that the union-link graph $L_\mu^\cup$ is connected.
Thus, the ordering of the virtual edges in $\skel(\mu)$ is fixed up to
reversal.  Hence, the embedding choices for $\mu$ are the same as
those for an R-node.

In the following, we provide further simplifications by eliminating
some types of simultaneous separating pairs.  Let $u$ and $v$ be the
poles of the P-node $\mu$.  Consider the case that $\{u, v\}$ is a
separating pair in the union graph $G^\cup$ with split components
$G_1^\cup, \dots, G_k^\cup$ (we can assume by
Theorem~\ref{thm:no-union-cutvertices} that neither $u$ nor $v$ is a
cutvertex in $G^\cup$).  As before, we denote the common graph and the
exclusive graphs corresponding to the \textsc{Sefe} instances
$G_i^\cup$ (for $i = 1, \dots, k$) by $G_i$, $\1{G_i}$ and $\2{G_i}$,
respectively.

We define $G_i^\cup$ to be \emph{common connected} if $u$ and $v$ are
connected by a path in $G_i$; see
Figure~\ref{fig:union-split-comp-conn-types}a.  The split component
$G_i^\cup$ is \emph{exclusive connected}, if it is not common
connected but $u$ and $v$ are connected by exclusive paths in both
graphs $\1{G_i}$ and $\2{G_i}$; see
Figure~\ref{fig:union-split-comp-conn-types}b.  It is
\emph{$\circled{1}$-connected}, if $u$ and $v$ are connected by a path
in $\1{G_i}$ but not in $\2{G_i}$; see
Figure~\ref{fig:union-split-comp-conn-types}c.  The term
\emph{$\circled{2}$-connected} is defined analogously; see
Figure~\ref{fig:union-split-comp-conn-types}d.  Note that being
$\circled{1}$- or $\circled{2}$-connected excludes being common or
exclusive connected.  Finally, if $G_i^\cup$ is neither of the above,
it is \emph{union connected}; see
Figure~\ref{fig:union-split-comp-conn-types}e.

We say that $\mu$ is an \emph{impossible P-node} if $\1{L_\mu}$ is a
cycle and one of the split components is $\circled{1}$-connected, if
$\2{L_\mu}$ is a cycle and one of the split components is
$\circled{2}$-connected, or if $\1{L_\mu} \cup \2{L_\mu}$ is a cycle
and one of the split components is exclusive connected.


\begin{figure}
  \centering
  \includegraphics[page=3]{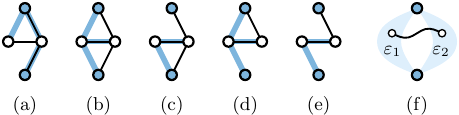}
  \caption{(a--e)~A split component that is common connected,
    exclusive connected, $\circled{1}$-connected,
    $\circled{2}$-connected, and union connected, respectively.
    (f)~The face between two virtual edges that are
    $\circled{1}$-linked (i.e., connected in $\1{L_\mu}$).}
  \label{fig:union-split-comp-conn-types}
\end{figure}

\begin{lemma}
  \label{lem:no-impossible-P-nodes}
  A \textsc{Sefe} instance with an impossible P-node is a no-instance.
\end{lemma}
\begin{proof}
  Let $G_1^\cup, \dots, G_k^\cup$ be the split components with respect
  to the poles $u$ and $v$ of the P-node $\mu$.  As $\mu$ is an
  impossible P-node, the union-link graph $L_\mu^\cup$ is a cycle.
  Thus, at most one split component can be common connected.  As $u$
  and $v$ are the poles of a P-node of the common graph, one of the
  split components must be common connected.  Thus, exactly one split
  component, without loss of generality $G_1^\cup$, is common
  connected.

  First assume that $\mu$ is an impossible P-node due to the fact that
  $\1{L_\mu}$ is a cycle and one of the split components, without loss
  of generality $G_2^\cup$ is $\circled{1}$-connected.  

  Assume the given \textsc{Sefe} instance $G^\cup$ admits a
  \textsc{Sefe} and assume $G_1^\cup$ and $G_2^\cup$ are embedded
  according to this \textsc{Sefe}.  As $G_2^\cup$ is
  $\circled{1}$-connected
  (Figure~\ref{fig:union-split-comp-conn-types}c), the graph $\1{G_2}$
  includes a path $\1\pi$ from $u$ to $v$.  Clearly, $\1\pi$ lies in a
  single face of $\1{G_1}$.  Let $f$ be the corresponding face of the
  common graph $G_1$.  The boundary of $f$ belongs to the expansion
  graphs of two different virtual edges $\eps_1$ and $\eps_2$ of
  $\skel(\mu)$; see Figure~\ref{fig:union-split-comp-conn-types}f.
  However, $\eps_1$ and $\eps_2$ cannot be $\circled{1}$-linked (as in
  Figure~\ref{fig:union-split-comp-conn-types}f), as otherwise $\1\pi$
  could not be embedded into the face $f$ without having a crossing in
  $\1G$.  It follows that $\1{L_\mu}$ cannot be a cycle, a
  contradiction.

  Analogously, if $\2{L_\mu}$ is a cycle and one of the split
  components is $\circled{2}$-connected, we find a path $\2\pi$ that
  is a witness for a pair of adjacent virtual edges that are not
  $\circled{2}$-linked.  It remains to consider the case where
  $\1{L_\mu} \cup \2{L_\mu}$ is a cycle and $G_2^\cup$ is exclusive
  connected.  In this case, $\1{G_2}$ and $\2{G_2}$ include paths
  $\1\pi$ and $\2\pi$, respectively, connecting $u$ and $v$.  As they
  both belong to the same split component, they have to be embedded in
  the same common face of $G_1$.  Thus, there are adjacent virtual
  edges that are neither $\circled{1}$- nor $\circled{2}$-linked.
  Hence, $\1{L_\mu} \cup \2{L_\mu}$ is not a circle.
\end{proof}

Due to this lemma, it is sufficient to consider the case that $\mu$ is
not an impossible P-node.  We want to show that the different split
components (in the union graph, with respect to the poles $u$ and $v$
of $\mu$) can be handled independently.  However, we have to exclude a
special case to make this true.  Let $G_i^\cup$ be one of the split
components that is exclusive connected.  We say that $G_i^\cup$ has
\emph{common ends} if it contains a common edge incident to $u$ or to
$v$.  Figure~\ref{fig:union-sep-pair-common-ends} shows an example,
where the following lemma does not hold without excluding exclusive
connected components with common ends.

\begin{figure}
  \centering
  \includegraphics[page=4]{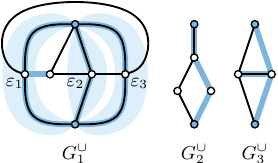}
  \caption{The union split component $G_1^\cup$ includes the expansion
    graphs of all three virtual edges $\eps_1$, $\eps_2$, and
    $\eps_3$.  The edge pairs $\eps_1, \eps_3$ and $\eps_2, \eps_3$
    are $\circled{1}$-linked, thus the exclusive connected split
    component $G_2^\cup$ cannot be embedded into the faces between
    $\eps_1$ and $\eps_3$ or between $\eps_2$ and $\eps_3$.  Although
    $\eps_1$ and $\eps_2$ are neither $\circled{1}$- nor
    $\circled{2}$-linked, $G_2^\cup$ cannot be embedded into the face
    between $\eps_1$ and $\eps_2$ due to its common end.  The
    component $G_3^\cup$ has no common end and can be embedded into
    the face between $\eps_1$ and $\eps_2$.}
  \label{fig:union-sep-pair-common-ends}
\end{figure}

\begin{lemma}
  \label{lem:union-sep-pair-no-common-ends}
  Let $G^\cup$ be a \textsc{Sefe} instance and let $\mu$ be a
  non-impossible P-node whose poles are a separating pair with split
  components $G_1^\cup, \dots, G_k^\cup$.  Assume $G_1^\cup$ is the
  only common connected split component and none of the exclusive
  connected components has common ends.  Then $G^\cup$ admits a
  \textsc{Sefe} if and only if $G_1^\cup$ admits a \textsc{Sefe} and
  $G_i^\cup$ together with the common edge $uv$ admits a \textsc{Sefe}
  for $i = 2, \dots, k$.
\end{lemma}
\begin{proof}
  Assume that $G_1^\cup$ and $G_i^\cup + uv$ (for $i = 2, \dots, k$)
  admit simultaneous embeddings.  We show how to combine the
  simultaneous embeddings of $G_1^\cup$ and $G_2^\cup + uv$ to a
  simultaneous embedding of $G_1^\cup \cup G_2^\cup$.  The procedure
  can then be iteratively applied to the other split components.  We
  have to distinguish the cases that $G_2^\cup$ is union connected,
  $\circled{1}$-connected, $\circled{2}$-connected, and exclusive
  connected (without common ends).

  First assume that $G_2^\cup$ is union connected.
  Figure~\ref{fig:union-sep-pair-example-combination} shows an example
  illustrating the proof for this case.  As $u$ and $v$ are the poles
  of a P-node, the common graph $G_1$ has a face $f$ that is incident
  to $u$ and to $v$.  Let further $\1{f_u}$ and $\1{f_v}$ be faces of
  $\1{G_1}$ incident to $u$ and $v$, respectively, that are both part
  of the union face $f$.  Similarly, we choose faces $\2{f_u}$ and
  $\2{f_v}$ in $\2{G_2}$ that are incident to $u$ and $v$,
  respectively, and that are both part of $f$.  Note that $u$ might
  have several incidences to the face $f$, i.e., when $u$ is a
  cutvertex in $G_1$ and one of the corresponding blocks is embedded
  into $f$.  In this case, we choose $\2{f_u}$ such that it has
  \emph{the same incidence} to $u$ as $\1{f_u}$, i.e., the common
  edges appearing in the cyclic order around $u$ before and after
  $\2{f_u}$ are the same as those that appear before and after
  $\2{f_u}$.  We ensure the same for $\1{f_v}$ and $\2{f_v}$.  In the
  example in Figure~\ref{fig:union-sep-pair-example-combination}, $f$
  has two incidences to $u$ and two incidences to $v$ and the chosen
  incidence is marked by an angle.

  \begin{figure}
    \centering
    \includegraphics[page=5]{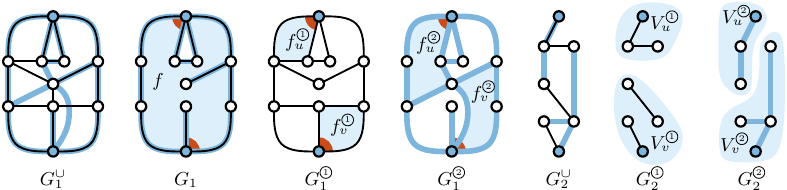}
    \caption{Two split components of the union graph illustrating the
      proof of Lemma~\ref{lem:union-sep-pair-no-common-ends}.}
    \label{fig:union-sep-pair-example-combination}
  \end{figure}

  Due to the common edge $uv$, we can assume that the \textsc{Sefe} of
  $G_2^\cup$ has $uv$ and $uv$ on the outer face.  As $G_2^\cup$ is not
  $\circled{1}$-connected, we can separate the vertices of $\1{G_2}$
  into two subsets $\1{V_u}$ and $\1{V_v}$, such that $\1{V_u}$
  contains all vertices of the connected component of $\1{G_2}$
  containing $u$, while $\1{V_v}$ contains all other vertices.  We can
  then embed the vertices of $\1{V_u}$ into $\1{f_u}$ and the vertices
  of $\1{V_v}$ into $\1{f_v}$ without changing the embedding of
  $\1{G_2}$.  In the same way, $\2{G_2}$ can be embedded into
  $\2{f_u}$ and $\2{f_v}$.

  As we did not change the embedding of $G_1^\cup$ or $G_2^\cup$, the
  edge orderings are consistent for all vertices except maybe $u$ and
  $v$.  Moreover, the relative positions between connected components
  in $G_1^\cup$ is consistent and the same holds for $G_2^\cup$.  As
  the four faces $\1{f_u}$, $\1{f_v}$, $\2{f_u}$, and $\2{f_v}$ belong
  to the same common face $f$, the relative positions of components in
  $G_2$ with respect to components in $G_1$ are also consistent.
  Moreover, all components of $G_1$ lie in the outer face of $G_2$
  with respect to $\1{G_2}$ and $\2{G_2}$.  Finally, the edge ordering
  at $u$ is consistent, as all edges incident to $u$ in $\1{G_2}$ and
  $\2{G_2}$ are embedded between the same pair of common edges in
  $G_1$.  As the same holds for $v$, we obtain a simultaneous
  embedding of $G_1^\cup \cup G_2^\cup$.

  If $G_2^\cup$ is $\circled{1}$-connected, we know that $\1{L_\mu}$
  is not a circle (otherwise, $\mu$ would be impossible).  Thus, we
  can choose the faces $f$, $\1{f_u}$, and $\1{f_v}$ such that
  $\1{f_u} = \1{f_v}$.  Then we can embed $\1{G_2}$ into this face
  without separating it.  All remaining arguments work the same as
  above.  The case that $G_2^\cup$ is $\circled{1}$-connected is
  symmetric.

  Finally, if $G_2^\cup$ is exclusive connected, there is a pair of
  virtual edges that are neither $\circled{1}$- nor
  $\circled{2}$-linked.  Thus, we can choose the common face $f$ and
  the faces $\1{f_u}$, $\1{f_v}$, $\2{f_u}$, and $\2{f_v}$ belonging
  to $f$ such that $\1{f_u} = \1{f_v} = \1f$ and $\2{f_u} = \2{f_v} =
  \2f$.  Unfortunately, we cannot always ensure that $\1f$ and $\2f$
  have the same incidence to $u$ or $v$; see
  Figure~\ref{fig:union-sep-pair-common-ends}.  However, the arguments
  form the previous cases still ensure that all relative positions and
  all cyclic orders except for maybe at $u$ and $v$ are consistent.
  As $G_2^\cup$ has no common ends, all common edges incident to $u$
  and $v$ are contained in $G_1$ and thus the cyclic orders around
  these vertices are also consistent.

  Note that combining the simultaneous embeddings $G_1^\cup$ and
  $G_2^\cup$ in this way (for all four cases) maintains the properties
  that there are faces in $G_1$, $\1{G_1}$, or $\2{G_2}$ that are
  incident to both poles $u$ and $v$.  Thus, we can continue adding
  embeddings of all remaining subinstances $G_3^\cup, \dots, G_k^\cup$
  in the same way.
\end{proof}

Assume we exhaustively applied
Lemma~\ref{lem:union-separating-pair-separating-a-cycle},
Lemma~\ref{lem:no-impossible-P-nodes}, and
Lemma~\ref{lem:union-sep-pair-no-common-ends} to a given instance of
\textsc{Sefe} and let $(\1G, \2G)$ be the resulting instance.  Let $u$
and $v$ be the poles of a P-node $\mu$ of the common graph such that
$\{u, v\}$ are a separating pair in the union graph.  By
Lemma~\ref{lem:union-separating-pair-separating-a-cycle} we can assume
that $\{u, v\}$ does not separate a common cycle.  Thus, exactly one
split component has a common $uv$-path.  By
Lemma~\ref{lem:no-impossible-P-nodes}, we can assume that $\mu$ is a
non-impossible P-node.  Thus, we could apply
Lemma~\ref{lem:union-sep-pair-no-common-ends} if there were split
components without common ends.  Hence we obtain the following theorem.

\begin{theorem}
  \label{thm:only-special-union-separating-pairs}
  Let $(\1G, \2G)$ be an instance of \textsc{Sefe}.  In linear time,
  we can find equivalent instances such that every union separating
  pair $\{u, v\}$ has one of the following properties.
  \begin{itemize}
  \item The vertices $u$ and $v$ are not the poles of a P-node of a
    common block.
  \item Every split component has a common edge incident to $u$ or to
    $v$ but only one has a common $uv$-path.
  \end{itemize}
\end{theorem}
\begin{proof}
  It remains to prove the claimed running time.  The linear running
  time for decomposing the instances along its union separating pairs
  that separate a common cycle was already shown for
  Theorem~\ref{thm:all-link-graphs-are-connected}.  In
  Section~\ref{sec:simult-embedd-union-bridge-constr} we extend the
  algorithm by Angelini et al.~\cite{adfpr-tsegi-12} for solving
  \textsc{Sefe} if the common graph is biconnected to the case where
  we allow exclusive vertices and have so-called union bridge
  constrains.  It is not hard to see that testing $(\1G, \2G)$ for the
  existence of impossible P-nodes can be done using the linear-time
  algorithm from Section~\ref{sec:simult-embedd-union-bridge-constr}.

  It remains to decompose the union graph $G^\cup$ according to
  separating pairs that separated $G^\cup$ according to
  Lemma~\ref{lem:union-sep-pair-no-common-ends}.  As in the proof of
  Theorem~\ref{thm:all-link-graphs-are-connected}, we consider the
  SPQR-tree $\mathcal T$ of $G^\cup$.  For
  Theorem~\ref{thm:all-link-graphs-are-connected}, we had to compute
  for every virtual edge, whether its expansion graph included a
  common path between its endpoints.  Now, we in addition have to know
  which expansion graphs are exclusive connected and have common ends.
  This can be done analogously to the proof of
  Theorem~\ref{thm:all-link-graphs-are-connected}.
\end{proof}

\subsection{Connected Components that are Biconnected}
\label{sec:conn-comp-that-are-biconn}


Let~$(\1G,\2G)$ be a \textsc{Sefe} instance and let~$C$ be a connected
component of the common graph~$G$ that is a cycle; see
Figure~\ref{fig:2-components-cycle}a.  A \emph{union bridge} of~$\1G$
and~$\2G$ with respect to~$C$ is a connected component of~$G^\cup - C$
together with all its \emph{attachment vertices} on~$C$; see
Figure~\ref{fig:2-components-cycle}b.  Equivalently, the union bridges
are the split components of $G^\cup$ with respect to the vertices of
$C$ excluding the edges of $C$.  Similarly, there are
\emph{$\circled{1}$-bridges} and \emph{$\circled{2}$-bridges}, which
are connected components of~$\1G - C$ and~$\2G -C$ together with their
attachment vertices on~$C$, respectively; see
Figure~\ref{fig:2-components-cycle}c--d.  We say that two bridges~$B_1$
and~$B_2$ \emph{alternate} if there are attachments~$a_1,b_1$ of~$B_1$
and attachments~$a_2,b_2$ of~$B_2$, such that the order along~$C$
is~$a_1a_2b_1b_2$; see Figure~\ref{fig:2-components-cycle}e.  We have
the following lemma, which basically states that we can handle
different union bridges independently

\begin{figure}
  \centering
  \includegraphics[page=1]{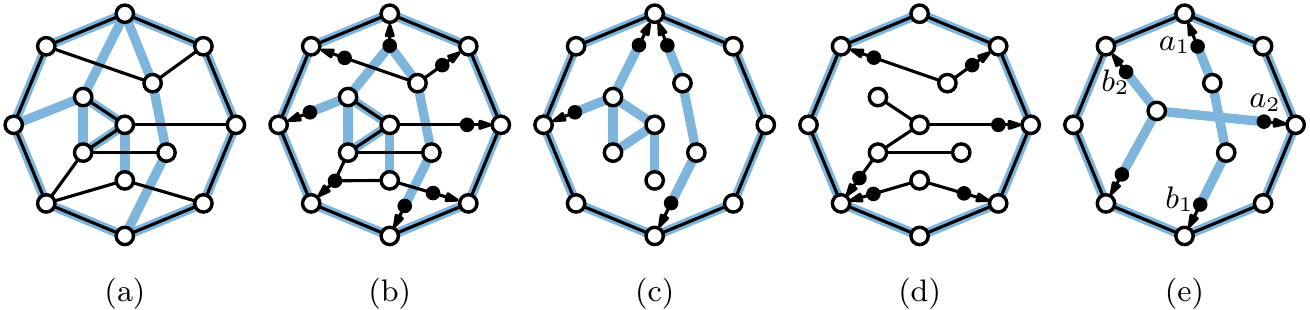}
  \caption{Situation where the connected component $C$ of $G$ is a
    cycle.  (a)~A simultaneous embedding of $(\1G,\2G)$ with $C$ on
    the outer face.  (b)~Removing $C$ yields a single connected
    component in $G^\cup$.  Thus, there is only one union bridge.
    Its attachment vertices are illustrated as black dots.  (c)~The
    two $\circled{1}$-bridges.  (d)~The three $\circled{2}$-bridges.
    Note that different bridges might share attachment vertices.
    (e)~Two alternating $\circled{1}$-bridges.}
  \label{fig:2-components-cycle}
\end{figure}

\begin{lemma}
  \label{lem:bridge-sefe}
  Let~$\1G$ and~$\2G$ be two planar graphs and let~$C$ be a connected
  component of the common graph that is a cycle.  Then the
  graphs~$\1G$ and~$\2G$ admit a \textsc{Sefe} where~$C$ is the boundary
  of the outer face if and only if
  \begin{inparaenum}[(i)]
    \item each union bridge admits a \textsc{Sefe} together with~$C$ and
    \item no two~$\circled{i}$-bridges of~$C$ alternate for~$i=1,2$.
  \end{inparaenum}
\end{lemma}
\begin{proof}
  Clearly the conditions are necessary; we prove sufficiency.
  Let~$B_1,\dots,B_k$ be the union bridges with respect to~$C$, and
  let~$(\1{\E_1},\2{\E_1}), \dots, (\1{\E_k},\2{\E_k})$ be the
  corresponding simultaneous embeddings of~$B_i$ together with~$C$,
  which exist by condition~(i).  Note that each union bridge is
  connected, and hence all its edges and vertices are embedded on the
  same side of~$C$.  After possibly flipping some of the embeddings,
  we may assume that each of them has~$C$ with the same clockwise
  orientation as the outer face.

  We now glue~$\1{\E_1},\dots,\1{\E_k}$ to an embedding~$\1\E$
  of~$\1{G}$, which is possible by condition~(ii).  In the same way,
  we find an embedding~$\2\E$ of~$\2G$ from~$\2{\E_1},\dots,\2{\E_k}$.
  We claim that~$(\1\E,\2\E)$ is a \textsc{Sefe} of~$\1G$ and~$\2G$.
  For the consistent edge orderings, observe that any common
  vertex~$v$ with common-degree at least~3 is contained, together with
  all neighbors, in some union bridge~$B_i$.  The compatibility of the
  edge ordering follows since~$(\1{\E_i},\2{\E_i})$ is a \textsc{Sefe}.
  Concerning the relative position of a vertex~$v$ and some common
  cycle~$C'$, we note that the relative positions clearly coincide
  in~$\1\E$ and~$\2\E$ for~$C' = C$.  Otherwise~$C'$ is contained in
  some union bridge.  If~$v$ is embedded in the interior of~$C'$ in
  one of the two embeddings, then it is contained in the same union
  bridge as~$C'$, and the compatibility follows.  If this case does
  not apply, it is embedded outside of~$C'$ in both embeddings, which
  is compatible as well.
\end{proof}

We note that this approach fails, when the cycle $C$ is not a
connected component of $G$, i.e., when a union bridge contains common
edges incident to an attachment vertex.  The reason is that the order
of common edges incident to this attachment vertex is chosen in the
moment one reinserts the union bridges into $C$.

Now consider a connected component~$C$ of the common graph~$G$ of a
\textsc{Sefe} instance such that~$C$ is biconnected.  Such a component is
called \emph{2-component}.  If $C$ is a cycle, it is a \emph{trivial
  2-component}.  We define the union bridges, and the ~$\circled{1}$-
and~$\circled{2}$-bridges of~$\1G$ and~$\2G$ with respect to~$C$ as
above.  We call an embedding~$\mathcal E$ of $C$ together with an
assignment of the union bridges to its faces \emph{admissible} if and
only if, (i) for each union bridge, all attachments are incident to
the face to which it is assigned, and (ii) no two~$\circled{1}$- and
not two~$\circled{2}$-bridges that are assigned to the same face
alternate.

In the following, we try to solve the given \textsc{Sefe} instance
$(\1G, \2G)$ by first finding an admissible embedding of the
2-component $C$.  Then we test for every face of $C$ whether all union
bridges can be embedded inside the corresponding facial cycle.  By
Lemma~\ref{lem:bridge-sefe} we know that this is possible if and only
if each union bridge together with the facial cycle admits a
\textsc{Sefe} and no two $\circled{i}$-bridges (for $i = 1, 2$)
alternate.  The latter is ensured by property (ii) of the admissible
embedding of $C$.  The former yields simpler \textsc{Sefe} instances
in which the 2-component $C$ is represented by a simple cycle.  It
remains to show that, if this approach fails, there exists no
\textsc{Sefe}.  First note that the properties (i) and (ii) of an
admissible embedding are clearly necessary.  Thus, if there is no
admissible embedding of $C$, then there is no \textsc{Sefe}.  It
remains to show that it does not depend on the admissible embedding of
$C$ one chooses, whether a union bridge together with the facial cycle
of the face it is assigned to admits a \textsc{Sefe} or not.  In fact,
the following lemma shows that the facial cycle one gets for a union
bridge is more or less independent from the embedding of $C$, i.e.,
the attachment vertices of the bridge always appear in the same order
along this cycle.

\begin{lemma}
  \label{lem:bico-unique-ordering}
  Let~$G$ be a biconnected planar graph and let~$X$ be a set of
  vertices that are incident to a common face in some planar
  embedding of~$G$.  Then the order of $X$ in any simple cycle of
  $G$ containing $X$ is unique up to reversal.
\end{lemma}
\begin{proof}
  Consider a planar embedding~$\E$ of~$G$ where all vertices in~$X$
  share a face, and let~$C_X$ denote the corresponding facial cycle.
  Note that~$C_X$ is simple since~$G$ is biconnected.  Let~$C$ be an
  arbitrary simple cycle in~$G$ containing all vertices in~$X$.
  In~$\E$, all parts of~$C$ that are disjoint from~$C_X$ are embedded
  outside of~$C_X$.  Let~$C_X'$ denote the cycle obtained from~$C_X$
  by contracting all maximal paths whose internal vertices do not
  belong to~$C$ to single edges.  Observe that~$C_X$ and~$C_X'$ visit
  the vertices of~$X$ in the same order.  Consider the graph~$C \cup
  C_X'$, which is clearly outerplanar and biconnected.  Hence both~$C$
  and~$C_X'$ visit the vertices of~$X$ in the same order (up to
  complete reversal).  Since~$C$ was chosen arbitrarily, the claim
  follows.
\end{proof}

For a union bridge~$B$, let~$C_B$ denote the cycle consisting of the
attachments of~$B$ in the ordering of an arbitrary cycle of~$G$
containing all the attachments.  By
Lemma~\ref{lem:bico-unique-ordering}, the cycle $C_B$ is uniquely
defined.  Let further~$G_B$ denote the graph consisting of the union
bridge~$B$ and the cycle~$C_B$ connecting the attachment vertices of
$B$.  We call this graph the \emph{union bridge graph} of the bridge
$B$.  The following lemma formally states our above-mentioned strategy
to decompose a \textsc{Sefe} instance.

\begin{lemma}
  \label{lem:bico-sefe}
  Let~$\1G$ and~$\2G$ be two connected planar graphs and let~$C$ be a
  2-component of the common graph $G$.  Then the graphs~$\1G$
  and~$\2G$ admit a \textsc{Sefe} if and only if
  \begin{inparaenum}[(i)]
  \item $C$ admits an admissible embedding, and
  \item each union bridge graph admits a \textsc{Sefe}.
  \end{inparaenum}
  If a \textsc{Sefe} exists, the embedding of~$C$ can be chosen as an
  arbitrary admissible embedding.
\end{lemma}
\begin{proof}
  Clearly, a \textsc{Sefe} of~$\1G$ and~$\2G$ defines an embedding of~$C$
  and a bridge assignment that is admissible.  Moreover, it induces a
  \textsc{Sefe} of each union bridge graph.  

  Conversely, assume that~$C$ admits an admissible embedding and each
  union bridge graph admits a \textsc{Sefe}.  We obtain a \textsc{Sefe} of~$\1G$
  and~$\2G$ as follows.  Embed~$C$ with the admissible embedding and
  consider a face~$f$ of this embedding with facial cycle~$C_f$.
  Let~$B_1,\dots,B_k$ denote the union bridges that are assigned to
  this face, and let~$(\1{\E_1},\2{\E_1}),\dots,(\1{\E_k},\2{\E_k})$
  be simultaneous embeddings of the bridge graphs~$G_B$.  By
  subdividing the cycle~$C_B$, in each of the embeddings, we may
  assume that the outer face of each~$B_i$ in the
  embedding~$(\1{\E_i},\2{\E_i})$ is the facial cycle~$C_f$ with the
  same orientation in each of them.  By Lemma~\ref{lem:bridge-sefe},
  we can hence combine them to a single \textsc{Sefe} of all union bridges
  whose outer face is the cycle~$C_f$.  We embed this \textsc{Sefe} into the
  face~$f$ of~$C$.  Since we can treat the different faces of~$C$
  independently, applying this step for each face yields a \textsc{Sefe}
  of~$\1G$ and~$\2G$ with the claimed embedding of~$C$.
\end{proof}

Lemma~\ref{lem:bico-sefe} suggests a simple strategy for reducing
\textsc{Sefe} instances containing non-trivial 2-components.  Namely,
take such a component, construct the corresponding union bridge
graphs, where~$C$ occurs only as a cycle, and find an admissible
embedding of~$C$.  Finding an admissible embedding for~$C$ can be done
as follows.  To enforce the non-overlapping attachment property,
replace each~$\circled{1}$-bridge of~$C$ by a
\emph{dummy~$\circled{1}$-bridge} that consists of a single vertex
that is connected to the attachments of that bridge via edges
in~$\1E$.  Similarly, we replace~$\circled{2}$-bridges, which are
connected to attachments via exclusive edges in~$\2E$.  We seek a
\textsc{Sefe} of the resulting instance (where the common graph is
biconnected), additionally requiring that dummy bridges belonging to
the same union bridge are embedded in the same face.  We also refer to
such an instance as \emph{\textsc{Sefe} with union bridge
  constraints}.  A slight modification of the algorithm by Angelini et
al.~\cite{adfpr-tsegi-12} can decide the existence of such an
embedding in polynomial time.  This gives the following lemma.

\begin{lemma}
  \label{lem:find-admissible-embedding}
  Computing an admissible embedding of a 2-component $C$ is equivalent
  to solving \textsc{Sefe} with union bridge constraints on an
  instance having $C$ as common graph.  This can be done in polynomial
  time.
\end{lemma}

It then remains to treat the union bridge graphs.  Exhaustively
applying Lemma~\ref{lem:bico-sefe} (using
Lemma~\ref{lem:find-admissible-embedding} to find admissible
embeddings) results in a set of \textsc{Sefe} instances where each
2-component is trivial.  Note that we could go even further and
decompose along cycles that have more than one union bridge.  However,
this is not necessary to obtain the following theorem.

\begin{theorem}
  \label{thm:equivalence-disj-cycles-poly}
  Given a \textsc{Sefe} instance, an equivalent set of instances of total
  linear size such that each 2-component of these instances is trivial
  can be computed in polynomial time.
\end{theorem}


We can improve the running time in
Theorem~\ref{thm:equivalence-disj-cycles-poly} to linear.  However, it
is quite tedious work, involving a lot of data structures, and results
in a lengthy proof.  To not disturb the reading flow too much, the
proof is deferred to its own section (Section~\ref{sec:prepr-2-comp},
starting on page~\pageref{sec:prepr-2-comp}).  Here, we only sketch it
very roughly.

We do not apply an iterative process, removing one 2-component after
another (as suggested above), but we decompose the whole instance at
once.  For this, we introduce the notion of \emph{subbridges}.  A
subbridge of a graph~$G$ with respect to components~$C_1,\dots,C_k$ is
a maximal connected subgraph of~$G$ that does not become disconnected
by removing all vertices of one component~$C_i$; see
Figure~\ref{fig:subbridges}.

\begin{figure}
  \centering
  \includegraphics[page=1]{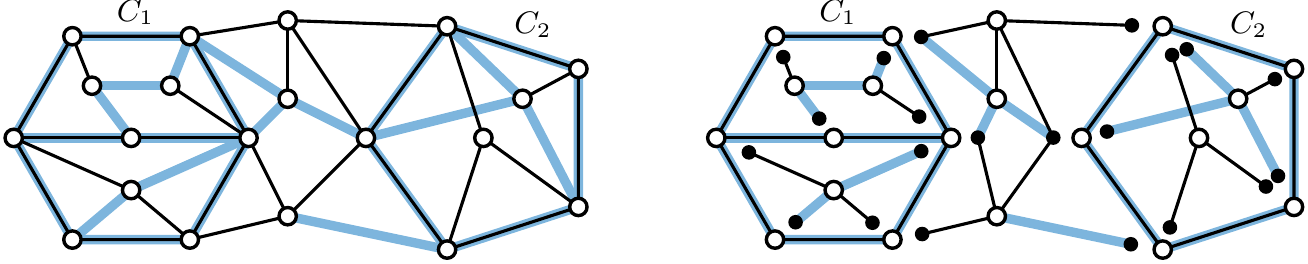}
  \caption{The instance on the left contains two 2-components $C_1$
    and $C_2$.  The corresponding union subbridges are shown on the
    right.}
  \label{fig:subbridges}
\end{figure}

Recall that we have to deal with bridges in two ways.  First, each
2-component forms a \textsc{Sefe} instance with its
$\circled{i}$-bridges, while the union bridges partition these
$\circled{i}$-bridges (yielding union bridge constraints).  Second,
each union bridge yields a union bridge graph, which is a simpler
instance one has to solve.  In both cases, one can deal with
subbridges instead of the whole bridges for the following reason.  For
the first case, we need the $\circled{i}$-bridges only to create the
corresponding dummy bridges.  Thus, it suffices to know their
attachment vertices.  It is readily seen that each bridge $B$ of $C_i$
contains a unique subbridge $S$ incident to $C_i$ and that the
attachments of~$S$ at~$C_i$ are exactly the attachments of~$B$
at~$C_i$.  As this also holds for the union bridges, the union
subbridges already define the correct grouping of the
$\circled{i}$-subbridges.  Concerning the second case, the
\textsc{Sefe} instances that remain after exhaustively applying
Lemma~\ref{lem:bico-sefe} are exactly the union subbridges together
with a set of cycles, one for each incident 2-component.  To conclude,
it remains to show that each of the following three steps runs in
linear time.





\begin{enumerate}
\item Compute for each 2-component the number of
  incident~$\circled{1}$- and $\circled{2}$-subbridges, for each such
  subbridge its attachments, and the grouping of these subbridges into
  union subbridges.
\item Solve \textsc{Sefe} with union bridge constraints on instances
  with biconnected common graph.
\item Compute for each union subbridge 
  a corresponding instance where each 2-component has been replaced by
  a suitable cycle.
\end{enumerate}

For step 1 (Section~\ref{sec:comp-sefe-inst-union-bridge-constr}), we
contract every 2-component of~$G^\cup$ into a single vertex.  The
union subbridges are then basically the split components with respect
to the resulting vertices.  The same holds for the $\circled{1}$- and
$\circled{2}$-subbridges (using $\1G$ and $\2G$ instead of $G^\cup$).
For step 2 (Section~\ref{sec:simult-embedd-union-bridge-constr}), we
modify the algorithm due to Angelini et al.~\cite{adfpr-tsegi-12}.
Augmenting it such that it computes admissible embeddings in
polynomial time is straightforward.  Achieving linear running time is
quite technical and, like the linear version of the original
algorithm, requires some intricate data structures.  For Step~3
(Section~\ref{sec:constr-subbr-inst}), computing the union subbridges
is easy.  To compute a suitable cycle for each incident 2-component
$C_i$, one can make use of the fact that we already know an admissible
embedding of $C_i$ from step~2.

\begin{theorem}
  \label{thm:equivalence-disj-cycles-linear}
  Given a \textsc{Sefe} instance, an equivalent set of instances of total
  linear size such that each 2-component of these instances is trivial
  can be computed in linear time.
\end{theorem}

\subsection{Special Bridges and Common-Face Constraints}
\label{sec:spec-bridg-comm-face-constr}

In Section~\ref{sec:conn-comp-that-are-biconn}, we considered the case
that $C$ is a 2-component of the common graph $G$.  We called the
split components with respect to the vertices of $C$ bridges
(excluding edges in $C$).  Clearly, this definition extends to the
case where $C$ is an arbitrary connected component.  However, the
decomposition into smaller instances does not extend to this more
general case as for example Lemma~\ref{lem:bico-unique-ordering} fails
for non-biconnected graphs.  Nonetheless, we are able to eliminate
some special types of bridges in exchange for so-called common-face
constraints.  The reduction we describe in this section is thus in a
sense weaker than the previous reductions as we reduce a given
\textsc{Sefe} instance to a set of equivalent instances with
common-face constraints.  Many algorithmic approaches allow an easy
integration of additional common-face constraints (see
Section~\ref{sec:common-face-constr}) and hence the more simple
instances resulting from the reduction outweigh the disadvantages
caused by the additional constraints (see
Section~\ref{sec:cons-relat-posit}).

Let $(\1G, \2G)$ be an instance of \textsc{Sefe} with common graph $G$
and let $\mathcal F \subseteq 2^{V(G)}$ be a family of sets of common
vertices.  A given \textsc{Sefe} satisfies the \emph{common-face
  constraints} $\mathcal F$ if and only if $G$ has a face incident to
all vertices in $V'$ for every $V' \in \mathcal F$.  The common-face
constraints $\mathcal F$ are \emph{block-local} if for every $V' \in
\mathcal F$ all vertices in $V'$ belong to the same block of $G$.

Similarly, we say that a union bridge $B$ is \emph{block-local} if all
attachment vertices of $B$ belong to the same block of $C$.  Let
$\ii{B_1}, \dots, \ii{B_k}$ be the $\circled{i}$-bridges (for $i \in
\{1, 2\}$) belonging to $B$.  We say that $B$ is \emph{exclusively
  one-attached} if $\ii{B_j}$ has only a single attachment vertex for
$j = 1, \dots, k$.

Let $B$ be a block-local union bridge of the common connected
component $C$.  Then the attachment vertices of $B$ appear in the same
order in every cycle of $C$ (Lemma~\ref{lem:bico-unique-ordering}).
Thus, we can define the union bridge graph $G_B$ of $B$ as in
Section~\ref{sec:conn-comp-that-are-biconn}.  Consider the
\textsc{Sefe} instance $(\1H, \2H)$ obtained from $(\1G, \2G)$ by
removing the union bridge $B$ (the attachment vertices are not
removed).  It follows from Section~\ref{sec:conn-comp-that-are-biconn}
that $(\1G, \2G)$ admits a \textsc{Sefe} if and only if the union
bridge graph $G_B$ admits a \textsc{Sefe}, and $(\1H, \2H)$ admits a
\textsc{Sefe} with an assignment of $B$ to one of its faces $f$ such
\begin{inparaenum}[(i)]
\item all attachment vertices of $B$ are incident to $f$, and
\item for $i \in \{1, 2\}$, no $\circled{i}$-bridge in $B$ alternates
  with another $\circled{i}$-bridge in $f$.
\end{inparaenum}

If $B$ is not only block-local but also exclusive one-attached, the
latter requirement is trivially satisfied (a $\circled{i}$-bridge that
has only a single attachment vertex cannot alternate).  Thus, it
remains to test whether $G_B$ admits a \textsc{Sefe} and $(\1H, \2H)$
admits a \textsc{Sefe} with block-local common-face constraints.  We
obtain the following theorem.  The linear running time can be shown as
in Section~\ref{sec:prepr-2-comp}.

\begin{theorem}
  \label{thm:no-block-local-exclusive-one-attaced}
  Given a \textsc{Sefe} instance, an equivalent set of instances with
  block-local common-face constraints of total linear size can be
  computed in linear time such that each instances satisfies the
  following property.  No union bridge of a common connected component
  that is not a cycle is block-local and exclusively one-attached.
\end{theorem}

\section{Preprocessing 2-Components in Linear Time}
\label{sec:prepr-2-comp}

As promised in the end of Section~\ref{sec:conn-comp-that-are-biconn},
we prove in this section that the decomposition of a \textsc{Sefe}
instance into equivalent instances where every 2-component is a cycle
can be done in liner time.  Readers who want to skip this section can
continue with Section~\ref{sec:graph-with-common} on
page~\pageref{sec:graph-with-common}.

\subsection{Computing the \textsc{Sefe}-Instances with Union Bridge
  Constraints}
\label{sec:comp-sefe-inst-union-bridge-constr}

We first consider a slightly more general setting.  Let~$G=(V,E)$ be a
graph and let~$C_1,\dots,C_k$ be disjoint connected subgraphs of~$G$.
We are interested in computing the number of bridges of each connected
component $C_i$ together with the attachments to $C_i$.  We show that
this can be done in~$O(n+m)$ time (even if $G$ is non-planar),
where~$n = |V|$ and~$m=|E|$.  Instead of computing directly the
bridges and their attachments, our goal is rather to label each
edge~$e$ that is incident to a vertex of some~$C_i$ but does not
belong to any of the~$C_i$ itself, by the bridge of~$C_i$ that
contains~$e$.  Observe that, if each such incident edge has been
labeled, the information about the number of~$C_i$-bridges and their
attachments can easily be extracted by scanning all incidences of
vertices of~$C_i$ for~$i=1,\dots,k$.  This scanning process can
clearly be performed in total~$O(n + m)$ time.  In the following, we
thus focus on computing this incidence labeling.  Note that, since we
are only interested in the attachments of bridges, it suffices to
consider the corresponding subbridges as they have the same attachment
sets.

Recall that a subbridge is a maximal connected subgraph of~$G$ for
which none of the~$C_i$ is a separator.  Note the high similarity of
the definition of subbridges and the blocks of a graph, which are
maximal connected subgraphs for which no single vertex is a separator.
As with the blocks of a graph, it is readily seen that each edge
of~$G$ that is not contained in one of the~$C_i$ is contained in
exactly one subbridge of~$G$.  We exploit this similarity further and
define the \emph{component-subbridge tree} of~$G$ with respect
to~$C_1,\dots,C_k$ as the graph that contains one vertex~$c_i$ for
each component~$C_i$ and one vertex~$s_j$ for each subbridge~$S_j$.
Two vertices~$c_i$ and~$s_j$ are connected by an edge if and only if
the subbridge~$S_j$ is incident to the component~$C_i$.  Note that,
indeed, the component-subbridge tree is a tree.  Once the
component-subbridge tree has been computed, we can label each edge
of~$E \setminus (\bigcup_{i=1}^k C_i)$ with the subbridge containing
it.

\begin{lemma}
  \label{lem:component-subbridge-tree}
  The component-subbridge tree of a graph~$G$ with respect to disjoint
  connected subgraphs~$C_1,\dots,C_k$ can be computed in linear time.
\end{lemma}

\begin{proof}
  First, contract each component~$C_i$ to a single vertex~$c_i$; call
  the resulting graph~$G'$.  Note that, in~$G'$, the subbridges are
  exactly the maximal connected subgraphs for which none of the
  vertices~$c_i$ is a separator.  We compute the component-subbridge
  tree~$T$ in three steps.  First, compute the block-cutvertex tree
  of~$G'$.  Second, for each cutvertex~$v$ that does not correspond to
  one of the~$c_i$, remove~$v$ and merge its incident blocks into the
  same subbridge.  Finally, create for each component~$C_i$ that has
  only one bridge a corresponding vertex~$c_i$ and attach it as a leaf
  to the unique subbridge incident to~$C_i$.  Clearly, each of the
  steps can be performed in~$O(n+m)$ time.
\end{proof}

As argued above, Lemma~\ref{lem:component-subbridge-tree} can be used
to label in linear time the incident edges of the
components~$C_1,\dots,C_k$ by their corresponding bridges.  For step~1
of our reduction, we take~$C_1,\dots,C_k$ as the 2-components of a
\textsc{Sefe} instance~$(\1G,\2G)$.  We then use the above approach to
label the attachment incidences of the~$\circled{1}$-, $\circled{2}$-,
and the union bridges of~$C_1,\dots,C_k$.  From this we can create the
dummy bridges for each 2-component $C_i$ together with the union
bridge constraints in time linear in the sum of degrees of vertices
in~$C_i$.  By the arguments for
Lemma~\ref{lem:find-admissible-embedding}, the resulting instance
admits a \textsc{Sefe} if and only if $C_i$ has an admissible
embedding.  Since the~$C_i$ are disjoint, it follows that the
construction of all instances can be done in linear time.  This
finishes step~1.

\subsection{Constructing the Subbridge Instances}
\label{sec:constr-subbr-inst}

Let us assume that each 2-component has an admissible embedding, which
is found using the linear-time algorithm described in
Section~\ref{sec:simult-embedd-union-bridge-constr}.  Otherwise a
\textsc{Sefe} of the original graph does not exist.  In the final step
of our reduction, we substitute, in each subbridge, the incident
2-components by a cycle.  This results in a set of \textsc{Sefe}
instances---one for each subbridge---that all admit a solution if and
only if the original instance admits a solution.  They can hence be
handled completely independently.  To efficiently extract all
instances, we process the 2-components independently and replace each
one by a cycle in their incident subbridges.  The time to process a
single 2-component with all its incident subbridges is linear in the
size of the 2-component plus the number of attachments of these
subbridges in the respective 2-component.  It then immediately follows
that processing all 2-components in this way takes linear time.

Consider a fixed 2-component~$C$ with an admissible embedding as
computed in step~2 of the reduction.  Consider a fixed face~$f$
together with the subbridges that are embedded in that face.  For each
bridge~$B$ embedded in~$f$, we construct a list of attachments~$A_B$.
We traverse the facial cycle of~$f$.  At each vertex, we check the
edges embedded inside this face and append the vertex to the list of
each subbridge for which it is an attachment.  Afterwards, we traverse
for each subbridge its list of attachments and replace~$C$ by a cycle
that visits the attachments in the order of the attachment list.  The
time is clearly proportional to the size of~$f$ and the attachments of
the subbridges embedded in~$f$.  Hence processing all faces of all
components in this way takes linear time and yields the claimed
result.  This implements step 3 in linear time.

\subsection{Simultaneous Embedding with Union Bridge Constraints}
\label{sec:simult-embedd-union-bridge-constr}

In this section, we show how to solve \textsc{Sefe} with union bridge
constraints in linear time if the common graph is biconnected.  Our
algorithm is based on the algorithm by Angelini et
al.~\cite{adfpr-tsegi-12}.  Note that our extension to this algorithm
is twofold.  We allow bridges with an arbitrary number of attachment
vertices.  The original algorithm allows only two attachments per
$\circled{i}$-bridge (i.e., each bridge is a single exclusive edge).
Moreover, we have to deal with union bridge constraints.  To avoid
some special cases and simplify the description, we sometimes deviate
from the notation used by Angelini et al.  Our focus lies on a
linear-time implementation, the correctness of our approach directly
follows from the correctness of the algorithm by Angelini et al.

\begin{figure}
  \centering
  \includegraphics[page=1]{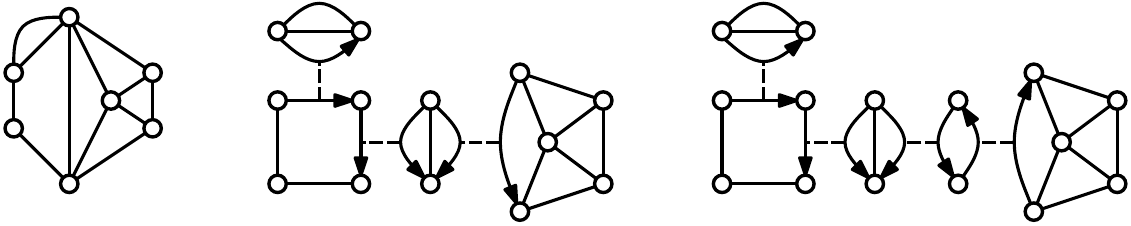}
  \caption{A graph (left) together with its SPQR-tree (middle).  The
    Q-nodes are omitted to improve readability.  The augmented
    SPQR-tree (right) contains an additional S-node whose skeleton is
    a cycle of length~2.}
  \label{fig:spqr-tree-augmented}
\end{figure}

Let $G$ be the biconnected common graph.  We consider the (unrooted)
\emph{augmented SPQR-tree} $\T$ of $G$, which is defined as follows;
see Figure~\ref{fig:spqr-tree-augmented}.  Let $\T'$ be the SPQR-tree
of $G$ and let $\mu_1$ and $\mu_2$ be two adjacent nodes in $\T'$ such
that each of them is a P- or an R-node.  We basically insert a new
S-node between $\mu_1$ and $\mu_2$ whose skeleton is a cycle of
length~2 (i.e., a pair of parallel virtual edges).  More precisely,
let $\eps_1 = \{s, t\}$ and $\eps_2 = \{s, t\}$ be the virtual edges
in $\mu_1$ and $\mu_2$, respectively, that correspond to each other.
We subdivide the edge ${\mu_1,\mu_2}$ in $\T'$; let $\mu$ be the new
subdivision vertex.  The skeleton $\skel(\mu)$ contains the vertices
$s$ and $t$ with two virtual edges $\eps_1'$ and $\eps_2'$ between
them corresponding to $\eps_1$ in $\skel(\mu_1)$ and $\eps_2$ in
$\skel(\mu_2)$, respectively.  Applying this augmentation for every
pair of adjacent nodes that are P- or R-nodes gives the augmented
SPQR-tree $\T$.  Note that P- and R-nodes in $\T$ have only S- and
Q-nodes as neighbors.  Moreover, no two S-nodes are adjacent.

Consider a bridge $B$ and let $\mu$ be a node of $\T$.  A virtual edge
$\eps$ in $\skel(\mu)$ is an \emph{attachment} of $B$ if its expansion
graph contains an attachment vertex of $B$.  We say that $B$ is
\emph{important} for $\mu$ if it has at least two distinct attachments
among the vertices and virtual edges of $\skel(\mu)$ that are not two
adjacent vertices in $\skel(\mu)$.  It is clearly necessary, that
$\skel(\mu)$ admits an embedding such that for every union bridge $B$
the attachments in $\skel(\mu)$ are incident to a common face.  An
embedding having this property is called \emph{compatible}.  This
leads to the following first step of the algorithm.

\begin{sefe-bico-step}[Compatible embeddings]
  \label{step:1}
  For every P- and R-node $\mu$, compute the important union bridges
  with their attachments.  If $\mu$ is an R-node, check whether the
  unique (up to flip) embedding $\skel(\mu)$ is compatible.  If $\mu$
  is a P-node check whether it admits a compatible embedding and fix
  such an embedding (up to flip).
\end{sefe-bico-step}

If Step~\ref{step:1} fails, the instance does not admit a
\textsc{Sefe}.  Note that the skeleton of a P-node might admit several
compatible embeddings.  However, fixing an arbitrary compatible
embedding up to flip does not make a solvable \textsc{Sefe} instance
unsolvable~\cite{adfpr-tsegi-12}.  Thus, after Step~\ref{step:1}, we
can assume that the embedding of every skeleton is fixed and it
remains to decide for each skeleton whether its embedding should be
flipped or not.  We call the embedding fixed for a skeleton its
\emph{reference embedding}.

For every P- and R-node $\mu$ let $x_\mu$ be a binary decision
variable with the following interpretation.  The skeleton $\skel(\mu)$
is embedded according to its reference embedding and according to the
flipped reference embedding if $x_\mu = 0$ and $x_\mu = 1$,
respectively.  By considering the S-nodes of the augmented SPQR-tree,
one can derive necessary conditions for these variables that form an
instance of \textsc{2-Sat} (actually, we only get equations and
inequalities, which is a special case of \textsc{2-Sat}).

Let $\mu$ be an S-node.  We assume the edges in $\skel(\mu)$ to be
oriented such that $\skel(\mu)$ is a directed cycle.  Thus, we can use
the terms left face and right face to distinguish the faces of
$\skel(\mu)$.  Let $B$ be a bridge that is important for $\mu$.  We
can either embed $B$ into the left or into the right face of
$\skel(\mu)$.  We define the binary variable $x_B^\mu$ with the
interpretation that $B$ is embedded into the right face and into the
left face of $\skel(\mu)$ if $x_B^\mu = 0$ and $x_B^\mu = 1$,
respectively.

\begin{figure}[tb]
  \centering
  \includegraphics[page=1]{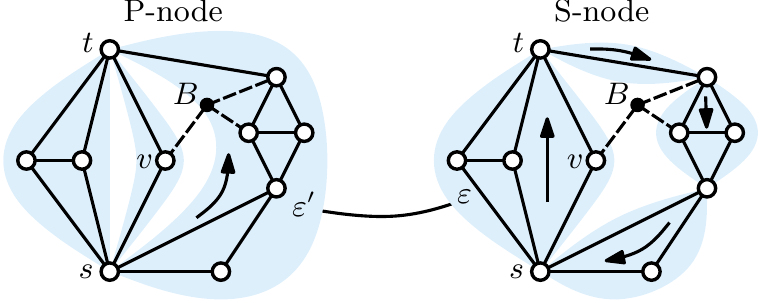}
  \caption{An S-node with a bridge $B$ having $\eps$ as right-sided
    attachment.  The virtual edges are illustrated as gray regions
    with their expansion graph inside.}
  \label{fig:attachment-in-s-node}
\end{figure}

Assume that the virtual edge $\eps = st$ (oriented from $s$ to $t$) in
$\skel(\mu)$ is an attachment of $B$, i.e., an attachment vertex $v
\notin \{s, t\}$ of $B$ lies in the expansion graph of $\eps$.  Let
$\mu'$ be the neighbor of $\mu$ corresponding to $\eps$ and let
$\eps'$ be the twin of $\eps$ in $\skel(\mu')$ (also oriented from $s$
to $t$); see Figure~\ref{fig:attachment-in-s-node}.  Clearly, $B$ is
also important for $\mu'$ as $\eps'$ contains an attachment vertex of
$B$ while the attachment vertex $v$ is not contained in $\eps'$.  In
Step~\ref{step:1} we ensured that $\skel(\mu')$ is embedded such that
$v$ (or the virtual edge containing $v$) shares a face with $\eps'$.
If this face lies to the left of $\eps'$ in $\skel(\mu')$, we say that
$v$ is an attachment \emph{on the right side} of $\eps$ in
$\skel(\mu)$ (as in the example in
Figure~\ref{fig:attachment-in-s-node}).  Otherwise, if this faces lies
to the right of $\eps'$, we say that $v$ is an attachment \emph{on the
  left side} of $\eps$.  For a bridge $B$ with attachment vertex $v$
we also say that the attachment $\eps$ is \emph{right-sided} and
\emph{left-sided} if $v$ lies on the right and left side of $\eps$,
respectively.

Assume $B$ has an attachment on the right side of the virtual edge
$\eps$ in the skeleton $\skel(\mu)$ (as in
Figure~\ref{fig:attachment-in-s-node}).  Assume further that the
skeleton $\skel(\mu')$ of the corresponding neighbor is not flipped,
i.e., $x_{\mu'} = 0$.  Then $B$ must be embedded into the face to the
right of the cycle $\skel(\mu)$, i.e., $x_B^\mu = 0$.  Conversely, if
the embedding of $\skel(\mu')$ is flipped, i.e., $x_{\mu'} = 1$, then
$B$ must lie in the left face of $\skel(\mu)$, i.e., $x_B^\mu = 1$.
This necessary condition is equivalent to the equation~$x_{\mu'} =
x_B^\mu$.  Similarly, if $B$ has an attachment on the left side of
$\eps$, we obtain the inequality $x_{\mu'} \not= x_B^\mu$.  We call
the resulting set of equations and inequalities the \emph{consistency
  constraints} of the bride $B$ in $\mu$.  This leads to the second
step of the algorithm.

\begin{sefe-bico-step}[Consistency constraints]
  \label{step:2}
  For every S-node $\mu$ compute the important $\circled{i}$-bridges
  (for $i \in \{1, 2\}$) and union bridges together with their
  attachments.  For attachments in virtual edges also compute whether
  they are left- or right-sided.  Then add the consistency constraints
  of these bridges in $\mu$ to a global \textsc{2-Sat} formula.
\end{sefe-bico-step}

The consistency constraints are necessary but not sufficient as they
do not ensure that no two alternating bridges of the same type are
embedded into the same face.  Consider an S-node $\mu$ with two
important bridges $B$ and $B'$.  Assume these two bridges alternate
(i.e., they have alternating attachments in the cycle $\skel(\mu)$).
Embedding $B$ and $B'$ on the same side of $\skel(\mu)$ yields a
crossing between an edge in $B$ and an edge in $B'$.  Thus, if $B$ and
$B'$ are both $\circled{1}$- or both $\circled{2}$-bridges, then they
must be embedded to different side of $\skel(\mu)$.  In this case, we
obtain the inequality $x_B^\mu \not= x_{B'}^\mu$.  This inequality is
called \emph{planarity constraint}.

\begin{sefe-bico-step}[Planarity constraints]
  \label{step:3}
  For every S-node $\mu$ compute the pairs of important
  $\circled{1}$-bridges that alternate in $\skel(\mu)$.  Do the same
  for $\circled{2}$-bridges.  For each such pair add the planarity
  constraint to the global \textsc{2-Sat} formula.
\end{sefe-bico-step}

Finally, we have to embed $\circled{i}$-bridges belonging to the same
union bridge into the same face.  Let $B$ be an $\circled{i}$-bridge
and let $B'$ be the union bridge it belongs to.  The
\emph{union-bridge constraint} of $B$ in $\mu$ is the equations
$x_{B}^\mu = x_{B'}^\mu$.

\begin{sefe-bico-step}[Union-bridge constraints]
  \label{step:4}
  For every S-node $\mu$, add the union-bridge constraint of each
  important $\circled{i}$-bridge to the global \textsc{2-Sat} formula.
\end{sefe-bico-step}

After Steps~\ref{step:2}--\ref{step:4}, the global \textsc{2-Sat}
formula is solved in linear time~\cite{apt-ltatt-79}.  The solution
determines for every P- and every R-node $\mu$, whether the reference
embedding of $\skel(\mu)$ should be flipped or not, which completely
fixes the embedding of the common graph $G$.  Of course, there might
be different solutions of the \textsc{2-Sat} formula, yielding
different embeddings.  However, if one of these solutions yields a
\textsc{Sefe}, then any of the solutions does~\cite{adfpr-tsegi-12}.
Thus, one can simply take one solution and check whether it yields a
\textsc{Sefe} (with union bridge constraints) or not.

\begin{sefe-bico-step}[Final step]
  \label{step:5}
  Test whether the given instance admits a \textsc{Sefe} with union
  bridge constraints assuming that the embedding of the common graph
  is fixed.
\end{sefe-bico-step}

It remains to implement Steps~\ref{step:1}--\ref{step:5} in linear
time, which is done in the following.  We first note that there are
too many important bridges to be able to compute them in linear time
(as required in Step~\ref{step:1} and Step~\ref{step:2}).  However,
similar to Angelini et al.~\cite{adfpr-tsegi-12}, we can show that
many important bridges can be omitted without loosing the correctness
of the algorithm, which leads to a linear-time implementation.

\subsubsection{Too Many  Bridges are Important}

We start with the observation that computing all important union
bridges for every P- and R-node of the SPQR-tree is actually a bad
idea, as there may be $\Omega(n)$ bridges each being important in
$\Omega(n)$ nodes.  Thus, explicitly computing all of them would
require $\Omega(n^2)$ time.  Consider the graph $G$ in
Figure~\ref{fig:compatible-embeddings-quadratic-size} whose SPQR-tree
$\mathcal T$ (without Q-nodes) is a path.  Let $\mu$ be one of the
P-nodes (note that $\skel(\mu)$ has two virtual edges and one normal
edge) and let $B$ be one of the bridges shown in
Figure~\ref{fig:compatible-embeddings-quadratic-size}.  Clearly, the
expansion graphs of both virtual edges of $\skel(\mu)$ contain at
least one attachment vertex of $B$.  Thus, $B$ is important for $\mu$
and $B$ has the two virtual edges of $\skel(\mu)$ as attachments.  As
this may hold for a linear number of bridges, we get the above
observation.

\begin{figure}
  \centering
  \includegraphics[page=1]{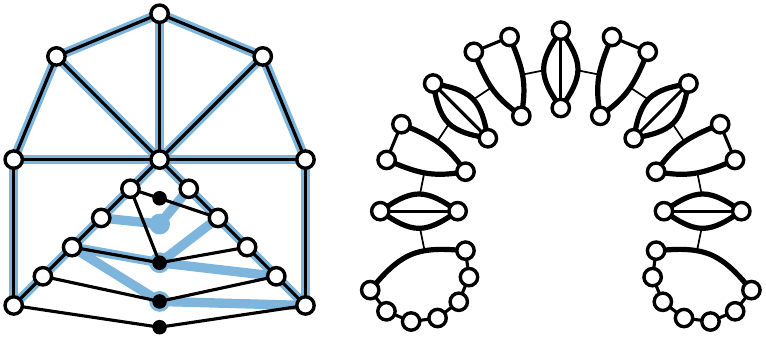}
  \caption{A graph with many bridges (left) each of which being
    important for every node of the SPQR-tree (right).}
  \label{fig:compatible-embeddings-quadratic-size}
\end{figure}

To resolve this issue, note that from the perspective of $\mu$ (in the
above example), all bridges look the same in the sense that they have
the same set of attachments.  Intuitively, they thus lead to similar
constraints and it seems to suffices to know only one of these
bridges.  In the following, we first show that omitting some of the
bridges is indeed safe in the sense that the algorithm remains
correct.  Then we show how to compute the remaining bridges
efficiently.

\subsubsection{Omitting Some Important Bridges}

In this section we show how to change the algorithm described above
(Steps~\ref{step:1}--\ref{step:5}) slightly without changing its
correctness.  We say it is \emph{safe} to do something if doing it
does preserve the correctness of the algorithm.  In particular, we
show that it is safe to omit some of the important bridges.  In the
subsequent sections we then show that the remaining important bridges
can be computed efficiently leading to an efficient implementation of
all five steps.

Let $G$ be the common graph and let $\mathcal T$ be its SPQR-tree.
Let $\mu$ be an inner node of $\mathcal T$ and let $B$ be a bridge
that is important for $\mu$ with attachments $a_1, \dots, a_\ell$ in
$\skel(\mu)$ (recall that each of the $a_i$ is either a vertex or a
virtual edge of $\skel(\mu)$).  We call an attachment $a_i$ (with $1
\le i \le \ell$) \emph{superfluous}, if $a_i$ is a vertex in
$\skel(\mu)$ such that $B$ has another attachment $a_j$ that is a
virtual edge incident to the vertex $a_i$; see
Figure~\ref{fig:superfluous-attachment}a.  The following lemma shows
that the term ``superfluous'' is justified.

\begin{figure}
  \centering
  \includegraphics[page=1]{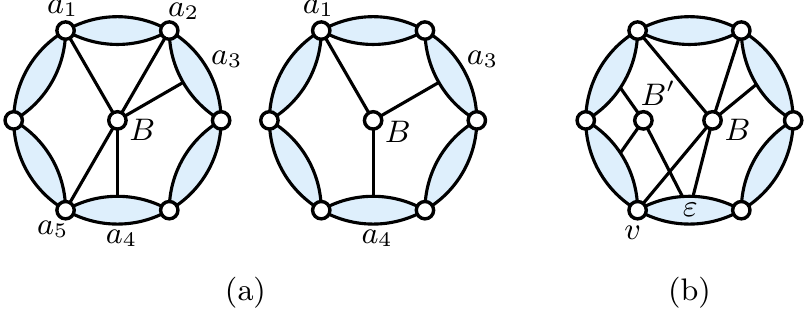}
  \caption{(a)~The bridge $B$ has five attachments $a_1, \dots, a_5$.
    Attachments $a_2$ and $a_5$ are superfluous due to $a_3$ and
    $a_4$, respectively.  (b)~The bridges $B$ and $B'$ alternate only
    due to the superfluous attachment $v$ of $B$.  However, the
    consistency constraints already synchronize $B$ and $B'$ as they
    have $\eps$ as common attachment.}
  \label{fig:superfluous-attachment}
\end{figure}

\begin{lemma}
  \label{lem:superfluous-attachments}
  Omitting superfluous attachments is safe.
\end{lemma}
\begin{proof}
  There are two situation in which missing superfluous attachments
  might play a role.  First, when we check a P- or R-node skeleton for
  a compatible embedding (Step~\ref{step:1}).  Second, when we add the
  planarity constraints for alternating $\circled{i}$-bridges
  (Step~\ref{step:3}).  Let $v$ be the superfluous attachment of $B$
  and let $\eps$ be a virtual edge incident to $v$ that is also an
  attachment of $B$ in $\skel(\mu)$.  Concerning compatible
  embeddings, we have to make sure that $\skel(\mu)$ admits an
  embedding where all attachments of $B$ are incident to a common
  face.  Clearly, $v$ is incident to both faces the virtual edge
  $\eps$ is incident to.  Thus, omitting the attachment $v$ does not
  change anything.

  Concerning the planarity constraints, we have to consider the case
  that $B$ alternates with another bridge $B'$ only due to the
  attachment $v$ in $B$.  This can only happen if $B'$ has $\eps$ as
  attachment; see Figure~\ref{fig:superfluous-attachment}b.  However,
  then the consistency constraints (Step~\ref{step:2}) either forces
  $B$ and $B'$ into different faces (if their attachment in $\eps$
  lies on different sides of $\eps$) and everything is fine, or they
  force $B$ and $B'$ to lie in the same face.  In the latter case, the
  instance is clearly not solvable, which will be found out in
  Step~\ref{step:5}.
\end{proof}

In the following we always omit superfluous attachments even when we
do not mention it explicitly.  Note that this retroactively changes
the definition of important bridges slightly, i.e., a bridge $B$ is
important for a node $\mu$ if $B$ has at least two (non-superfluous)
attachments in $\skel(\mu)$.  

To show that we can omit sufficiently many important bridges to get a
linear running time, we have to root the SPQR-tree $\mathcal T$.  More
precisely, we choose an arbitrary Q-node as the root of~$\mathcal T$.

We categorize the important bridges in different types of bridges
depending on their attachments.  To this end, we first define
different types of attachments.  Let $\mu$ be a node of the (rooted)
SPQR-tree and let $B$ be an important bridge of $\mu$ with attachments
$a_1, \dots, a_\ell$.  Recall that an attachment is either a vertex of
a virtual edge of $\skel(\mu)$.  If $a_i$ is a pole of $\skel(\mu)$,
we call it a \emph{pole attachment}.  If $a_i$ is the parent edge of
$\skel(\mu)$, we call it \emph{parent attachment}.  All other
attachments are called \emph{child attachments}.  

We say that the important bridge $B$ is a \emph{regular bridge} of
$\mu$ if $B$ has at least two child attachments.  If $B$ has only a
single child attachment, it has either a parent attachment or one or
two pole attachments (note that pole attachments are superfluous in
the presence of a parent attachment).  We call $B$ \emph{parent
  bridge} and \emph{pole bridge} in the former and latter case,
respectively.  Note that $B$ must have at least one child attachment
as it otherwise cannot be important.

As shown before, we cannot hope to compute all important bridges
efficiently as there may be too many of them.  Thus, we show in the
following that omitting some important bridges is safe.

\begin{lemma}
  \label{lem:omitting-bridges}
  Let $B$ be a parent bridge whose child attachment $a$ is a virtual
  edge.  Omitting another parent bridge or a pole bridge with child
  attachment $a$ (while keeping $B$) is safe.
\end{lemma}
\begin{proof}
  Let $B'$ be an other important bridge of $\mu$ and let $a$ be the
  child attachment of $B$ and~$B'$.  Let $\eps$ be the parent edge of
  $\skel(\mu)$.  In Step~\ref{step:1}, the bridge $B$ forces the
  embedding of $\skel(\mu)$ to have $a$ and $\eps$ on a common face.
  If $B'$ is a parent bridge, it requires the same (and thus no
  additional) condition for the embedding of $\skel(\mu)$.  If $B'$ is
  a pole bridge, it requires one of the poles (or both of them) to be
  incident to a common face with $a$.  However, this is clearly the
  case if the parent edge $\eps$ (which is incident to both poles)
  shares a face with $a$.  Hence, omitting $B'$ does not change
  anything in Step~\ref{step:1}.

  We cannot argue separately for Steps~\ref{step:2}--\ref{step:4} as
  they all contribute to the same global \textsc{2-Sat} formula.  We
  call the \textsc{2-Sat} formula we get when omitting $B'$
  \textsc{new}.  The formula we get when not omitting $B'$ is called
  \textsc{original}.  The straightforward way to prove that omitting
  $B'$ is safe would be to show that the new and the original
  \textsc{2-Sat} formulae are equivalent in the sense that they have
  the same solutions.  However, this is not true as omitting $B'$ can
  make the \textsc{2-Sat} formula solvable whereas it was unsolvable
  before.

  Assume the original \textsc{2-Sat} formula is not solvable.  Then
  the given instance does not admit a \textsc{Sefe} with union bridge
  constraints.  In this case, Step~\ref{step:5} will never succeed no
  matter what we do in the steps before.  Thus, we only have to argue
  for the case that the original \textsc{2-Sat} formula is solvable.
  In this case, the only thing that can go wrong is that the new
  \textsc{2-Sat} formula admits a solution that is not a solution in
  the original formula.  This solution could then result in an
  embedding of the common graph $G$ that does not admit a
  \textsc{Sefe}, whereas all solutions of the original \textsc{2-Sat}
  formula lead to a \textsc{Sefe}.

  To get a handle on this, we consider the conflict graph of a
  \textsc{2-Sat} formula.  First note that the formula we obtain is
  special in the sense that each constraint is either an equation or
  an inequality.  We define the conflict graph to have a vertex for
  each variable and an edge connecting two vertices if there is an
  (in-)equality constraint between the corresponding variables.
  Recall that the \textsc{2-Sat} formula has a variable $x_\mu$ for
  every P- or R-node $\mu$ indicating whether the embedding of
  $\skel(\mu)$ has to be flipped or not.  In the following we show
  that two such variables $x_\mu$ and $x_\eta$ that are in the same
  connected component of the conflict graph of the original
  \textsc{2-Sat} formula are also in the same connected component with
  respect to the new formula.  This shows that every embedding
  resulting from a solution of the new formula also yields a solution
  of the original formula (assuming that the original formula has a
  solution at all).

  First note that omitting an important bridge $B'$ in an S-node $\mu$
  has the effect that the vertex corresponding to the variable
  $x_{B'}^\mu$ is deleted in the conflict graph (for convenience we
  also denote the vertex by $x_{B'}^\mu$).  Recall that the variable
  $x_{B'}^\mu$ represents the decision of whether $B'$ is embedded
  into the left or the right face of $\skel(\mu)$.  We show that
  $x_{B'}^\mu$ is dominated by $x_B^\mu$ in the sense that every
  neighbor of $x_{B'}^\mu$ is also connected to $x_B^\mu$ by a path
  not containing~$x_{B'}^\mu$.  Thus, removing $x_{B'}^\mu$ does not
  change the connected components of the conflict graph.

  Consider the consistency constraints of $B'$ (Step~\ref{step:2}).
  Let $a$ be an attachment of $B'$.  If $a$ is a vertex in
  $\skel(\mu)$, we do not get a consistency constraint there.  If $a$
  is a virtual edge, it corresponds to a neighbor $\eta$ of $\mu$ and
  we get the constraint $x_\eta = x_{B'}^\mu$ or $x_\eta \not=
  x_{B'}^\mu$.  But then $B$ has also $a$ as attachment and thus we
  also have one of the two constraints $x_\eta = x_{B}^\mu$ or $x_\eta
  \not= x_{B}^\mu$.

  For the planarity constraints (Step~\ref{step:3}), first assume that
  $B'$ is a parent bridge.  Then $B$ alternates with another bridge
  $B''$ if and only if $B'$ alternates with $B''$.  Thus, if we have
  the constraint $x_{B'} \not= x_{B''}$ we also have the constraint
  $x_{B} \not= x_{B''}$.  If $B'$ is a pole bridge, there might be a
  bridge $B''$ alternating with $B'$ (yielding the constraint $x_{B'}
  \not= x_{B''}$), whereas $B$ and $B''$ do not alternate.  However,
  this can only happen if $B''$ has the parent edge $\eps$ of
  $\skel(\mu)$ as attachment.  Let $\eta$ be the neighbor of $\mu$
  corresponding to $\eps$ (i.e., the parent of $\mu$).  From
  Step~\ref{step:2} we then have the consistency constraint between
  $x_\eta$ and $x_{B''}^\mu$ and also one between $x_\eta$ and
  $x_{B}^\mu$ ($B$ is a parent bridge).  Thus, the conflict graph
  includes the path $x_B^\mu x_\eta x_{B''}^\mu$, which still exists
  after removing $B'$.

  Finally, if $B''$ is the union-bridge including $B'$, we loose the
  union-bridge constraint (Step~\ref{step:4}) $x_{B'}^\mu =
  x_{B''}^\mu$ by omitting $B'$.  Note that the attachments of $B'$
  are a subset of the attachments of $B''$.  Thus, $B''$ has the child
  attachment $a$ (which is a virtual edge).  Let $\eta$ be the
  neighbor of $\mu$ corresponding to $a$.  Then we have a consistency
  constraint connecting $x_{B''}^\mu$ with $x_\eta$.  Moreover,
  $x_\eta$ is also connected to $X_B^\mu$ as $a$ is an attachment of
  $B$.  This yields a connection from $x_B^\mu$ to $x_{B''}^\mu$.
\end{proof}

This lemma gives rise to the following definition.  Let $B'$ be a pole
bridge of $\mu$ with child attachment $a$ that is a virtual edge in
$\skel(\mu)$.  Let further $B$ be a parent bridge of $\mu$ with child
attachment $a$.  We say that $B'$ is \emph{dominated} by $B$.  If $B'$
is a parent bridge of $\mu$ with child attachment $a$ instead, we say
that $B$ and $B'$ are \emph{equivalent}.  Note that being equivalent
is clearly an equivalence relation.  Lemma~\ref{lem:omitting-bridges}
shows that we can omit dominated pole bridges and all but one parent
bridge for each equivalence class. 

\subsubsection{Computing the Remaining Important Bridges}

In this section we show that all bridges that are not omitted due to
Lemma~\ref{lem:omitting-bridges} can be computed in linear time.
Actually, we compute slightly more information, which we need in some
intermediate steps.  We for example never omit a parent bridge $B$ in
$\mu$ if $B$ is regular in the parent of $\mu$.  We call such a bridge
\emph{semi-regular} in $\mu$.

Moreover, consider an important bridge $B$ of $\mu$ and let $a$ be an
attachment of $B$ in $\skel(\mu)$ that is a virtual edge.  A vertex
$v$ of $G$ that is not incident to $a$ in $\skel(\mu)$ is a
\emph{representative} of the attachment $a$ if $v$ is an attachment
vertex of $B$ and included in the expansion graph of $a$.  If the
attachment $a$ is a vertex of $\skel(\mu)$, we say that $a$ is its own
representative.  Note that every attachment of $B$ has at least one
representative vertex.  When we compute the important bridges of a
node $\mu$, we also compute their attachments in $\mu$ together with
at least one representative for each attachment.  Before we actually
compute bridges, we provide some general tools.

\begin{lemma}
  \label{lem:parent-and-pole-attachments}
  After linear preprocessing time, the pole and parent attachments of
  a bridge in a given node together with a representative of each
  attachment can be computed in constant time.  
\end{lemma}
\begin{proof}
  We first do some preprocessing.  We define the
  \emph{SPQR-vertex-tree} $\mathcal T$ to be the tree obtained from
  the SPQR-tree by removing its Q-nodes and attaching every vertex $v$
  of $G$ as a leaf to the highest node that contains $v$ in its
  skeleton.  Note that this is the unique node that contains $v$ in
  its skeleton but not as pole.
  For a non-pole vertex $v$ in a skeleton $\skel(\mu)$, we say that
  the leaf $v$ (which is a child of $\mu$) corresponds to $v$.  Note
  that this makes sure that every child attachment in $\skel(\mu)$
  corresponds to a child of $\mu$ in $\mathcal T$.

  Assume the vertices of $G$ (i.e., the leaves of $\mathcal T'$) to be
  numbered according to a DFS-ordering (which we get in $O(n)$ time).
  We start by sorting the attachment vertices of each bridge according
  to this DFS-ordering.  For all bridges $B_1, \dots, B_k$ together,
  this can be done in time $O\left(n + \sum_{i=1}^k|B_i|\right)$ as
  follows.  To simplify the notation, we identify every
  bridge $B_i$ with its set of attachment vertices.  First, sort all
  pairs $(v, B_i)$ with $v \in B_i$ (i.e., basically the disjoint
  union of the $B_i$) using bucket sort with buckets $1, \dots, n$
  (for each vertex $v$, one bucket containing all pairs $(v, B_i)$).
  It then suffices to iterate over all these sorted pairs to extract
  the sets $B_1, \dots, B_k$ sorted according to this DFS-ordering.

  For a node of $\mathcal T$, the leaves that are descendents of this
  node appear consecutively in the DFS-ordering.  By going bottom-up
  in $\mathcal T$, we can compute for every node $\mu$ the leaf with
  the smallest and the leaf with the largest number among the
  descendents of $\mu$.  Thus, given a vertex $v$ of $G$, we can
  decide in constant time whether $v$ is a descendent of $\mu$.

  Now we answer the queries.  Let $B$ be a bridge that is important in
  a node $\mu$.  If the vertex with the smallest and the vertex with
  the largest number in $B$ (which we can find in constant time as $B$
  is sorted) are both descendents of $\mu$, all attachment vertices of
  $B$ are descendents of $\mu$.  In this case $B$ has neither a pole
  nor a parent attachment in $\mu$.  Otherwise, by looking at the
  first two and the last two elements in $B$ we can distinguish the
  following two cases.
  \begin{inparaenum}[(i)]
  \item There is an attachment vertex $v$ that is not a descendent of
    $\mu$ and not a pole of $\skel(\mu)$.  Then the parent edge of
    $\skel(\mu)$ is an attachment of $B$ in $\skel(\mu)$ and $v$ is a
    representative for this attachment.  If $B$ also has pole
    attachments in $\skel(\mu)$, they are superfluous and we can
    ignore them by Lemma~\ref{lem:superfluous-attachments}.
  \item There is no such vertex $v$.  Then all attachment vertices in
    $B$ are descendents of $\mu$ except for maybe the two poles.  In
    this case, the poles of $\skel(\mu)$ that are attachments of $B$
    are among the first or last two vertices in $B$ and thus we find
    all pole attachments of $B$ in $\skel(\mu)$.
  \end{inparaenum}
  This concludes the proof.
\end{proof}

\begin{lemma}
  \label{lem:regular-bridges}
  The regular bridges and their attachments together with a
  representative for each attachment can be computed in linear time.
\end{lemma}
\begin{proof}
  Let $\mathcal T$ be the SPQR-vertex-tree of $G$ and assume again
  that the vertices in every bridge are sorted according to a
  DFS-order of the leaves in $\mathcal T$.  We first show how to
  compute all nodes in which a given bridge $B_i$ is regular in
  $O(|B_i|)$ time.  Note that this implicitly shows that $B_i$ has
  only $O(|B_i|)$ regular vertices.

  The \emph{lowest common ancestor (LCA)} of two vertices $u$ and $v$
  is the highest node on the path between $u$ and $v$.  We denote it
  by $\LCA(u, v)$.  Clearly, the vertices $u$ and $v$ are descendents
  of different children of $\LCA(u, v)$.  Thus, if $B_i$ has the
  attachment vertices $u$ and $v$, then $B_i$ has two different child
  attachments in $\LCA(u, v)$ with representatives $u$ and $v$.  Thus,
  $\LCA(u, v)$ is active.  Conversely, if $B_i$ is regular in a node
  $\mu$, it has two different child attachments.  Let $u$ and $v$ be
  representatives of these two attachments.  Clearly, $u$ and $v$ are
  descendents of different children of $\mu$ and thus $\mu = \LCA(u,
  v)$ holds.  Hence, the nodes in which $B_i$ is regular are exactly
  the LCAs of pairs of attachment vertices in $B_i$.

  To see that it is not necessary to consider all pairs of vertices in
  $B_i$, let $u$, $v$, and $w$ with $u < v < w$ (according to the
  DFS-ordering) be contained in $B_i$.  Then $\LCA(u, w) = \LCA(u, v)$
  or $\LCA(u, w) = \LCA(v, w)$ holds for the following reason.  If
  $\mu = \LCA(u, v) = \LCA(v, w)$, then $u$, $v$, and $w$ are
  descendants of three different children of $\mu$.  Thus, $\LCA(u, w)
  = \mu$ also holds.  Otherwise, assume $\mu = \LCA(u, v)$ is a
  descendant of $\eta = \LCA(v, w)$.  Thus, from the perspective of
  $\eta$, the vertices $u$ and $v$ and the node $\mu$ are descendants
  of the same child whereas $w$ is the descendant of a different
  child.  Thus $\LCA(u, w) = \LCA(v, w) = \eta$.  Hence, to compute
  all nodes in which $B_i$ is regular, it suffices to compute the LCA
  for pairs of attachment vertices in $B_i$ that are consecutive (with
  respect to the ordering we computed before).  As the LCA can be
  computed in constant time after $O(n)$ preprocessing
  time~\cite{ht-fafnca-84}, this gives us all regular bridges of all
  nodes of the SPQR-tree in overall $O(n + \sum_{i = 1}^k|B_i|)$ time.

  Given a node $\mu$ and a regular bridge $B$ of $\mu$, we are still
  lacking the attachments of $B$ in $\skel(\mu)$.  We can get the
  parent and pole attachments using
  Lemma~\ref{lem:parent-and-pole-attachments}.  Consider a child
  attachment $a$ of $B$ in $\skel(\mu)$.  Note that the attachment
  vertices of $B$ that are descendants of $a$ are consecutive with
  respect to the DFS-order.  When choosing the first or last of these
  vertices, we get an attachment vertex $v$ in $B$ with the following
  properties: $v$ is a descendent of the child of $\mu$ that
  corresponds to $a$ and either $\mu = \LCA(u, v)$ for the predecessor
  $u$ of $v$ in $B$ (ordered according to the DFS-order) or $\mu =
  \LCA(v, w)$ for the successor $w$ of $v$.  Thus, we actually already
  know a representative for every child attachment of $B$ in
  $\skel(\mu)$.  To find for an attachment vertex $v$ (the
  representative) the corresponding attachment in $\skel(\mu)$, we
  need to find the child of $\mu$ that has the leaf $v$ as a
  descendant.  This can be done for all active bridges simultaneously
  by processing $\mathcal T$ bottom-up while maintaining a union-find
  data structure.  As the sequence of union operations is known in
  advance, each union and find operation takes amortized $O(1)$
  time~\cite{gt-ltascdsu-85}.  Thus, it takes $O\left(n +
    \sum_{i=1}^k|B_i|\right)$ time in total.
\end{proof}

\begin{lemma}
  \label{lem:semi-regular-bridges}
  The semi-regular bridges and their attachments together with a
  representative for each attachment can be computed in linear time.
\end{lemma}
\begin{proof}
  By Lemma~\ref{lem:regular-bridges}, we know for every node $\mu$ all
  the bridges that are regular in $\mu$.  We can assume that the
  regular bridges of a $\mu$ are stored as a list that is sorted
  according to an arbitrarily chosen order of the bridges (we just
  have to process the bridges in this order in the proof of
  Lemma~\ref{lem:regular-bridges}).  Moreover, we know all the
  attachments of $B$ in $\skel(\mu)$ together with a representative
  for each attachment.

  Let $a$ be a child attachment of $B$ in $\skel(\mu)$ with
  representative $v$.  Let $\eta$ be the child of $\mu$ corresponding
  to $a$.  If $\eta$ is a leaf, it actually must be the vertex $v$ and
  there is nothing to do.  Otherwise, $v$ is a descendent of $\eta$
  and thus $B$ has a child attachment in $\eta$.  Moreover, it also
  has a parent attachment in $\eta$ as it otherwise would not be
  regular in the parent $\mu$ of $\eta$.  Thus, $B$ is either regular
  or a parent bridge (and hence semi-regular) in $\eta$.

  By processing all regular bridges of $\mu$ like that, we can build
  for $\eta$ (and all other children of $\mu$) a list of bridges that
  contains all semi-regular bridges and additionally some regular
  bridges.  Note that building up this list takes constant time for
  each attachment of regular bridges.  As we computed those
  attachments in linear time, there cannot be more than that many
  attachments.

  To get rid of the bridges that are actually regular and not
  semi-regular, note that we can assume the list of semi-regular and
  regular bridges to be sorted (according to the arbitrary order of
  bridges we chose before).  Thus, we can simply process this list and
  the list of regular bridges of $\eta$ (which we computed using
  Lemma~\ref{lem:regular-bridges}) simultaneously and throw out all
  those that appear in both lists.  This leaves us with a list of
  semi-regular bridges for every node.  In addition to that, we know
  an attachment vertex for every semi-regular bridge that is a
  representative of the child-attachment of that bridge.  Thus, we
  also get the actual attachments in $\skel(\eta)$ in linear time
  using one bottom-up traversal as in the proof of
  Lemma~\ref{lem:regular-bridges}.  Moreover, we get the parent
  attachments using Lemma~\ref{lem:parent-and-pole-attachments}.
\end{proof}

\begin{lemma}
  \label{lem:parent-and-pole-bridges}
  The parent and pole bridges and their attachments together with a
  representative for each attachment can be computed in linear time
  when omitting dominated pole bridges and all but one parent bridge
  of each equivalent classes.

\end{lemma}
\begin{proof}
  Let $\mathcal T$ be the SPQR-vertex-tree of the common graph $G$.

  We now process $\mathcal T$ bottom-up computing a list of bridges
  for each node $\mu$.  This list will contain all bridges that
  potentially are parent or pole bridges in $\mu$ and we denote it by
  $\pot(\mu)$.  More precisely, if a bridge $B$ has an attachment
  vertex that is a descendent of $\mu$ and other attachment vertices
  that are not descendents of $\mu$, then $B$ is contained in
  $\pot(\mu)$.  As we do not have the time to tidy up properly,
  $\pot(\mu)$ can also contain bridges whose attachment vertices are
  all descendents of $\mu$.

  We start with the leaves of $\mathcal T$, which are vertices of $G$.
  Let $v$ be such a leaf.  Then $\pot(v)$ is the list of all bridges
  having $v$ as attachment.  By initializing $\pot(v)$ with the empty
  list and then processing all bridges once, we can compute $\pot(v)$
  for all vertices in $O\left(n + \sum_{i=1}^k|B_i|\right)$ time.  Now
  consider an inner node $\mu$.  We basically obtain the list
  $\pot(\mu)$ by concatenating the lists of all children.  But before
  concatenation we (partially) process these lists separately. 

  While processing $\mu$ (and the lists computed for the children of
  $\eta$) we want to answer for a given bridge $B$ the following
  queries in constant time.  First, is $B$ regular in $\mu$?  Second,
  have we seen $B$ already while processing $\mu$?  This can be done
  using timestamps.  Assume we have one global array with an entry for
  each bridge.  Before processing $\mu$ we increment a global
  timestamp and write this timestamp into the fields of the array
  corresponding to bridges that are regular in $\mu$.  While
  processing $\mu$, we can then check in constant time whether the
  current timestamp is set for a given bridge $B$ and thus whether $B$
  is regular.  Setting up this array takes time linear in the number
  of regular bridges of $\mu$ and thus we have an overall linear
  overhead.  We can handle the second query in constant time,
  analogously.

  Now let $\mu$ be the node we currently process, let $\eta$ be a
  child of $\mu$ and assume that $\pot(\eta)$ is already computed.  We
  process the list $\pot(\eta)$; let $B$ be the current bridge.  By
  Lemma~\ref{lem:parent-and-pole-attachments} we can check in constant
  time whether $B$ has pole or parent attachments in $\mu$.  If not,
  all attachment vertices of $B$ are descendents of $\mu$ and we
  remove $\mu$ form the list $\pot(\eta)$ as $B$ cannot be a parent or
  pole in $\mu$ or in any ancestor of $\mu$.

  Otherwise, $B$ has a parent or a pole attachment in $\mu$.  We first
  check (in constant time) whether $B$ is regular in $\mu$.  Assume it
  is regular and we see $B$ the first time since processing $\mu$
  (which we can also check in constant time as mentioned above).  Then
  we simply skip $B$ and continue processing $\pot(\eta)$.  If $B$ is
  regular in $\mu$ and $B$ occurred before while processing $\mu$, it
  would be contained twice in the list $\pot(\mu)$ when concatenating
  the lists of all children of $\mu$.  Thus we can remove this
  occurrence of $B$ from $\pot(\eta)$.  Afterwards, we continue
  processing the remaining bridges in $\pot(\eta)$.

  Now assume that $B$ is not regular.  Then $B$ is either a parent or
  a pole bridge.  If $B$ is a parent bridge, we store it as a parent
  bridge of $\mu$.  Note that we also know the attachments of $B$ in
  $\skel(\mu)$ (the parent edge and the attachment corresponding to
  the child $\eta$) and a representative for each of these
  attachments.  Afterwards, we stop processing $\pot(\eta)$ and
  continue with another child of $\mu$.  By stopping after processing
  $B$, we might miss a bridge $B'$ in $\pot(\eta)$ that is also a pole
  or parent bridge of $\mu$.  However, the child attachment of $B'$
  would be the attachment corresponding to $\eta$ and thus $B'$ is in
  the same equivalence class as $B$ (if $B'$ is a parent bridge) or
  $B'$ is dominated by $B$ (if $B'$ is a pole bridge).  In both cases
  we can omit $B'$.

  Finally, consider the case that $B$ is a pole bridge in $\mu$.  We
  save $B$ together with its attachments in $\skel(\mu)$ (and their
  representatives) as pole bridge of $\mu$.  Then we remove $B$ from
  $\pot(\eta)$ and continue processing $\pot(\eta)$.  Removing $B$
  from $\pot(\eta)$ has the effect that $B$ does not occur when
  processing an ancestors of $\mu$.  Thus, we have to show that $B$ is
  not a pole or parent bridge in one of these ancestors.  Consider an
  ancestor $\tau$ of $\mu$.  Then either all attachment vertices of
  $B$ are descendents of $\tau$ and $B$ is neither pole nor parent
  bridge for $\tau$.  Otherwise, the only attachment vertex of $B$
  that is not a descendent of $\tau$ is $s$ (without loss of
  generality).  However, the virtual edge in $\skel(\tau)$ containing
  all attachment vertices of $B$ is then incident to $s$ and thus the
  attachment $s$ is superfluous.  Hence, we can omit the bridge $B$ in
  $\skel(\tau)$ by Lemma~\ref{lem:superfluous-attachments}.

  It remains to show that the above procedure runs in linear time.
  First note that we add elements to lists $\pot(\cdot)$ only in the
  leaves of $\mathcal T$.  As we add each bridge $B$ to exactly $|B|$
  such bridges, the total size of these lists is linear.  Assume we
  are processing a list $\pot(\eta)$ and let $B$ be a bridge we delete
  after processing it.  Then we can ignore the (constant) running time
  for processing $B$, as we have only linearly many such deletion
  operations.  Otherwise, $B$ is a regular bridge that we see the
  first time or it is a parent bridge.  The former case happens only
  as many times as there are regular bridges of $\mu$ (which is
  overall linear).  The latter case happens at most once for each
  child of $\mu$ as we stop processing $\pot(\eta)$ afterwards.
  Hence, the overall running time is linear.
\end{proof}

\subsubsection{Linear Time Implementation of Steps~\ref{step:1}--\ref{step:5}}

\begin{lemma}
  \label{lem:step1-linear-time}
  Step~\ref{step:1} can be performed in linear time.
\end{lemma}
\begin{proof}
  Recall that Step~\ref{step:1} consist of computing the important
  union bridges for P- and R-nodes.  For every P- and R-node $\mu$ we
  then have to test whether $\skel(\mu)$ admits a compatible
  embedding, i.e., whether $\skel(\mu)$ can be embedded such that the
  attachments of each important bridge of $\mu$ share a face.

  For each node $\mu$, we compute the regular and semi-regular bridges
  using Lemma~\ref{lem:regular-bridges} and
  Lemma~\ref{lem:semi-regular-bridges}, respectively.  We moreover
  compute some of the pole and parent bridges using
  Lemma~\ref{lem:parent-and-pole-attachments}.  In this way we of
  course miss some important bridges but we know by
  Lemma~\ref{lem:omitting-bridges} that it is safe to do so.  Thus, we
  can focus on computing compatible embeddings.

  Let $\mu$ be a node of the SPQR-tree and assume that $\mu$ is a
  P-node.  Each of the two vertices $s$ and $t$ in $\skel(\mu)$ is
  incident to every face.  If an important bridge $B$ has $s$ or $t$
  as attachment, this attachment does not constrain the embedding of
  $\skel(\mu)$ (it shares a face with all other attachments of $B$ if
  and only if all other attachments share a face).  Thus, we can
  assume that only the virtual edges of $\skel(\mu)$ are attachments.
  If $B$ has three (or more) attachments in $\skel(\mu)$, it is
  impossible to find a compatible embedding as every face of
  $\skel(\mu)$ is incident to only two virtual edges.  It remains to
  deal with the case where every bridge has two virtual edges as
  attachment.  We build the conflict graph with one vertex $v(\eps)$
  for every virtual edge $\eps$ and an edge between two such vertices
  $v(\eps_1)$ and $v(\eps_2)$ if and only if there is a bridge with
  attachments $\eps_1$ and $\eps_2$.  It is not hard to see that
  $\skel(\mu)$ admits a compatible embedding if and only if this
  conflict graph has maximum degree~2 and either contains no cycle or
  is a Hamiltonian cycle.

  If $\mu$ is an R-node, its skeleton is triconnected and therefore
  has a fixed planar embedding.  To test whether the embedding of
  $\skel(\mu)$ is compatible, we need to check for every bridge $B$,
  whether there is a face incident to all its attachments.  We
  consider the graph $\skel'(\mu)$ obtained from $\skel(\mu)$ by
  subdividing every edge and inserting a vertex into every face that
  is connected to all incident vertices.  We denote the new vertex
  created by subdividing the virtual edge $\eps$ by $v(\eps)$ and the
  new vertex inserted into the face $f$ by $v(f)$.  For a vertex $v$
  that already existed in $\skel(\mu)$ we also write $v(v)$.  For a
  bridge $B$ with attachments $a_1, \dots, a_k$, we need to test
  whether there is a face $f$ such that for every pair $v(a_i)$ and
  $v(a_j)$ the path $v(a_i)v(f)v(a_j)$ is contained in $\skel'(\mu)$.
  To make sure that all paths of length~2 between $v(a_i)$ and
  $v(a_j)$ include a vertex $v(f)$ corresponding to a face $f$, we
  subdivide every edge twice except for those edges incident to a
  vertex $v(f)$ corresponding to a face.

  As $\skel'(\mu)$ is planar, we can use the data structure by Kowalik
  and Kurowski~\cite{kk-spqpgct-03} that can be computed in linear
  time and supports shortest path queries for constant distance in
  constant time.  More precisely, for any constant $d$, there exists a
  data structure that can test in $O(1)$ whether a pair of vertices is
  connected by a path of length $d$.  If so, a shortest path is
  returned.  We first rule out some easy cases.

  If there is an attachment $a_i$ that is a virtual edge in
  $\skel(\mu)$, the vertex $v(a_i)$ has only four neighbors in
  $\skel'(\mu)$.  Thus, we get the two faces incident to $a_i$ in
  constant time and can check in $O(k)$ time whether one of them is
  incident to every attachment of $B$.  We thus assume that all
  attachments are vertices.  If one of these vertices is adjacent to
  three or more others, there cannot be a compatible embedding.  Thus,
  we either find (in $O(k)$ time) a pair of non-adjacent attachments
  $a_i$ and $a_j$ or there are only two or three pairwise adjacent
  attachments.  As the latter case is easy, we can assume that we have
  non-adjacent attachments $a_i$ and $a_j$.  As $\skel(\mu)$ is
  triconnected, $a_i$ and $a_j$ are incident to at most one common
  face $f$.  Thus, there is only one path $v(a_i)v(f)v(a_j)$ of
  length~2 from $a_i$ to $a_j$ in $\skel'(\mu)$, which gives us $f$ in
  constant time.  Then it remains to check whether all other
  attachments are incident to $f$, which takes $O(k)$ time.  Doing
  this for every bridge yields a linear-time algorithm testing whether
  an R-node skeleton admits a compatible embedding.
\end{proof}

\begin{lemma}
  Step~\ref{step:2} can be performed in linear time.
\end{lemma}
\begin{proof}
  Recall that Step~\ref{step:2} consist of three parts.  First we have
  to compute the important $\circled i$-bridges of S-nodes together
  with their attachments.  For each attachment we then have to test
  whether it is left- or right-sided.  Finally, we have to add the
  consistency constraints.

  As for Step~\ref{step:1}, we use Lemma~\ref{lem:regular-bridges},
  Lemma~\ref{lem:semi-regular-bridges}, and
  Lemma~\ref{lem:parent-and-pole-bridges} to compute all important
  $\circled i$-bridges together with their attachments.  We actually
  compute these important bridges not only for the S-nodes but also
  for P- and R-nodes.  We use this additional information to compute
  which attachments are left- and which are right-sided.

  Note that the final step of adding the consistency constraints to a
  global 2-\textsc{Sat} formula is trivial.  Thus, it remains to show
  that we can compute for every attachment whether it is left- or
  right-sided.

  Let $\mu$ be a node of the SPQR-tree $\mathcal T$.  We iterate over
  all important bridges of $\mu$ (except those we omitted).  For every
  bridge $B$ we iterate over all attachments in $\mu$ and for every
  virtual edge $\eps$ among those attachments, we append $B$ to the
  list $\bridges(\eps)$.  Afterwards, $\bridges(\eps)$ contains all
  important (but not omitted) bridges of $\mu$ that have $\eps$ as
  attachment.  Note that we can assume that the bridges in
  $\bridges(\eps)$ are sorted according to an arbitrary but fixed
  order of the bridges.

  Let $\mu$ be an S-node with virtual edge $\eps$ in $\skel(\mu)$.
  Let further $\mu'$ be the neighboring P- or R-node corresponding to
  $\eps$ and let $\eps'$ be the virtual edge in $\skel(\mu')$
  corresponding to the S-node $\mu$.  For every bridge $B$ in
  $\bridges(\eps)$ we want to know whether the attachment $\eps$ is
  left- or right-sided.  To this end we iterate over the lists
  $\bridges(\eps)$ and $\bridges(\eps')$ simultaneously.  For every
  bridge $B$ in $\bridges(\eps)$ there are two different cases.
  Either $B$ also occurs in $\bridges(\eps')$ or it does not.  It is
  not hard to see that the latter can only happen if $B$ is in $\mu'$
  a pole or parent bridge that was omitted.

  If $B$ also occurs in $\bridges(\eps')$, we know from
  Step~\ref{step:1} that $\skel(\mu')$ has a (unique) face incident to
  all attachments of $B$ in $\skel(\mu')$.  In particular, this face
  is either the right or the left face of $\eps'$ and thus we
  immediately know whether the attachment $\eps$ of $B$ in
  $\skel(\mu)$ is left- or right-sided.

  It remains to consider the case where $B$ occurs in $\bridges(\eps)$
  but not in $\bridges(\eps')$.  If $\mu'$ is the parent of $\mu$, the
  bridge $B$ must be a pole- or parent bridge in $\mu'$ with child
  attachment~$\eps'$.  As in the proof of
  Lemma~\ref{lem:step1-linear-time}, we can find the (unique) face
  incident to $\eps'$ and one of the poles.  As $B$ has to lie in this
  face we know whether it lies to the right or to the left face of
  $\eps'$ in $\skel(\mu')$ and thus we know whether the attachment
  $\eps$ in $\skel(\mu)$ is left- or right-sided.

  If $\mu'$ is a child of $\mu$, then $B$ cannot be regular in $\mu$
  as otherwise $B$ would be semi-regular in $\mu'$ and thus contained
  in $\bridges(\eps')$ (which is not the case we consider).  Hence $B$
  is either a pole or a parent bridge (recall that semi-regular
  bridges are also parent bridges).  Thus, $B$ is a pole or parent
  bridge in $\mu$ with child attachment $\eps$.  By
  Lemma~\ref{lem:omitting-bridges}, we can omit all but a constant
  number of such bridges with $\eps$ as child attachment.  Thus, we
  can assume that $\bridges(\eps)$ contains only a constant number of
  bridges that do not occur in $\bridges(\eps')$.  For these bridges
  we allow a running time linear in the size of $\skel(\mu')$.  As
  $\mu$ is the unique parent of $\mu'$ this happens only a constant
  number of times for $\mu'$ and thus takes overall linear time.

  Let $B$ be such a bridge and let $v$ be the attachment vertex of the
  common graph $G$ representing the child attachment $\eps$ of $B$ in
  $\skel(\mu)$.  We show how to find a child attachment of $B$ in
  $\skel(\mu')$ in $O(|\skel(\mu')|)$ time.  Then we can (as in the
  cases before) find in constant time which face incident to $\eps'$
  contains $B$ in $\skel(\mu')$ and we are done.  As before we assume
  to have a DFS-ordering on the leaves of the SPQR-vertex-tree.  Then
  the leaves that a descendents of an inner node form an interval with
  respect to this order.  These intervals can be easily computed in
  linear time by processing the SPQR-vertex-tree bottom up once.
  Afterwards, we can check in constant time whether a vertex $v$ is
  the descendent of an inner node.  Hence we can check in
  $O(|\skel(\mu')|)$ time which child of $\mu'$ is an ancestor of $v$
  and thus which virtual edge or vertex in $\skel(\mu')$ represents
  $v$.  This yields the child attachment of $B$ in $\skel(\mu')$ and
  we are done.
\end{proof}

\begin{lemma}
  Step~\ref{step:3} can be performed in linear time.
\end{lemma}
\begin{proof}
  Let $\mu$ be an S-node.  Note that every important bridge in $\mu$
  may alternate with a linear number of important bridges.  Thus,
  there are instances where the planarity constraints have quadratic
  size.  In the following we describe how to compute constraints that
  are equivalent to the planarity constraints but have linear size.
  We only consider $\circled{1}$-bridges; for $\circled{2}$-bridges,
  the same procedure can be applied.

  We define the graph $H$ as follows.  We start with $H = \skel(\mu)$
  and subdivide every edge once.  Thus, $H$ has a vertex $v(a)$ for
  each attachment $a$ in $\skel(\mu)$.  For every bridge $B$ with
  attachments $a_1, \dots, a_k$, we add a \emph{bridge vertex} $v(B)$
  and connect it to the vertices $v(a_1), \dots, v(a_k)$.  When using
  the term \emph{cycle of $H$}, we refer to the subgraph of $H$ one
  obtains by removing the bridge vertices.

  Assume we have a planar embedding of $H$.  Then, every bridge vertex
  lies on one of the two sides of the cycle and no two bridges on the
  same side of the cycle alternate.  Conversely, an assignment of the
  bridges to the two sides of the cycles such that no two bridges
  alternate yields a planar embedding.  The choices the planarity
  constraints leave are thus equivalent to the embedding choices of
  $H$.

  As $H$ is biconnected, the embedding choices consist of reordering
  parallel edges in P-nodes and mirroring R-nodes of the SPQR-tree
  $\mathcal T_H$ of $H$.  Let $\eta$ be a P-node in $\mathcal T_H$.
  If the embedding of $\skel(\eta)$ determines on which side of the
  cycle a bridge $B$ lies, then $B$ is clearly not alternating with
  any other bridge.  For an R-node $\eta$ of $\mathcal T_H$, fixing
  the embedding of $\skel(\eta)$ to one of the two flips determines
  the side for some of the bridges.  We create a new binary variable
  $x_\eta$ with the interpretation that $x_\eta = 0$ if $\skel(\eta)$
  is embedded according to a reference embedding and $x_\eta = 1$ if
  the embedding is flipped.  For a bridge $B$ whose side is determined
  by the embedding of $\skel(\eta)$ we can then add the constraint
  $x_\eta = x_B^\mu$ or $x_\eta \not= x_B^\mu$ (depending on whether
  the reference embedding of $\skel(\eta)$ fixes $B$ to the left or
  right side of the cycle).

  It is not hard to see that one can compute for each R-node $\eta$
  the brides whose side is determined by the embedding of
  $\skel(\eta)$ in overall linear time in the size of $H$.  Thus, we
  get the above constraints (which are equivalent to the planarity
  constraints) in $O(|H|)$ time.  Clearly, the size of $H$ is linear
  in the size of $\skel(\mu)$ plus the total number of attachments of
  important bridges in $\mu$.  Thus, we get an overall linear running
  time.
\end{proof}

\begin{lemma}
  Step~\ref{step:4} can be performed in linear time.
\end{lemma}
\begin{proof}
  To add the union-bridge constraints, we have to group the
  $\circled{i}$-bridges that are important in an S-node $\mu$
  according to their union bridges.  To this end, we once create a
  global array $A$ with one entry $A[B']$ for each union bridge $B'$
  (which we can access in constant time).  Consider an
  $\circled{i}$-bridge $B$ that is important in $\mu$ (and was not
  omitted).  Let $B'$ be the union bridge containing $B$ (we can get
  $B'$ in constant time as every $\circled{i}$-bridge is contained in
  only one union bridge).  If the entry $A[B']$ was not modified while
  processing $\mu$ so far, we clear $A[B']$ (which might contain
  something from previous nodes) and set $A[B']$ to be a list
  containing only $B$.  If $A[B']$ was already modified, we append $B$
  to the list $A[B']$.  We can keep track of which entries of $A$ were
  already modified by using timestamps.

  For every union bridge $B'$ that contains an important
  $\circled{i}$-bridge, the entry $A[B']$ holds a list of all
  important $\circled{i}$-bridges of $\mu$ that belong to $B'$.  For
  each of these lists we add the union-bridge constraint for every
  pair of consecutive $\circled{i}$-bridges.  The transitivity
  enforces all pairwise union-bridge constraints (although not
  explicitly stated).  Clearly, this procedure takes linear time in
  the number of important $\circled{i}$-bridges.
\end{proof}

For Step~\ref{step:5}, assume the embedding of the biconnected common
graph $G$ is fixed.  We have to test whether this embedding of $G$ can
be extended to a \textsc{Sefe} that satisfies the union-bridge
constraints.  Thus, we basically have to assign each union bridge to a
face of $G$ such that no two $\circled{1}$-bridges and no two
$\circled{2}$-bridges alternate.

We first distinguish three different types of union bridges.  Let $B$
be a union bridge.  A face $f$ of $G$ is \textsc{feasible} for $B$ if
all attachment vertices of $B$ are incident to $f$.  We say that $B$
is \emph{flexible} if it has at least three feasible faces.  We say
that $B$ is \emph{binary} if $B$ has two feasible faces and
\emph{fixed} if it has only one feasible face.  If there is a bridge
that has no feasible face then the instance is obviously not solvable.

The overall strategy for Step~\ref{step:5} is the following.  We first
determine which union bridges are flexible, binary, and fixed,
respectively.  We first assign the fixed union bridges to their faces.
For the binary union bridges, we can encode the decision for one of
the two possible faces with a binary variable.  For two union bridges
with a common feasible face that include alternating
$\circled{i}$-bridges (for the same $i \in \{1, 2\}$), we have to make
sure that they are not embedded into the same face.  Note that these
kind of conditions are very similar to the planarity constraints we
had in Step~\ref{step:3}, which again leads to a \textsc{2-Sat}
formula.  Any solution of this formula induces an assignment of the
binary union bridges to faces.  Finally, we check whether the flexible
union bridges can be added.  For this to work we have to show that
this final step of assigning the flexible bridges is independent from
the solution we chose for the \textsc{2-Sat} formula.  We obtain the
following lemma.

\begin{lemma}
  Step~\ref{step:5} can be performed in linear time.
\end{lemma}
\begin{proof}
  Let $B$ be a flexible union bridge and assume all binary and fixed
  union bridges are already assigned to faces.  We first show the
  following.  If $B$ cannot be embedded into one of its feasible faces
  (due to alternating $\circled{i}$-bridges), then $B$ cannot be
  embedded into this face even when omitting all binary union bridges.
  This shows that the above strategy of first assigning the fixed,
  then the binary, and finally the flexible union bridges to faces is
  correct.  Afterwards we show how to do it in linear time.

  As the union bridge $B$ is flexible, it has at least three feasible
  faces.  As the common graph $G$ is biconnected, $B$ can have only
  two attachment vertices; let $u$ and $v$ be these attachment
  vertices.  Moreover, $u$ and $v$ must be the poles of a P-node $\mu$
  of the SPQR-tree of $G$.  Let $\eps_1, \dots, \eps_k$ be the virtual
  edges of $\skel(\mu)$ appearing in this order and let $f_i$ be the
  face between $\eps_i$ and $\eps_{i + 1}$ (subscripts are considered
  modulo $k$).  The feasible faces of $B$ are exactly the faces $f_1,
  \dots, f_k$.

  Assume $B$ cannot be assigned to the feasible face $f_i$, i.e., an
  $\circled{i}$-bridge contained in $B$ alternates with a
  $\circled{i}$-bridge belonging to another union bridge $B'$ that was
  assigned to $f_i$.  As $u$ and $v$ are the only attachment vertices
  of $B$, $B'$ must have attachment vertices $u'$ and $v'$ with $u',
  v' \notin\{u, v\}$ that belong to the expansion graphs of $\eps_i$
  and $\eps_{i+1}$, respectively.  Then $u'$ and $v'$ can share only a
  single face, namely $f_i$, and thus $B'$ is a fixed union bridge.
  This shows the above claim and thus it remains to show how to
  implement the procedure in linear time.

  Let $B$ be an arbitrary union bridge.  We show how to detect whether
  $B$ is flexible, binary, or fixed in $O(|B|)$ time.  For $B$ to be
  flexible, it must have only two different attachment vertices.  This
  can be easily tested in $O(|B|)$ time.  If $B$ has only two
  attachment vertices $u$ and $v$, we need to test whether $u$ and $v$
  are the poles of a P-node in the SPQR-tree of $G$.  We show how this
  can be done in constant time.  To this end, let $\mathcal T$ be the
  SPQR-vertex-tree of $G$.  Assume $\mu$ is a P-node with poles $u$
  and $v$.  Let $\eta$ be the parent of $\mu$.  Then $\skel(\eta)$
  contains both vertices $u$ and $v$ but at most one of them as pole.
  Assume without loss of generality that $u$ is not a pole of
  $\skel(\eta)$.  Then the leaf $u$ of $\mathcal T$ is a child of
  $\eta$ and thus we find $\eta$ in constant time (together with the
  vertex $u$ in $\skel(\eta)$).  If $v$ is not a pole of
  $\skel(\eta)$, we find $v$ in $\skel(\eta)$ in the same way,
  otherwise $v$ is a pole and we also get $v$ in $\skel(\eta)$ in
  constant time (by checking both poles).  As $\skel(\eta)$ is a
  planar graph we can get the virtual edge between the two vertices
  $u$ and $v$ in constant time via a shortest-path data
  structure~\cite{kk-spqpgct-03}.  Thus, we also find the
  corresponding child $\mu$ having $u$ and $v$ as poles.  Hence, we
  can either find the P-node $\mu$ with poles $u$ and $v$ in constant
  time (implying that $B$ is flexible) or we can conclude that such a
  P-node does not exists (implying that $B$ is not flexible).

  Next we determine whether $B$ is binary or fixed.  First note that
  the bridge $B$ (that is not flexible) is binary if and only if there
  exists an S-node $\mu$ such that every attachment vertex of $B$ is a
  vertex in $\skel(\mu)$.  This can be tested in $O(|B|)$ time using
  the SPQR-vertex-tree.  Assume $\mu$ is the S-node such that every
  attachment vertex of $B$ is a vertex in $\skel(\mu)$.  We can handle
  the case where $B$ has only the two poles of $\skel(\mu)$ as
  attachment vertex analogously to the case above (about flexible
  bridges) except that $\mu$ is an S-node instead of a P-node.  Thus,
  assume that $v$ is an attachment vertex of $B$ that is not a pole of
  $\skel(\mu)$.  Then we can find $\mu$ in constant time as it is the
  parent of the leaf $v$ in $\mathcal T$.  Every other attachment
  vertex in $B$ is either also a child of $\mu$ or a pole of
  $\skel(\mu)$, which we can check in constant time per attachment
  vertex.  Thus, we can detect in $O(|B|)$ time whether $B$ is binary
  or fixed.

  At this point we know which bridges are binary and which are
  flexible.  All remaining bridges are either fixed or do not have a
  feasible face at all, which implies that there is no \textsc{Sefe}.
  We show for such a bridge $B$ how we can assign it to its unique
  feasible face or decide that such a face does not exist.  Recall
  from Lemma~\ref{lem:step1-linear-time} that we can compute in
  constant time a face that is shared by a given pair of vertices (or
  conclude that such a face does not exist).  If $B$ has only two
  attachment vertices $u$ and $v$, then we can either find the unique
  feasible face of $B$ or decide that $B$ has no feasible face in
  constant time.

  We can thus assume that $B$ has at least three attachment vertices.
  Let $u$, $v$, and $w$ be three attachment vertices of $B$.  In
  constant time, we find a face $f_{u, v}$ that is incident to $u$ and
  $v$.  Analogously we find faces $f_{u, w}$ and $f_{v, w}$.  There
  are two different cases.  If there is a pair of vertices among $u$,
  $v$, and $w$ that shares only a single face, then one of the faces
  $f_{u, v}$, $f_{u, w}$, or $f_{u, w}$ is the only possible feasible
  face of $B$.  We can check that in $O(|B|)$ time.  Otherwise, assume
  there is a face $f \notin \{f_{u, v}, f_{u, w}, f_{u, w}\}$ that is
  feasible for $B$.  Then $u$ and $v$ are commonly incident to at
  least two different faces (namely $f$ and $f_{u, v}$) and thus $\{u,
  v\}$ is a separating pair of an edge.  The same holds for $u$ and
  $w$ and for $v$ and $w$.  In this case there must exist a node $\mu$
  in the SPQR-tree of $G$ such that $\skel(\mu)$ contains the triangle
  $u,v,w$.  Note that we can find this node as we did before for the
  flexible and binary bridges. 

  If $\mu$ is an S-node, $B$ must contain another attachment vertex
  $x$ (otherwise $B$ is binary).  Then $x$ is contained in the
  expansion graph of one of the three virtual edges in $\skel(\mu)$.
  Assume without loss of generality that $x$ belongs to the expansion
  graph of $uv$.  Then $w$ and $x$ share only a single face (otherwise
  $w$ and $x$ would be a separating pair which contradicts the fact
  that $\skel(\mu)$ is a triangle).  Thus, we can find the desired
  face $f$ by finding a common face of $x$ and $w$ in constant time.
  Of course one then needs to check if this face is actually incident
  to each attachment vertex in $B$ ($B$ has no feasibly face if not).

  It remains to consider the case that $\mu$ is an R-node.  First test
  whether the triangle $u, v, w$ forms a face in $\skel(\mu)$.  If so,
  this face is unique and thus we know the only potentially feasible
  face of $B$.  Note that this gives us only the face in $\skel(\mu)$.
  However, one can easily compute a mapping from the faces of
  skeletons to the faces in the actual graph in linear time in the
  size of $G$ (this has to be done only once for all bridges).

  To conclude, we now ensured that every union bridge that is neither
  flexible nor binary is fixed and we assigned the fixed union bridges
  to their unique feasible faces.  Let us continue with the binary
  bridges.  For every binary bridge $B$ we already computed the S-node
  $\mu$ containing all the attachment vertices of $B$.  Note that this
  already gives us the two possible faces in which $B$ can be embedded
  (of course we again have to translate from faces in a skeleton to
  faces in $G$).

  When assigning the binary bridges to faces, we have to make sure
  that no two $\circled{i}$-bridges alternate.  This can be ensured
  using a \textsc{2-Sat} formula as for Step~\ref{step:3}.  As before
  we can compute and solve this \textsc{2-Sat} formula in linear time.
  Thus, it remains to add the flexible bridges.

  Let $\mu$ be a P-node of the SPQR-tree of $G$ and let $s$ and $t$ be
  the poles of $\skel(\mu)$.  Let further $\eps_1, \dots, \eps_k$ be
  the virtual edges of $\skel(\mu)$ and let the faces $f_1, \dots,
  f_k$ be defined as before.  Assume we still know the important
  bridges for $\mu$ from the previous steps.  Assume the union
  flexible bridge $B$ contains only $\circled{1}$-bridges.  Clearly,
  we can embed $B$ into $f_i$ if and only if there is no
  $\circled{1}$-bridge with attachments $\eps_i$ and $\eps_{i+1}$.
  The analogous statement holds if $B$ contains only
  $\circled{2}$-bridges.  If $B$ contains both, $\circled{1}$- and
  $\circled{2}$-bridges, we can embed it into $f_i$ if and only if
  there is neither a $\circled{1}$ nor a $\circled{2}$ bridge with
  attachments $\eps_i$ and $\eps_{i+1}$.  Thus, we can check in
  $O(|\skel(\mu)|)$ time for an appropriate face for $B$.  Note that
  we cannot afford this amount of time for every flexible bridge with
  attachments $s$ and $t$.  However, consider two such bridges $B$ and
  $B'$ as equivalent in the sense that $B$ contains $\circled{1}$- and
  $\circled{2}$-bridges if and only if $B'$ contains $\circled{1}$-
  and $\circled{2}$-bridges, respectively.  Then $B$ and $B'$ can be
  assigned to the same face.  Thus, we have to spend $O(|\skel(\mu)|)$
  time only a constant number of times for each P-node.  This
  concludes the proof.
\end{proof}

\begin{theorem}
  \textsc{Sefe} with union bridge constraints can be solved in linear
  time if the common graph is biconnected.
\end{theorem}

\section{Edge Orderings and Relative Positions}
\label{sec:graph-with-common}

In this section, we consider a \textsc{Sefe} instance $(\1G, \2G)$
that has common P-node degree~3 and simultaneous cutvertices of common
degree at most~3.  Recall that $(\1G, \2G)$ admits a simultaneous
embedding if and only if $\1G$ and $\2G$ admit planar embeddings that
have consistent edge orderings and consistent relative positions on
the common graph.  We show how to address both requirements (more or
less) separately, by formulating necessary and sufficient constraints
using equations and inequalities on Boolean variables.  Moreover, we
show how to incorporate equations and inequalities equivalent to
block-local common-face constraints.  Together with the preprocessing
algorithms from the previous sections, this leads to a polynomial time
algorithm for instances with P-node degree~3 and simultaneous
cutvertices of common degree at most~3.

Before we can follow this strategy, we need to address one problem.
The relative position of a component $H'$ of $G$ with respect to
another connected component $H$, denoted by $\pos_H(H')$, is the face
of $H$ containing $H'$.  However, the set of faces of $H$ depends on
the embedding of $H$.  To be able to handle relative positions
independently from edge orderings, we need to express the relative
positions independently from faces.

This is done in the following section.  Afterwards, we show how to
enforce consistent edge orderings (Section~\ref{sec:cons-edge-order}),
block-local common-face constraints
(Section~\ref{sec:common-face-constr}), and consistent relative
positions (Section~\ref{sec:cons-relat-posit}).  Finally, we conclude
in Section~\ref{sec:putt-things-togeth}.

Before we start, we need one more definition.  Assume we have a set
$X$ of binary variables such that every variable $x \in X$ corresponds
to a binary embedding choice in a given graph $G$.  Let $\alpha \colon
X \to \{0, 1\}$ be a \emph{variable assignment}.  We say that an
embedding of $G$ \emph{realizes} the assignment $\alpha$, if the
embedding decision in $G$ corresponding to a variable $x \in X$ fits
to the value $\alpha(x)$.  Note that not every variable assignment can
be realized as the embedding choices can depend on each other.

\subsection{Relative Positions with Respect to a Cycle Basis}
\label{sec:relat-posit-with}

A \emph{generalized cycle} $C$ in a graph $H$ is a subset of its edges
such that every vertex of $H$ has an even number of incident edges in
$C$.  The \emph{sum} $C \oplus C'$ of two generalized cycles is the
symmetric difference between the edge sets, i.e., an edge $e$ is
contained in $C \oplus C'$ if and only if it is contained in $C$ or in
$C'$ but not in both.  The resulting edge set $C \oplus C'$ is again a
generalized cycle.  The set of all generalized cycles in $H$ is a
vector space over $\mathbb F_2$.  A basis of this vector space is
called \emph{cycle basis} of $H$.

Instead of considering the relative position $\pos_H(H')$ of a
connected component~$H'$ with respect to another component $H$, we
choose a cycle basis $\mathcal C$ of $H$ and show that the relative
positions of $H'$ with respect to the cycles in $\mathcal C$ suffice
to uniquely define $\pos_H(H')$, independent from the embedding of
$H$.  We assume $H$ to be biconnected.  All results can be extended to
connected graphs by using a cycle basis for each block.

Let $C_0, \dots, C_k$ be the set of facial cycles with respect to an
arbitrary planar embedding of $H$.  The set $\mathcal C = \{C_1,
\dots, C_k\}$ obtained by removing one of the facial cycles is a cycle
basis of $G$.  A cycle basis that can be obtained in this way is
called \emph{planar cycle basis}.  In the following we assume all
cycle bases to be planar cycle bases.  Moreover, we consider all
cycles to have an arbitrary but fixed orientation.  The binary
variable $\pos_C(p)$ represents the relative position of a point $p$
with respect to a cycle $C$, where $\pos_C(p) = 0$ and $\pos_C(p) = 1$
have the interpretation that $p$ lies to the right and left of $C$,
respectively.

\begin{theorem}
  \label{thm:cycle-basis}
  Let $H$ be a planar graph embedded on the sphere, let $p$ be a point
  on the sphere, and let $\mathcal C = \{C_1, \dots, C_k\}$ be an
  arbitrary planar cycle basis of $H$.  Then the face containing $p$
  is determined by the relative positions $\pos_{C_i}(p)$ for $1 \le i
  \le k$.
\end{theorem}
\begin{proof}
  Let $f$ be a face and let $C$ be the corresponding facial cycle.  We
  assume without loss of generality that $C = C_1 \oplus \dots \oplus
  C_\ell$ holds.  We show that the point $p$ belongs to the face $f$
  if and only if $\pos_{C_i}(p) = \pos_{C_i}(f)$ holds for $1 \le i
  \le \ell$.  Obviously, if $\pos_{C_i}(p) \not= \pos_{C_i}(f)$ holds
  for one of the basis cycles $C_i$, then $C_i$ separates $p$ from $f$
  and thus $p$ cannot belong to~$f$.

  Conversely, we have to show that there is no point lying on the same
  sides of the cycles $C_i$ for $1 \le i \le \ell$ not belonging to
  $f$.  To this end we define the \emph{position vector} $\pos(p) =
  (\pos_{C_1}(p), \dots, \pos_{C_\ell}(p))$ of a point $p$.  We show
  that there is not point outside $f$ having the same position vector
  as the points inside $f$.  Consider how the position vector of a
  point $p$ changes when moving it around.  First, all points inside
  $f$ have the same position vector.  Second, when $p$ does not lie in
  $f$ and crosses an edge $e$ while moving it, then $e$ is either
  contained in zero or in two of the basis cycles $C_1, \dots,
  C_\ell$.  This comes from the facts that $C = C_1 \oplus \dots
  \oplus C_\ell$ holds, that $e$ is not contained in $C$ and that our
  cycle basis is planar.  In the former case the position vector does
  not change, in the latter case exactly two values toggle.  No matter
  which case applies, the parity of the number of entries with the
  value \textsc{left} does not change, that is this number either
  remains odd or even.  Thus, this parity is the same for all points
  outside of $f$.  Finally, when $p$ moves from inside $f$ to the
  outside of $f$ (or the other way round), it has to cross an edge $e$
  contained in the cycle $C$.  Since $e$ has to be contained in
  exactly one of the cycles $C_1, \dots, C_\ell$, exactly one entry in
  the position vector $\pos(p)$ changes from \textsc{left} to
  \textsc{right} or vice versa.  Thus the parity of the number of
  entries with the value \textsc{left} changes.  It follows, that for
  every point not contained in $f$ this parity differs from the parity
  of the points in $f$.  Thus also the position vector must differ,
  which concludes the proof.
\end{proof}

To represent the relative position of one connected component $H$ with
respect to another connected component $H'$, it thus suffices to
consider the relative positions of $H$ with respect to cycles in a
cycle basis of $H'$.  However, there is one case for which we have a
slightly stronger requirement.  To motivate this, consider the
following example; see also Figure~\ref{fig:cycle-basis-not-suff}.

\begin{figure}
  \centering
  \includegraphics[page=1]{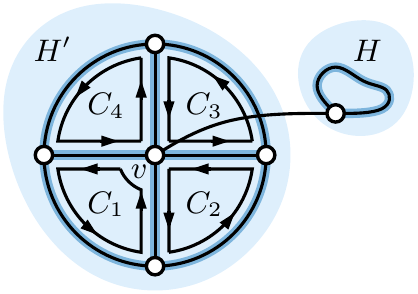}
  \caption{The exclusive edge connecting $H$ to the vertex $v$ of $H'$
    requires $H$ to lie to the left of exactly one of the cycles $C_1,
    \dots, C_4$ (or to the right of exactly one cycle if the embedding
    of $H'$ is flipped).  This type of constraint cannot be expressed
    using only equations or inequalities.  If, however, the cycle
    basis contained the facial cycle of the outer face of $H'$, it
    would be sufficient to require that $H$ and $v$ lie on the same
    side of this cycle, which can be expressed using an equation.}
  \label{fig:cycle-basis-not-suff}
\end{figure}

Consider the graph $\1G$ containing the common graph $G$ with
connected components $H$ and $H'$.  Let $\mathcal C$ be a cycle basis
of $H'$.  Let further $v$ be a vertex of $H'$ that is cutvertex of
$\1G$ separating $H$ from $H'$ and let $C \in \mathcal C$.  If $v$
lies to the right of~$C$ in a given embedding of $\1G$, then $H$ also
lies to the right of~$C$.  Conversely, if $v$ lies to the left of $C$,
then $H$ lies to the left of~$C$.  However, requiring for every cycle
$C \in \mathcal C$ (that does not contain $v$) that $v$ and $H$ lie on
the same side of $C$ does not ensure that $H$ lies in a face of $H'$
that is incident to $v$.  Figure~\ref{fig:cycle-basis-not-suff} shows a
somewhat degenerate example, were $v$ is contained in every cycle of
$\mathcal C$.

Thus, the relative positions of $v$ with respect to all cycles in
$\mathcal C$ (that do not contain $v$) do not uniquely determine a
face of $H' - v$.  To resolve this issue, we add further cycles of
$H'$ to $\mathcal C$.  More precisely, an \emph{extended cycle basis}
of $H'$ is a set of cycles $\mathcal C$ in $H'$ such that $\mathcal C$
includes a cycle basis of $H'$ and a cycle basis of $H' - v$ for every
vertex $v$ of $H'$.

Note that one can for example obtain an extended cycle basis of $H'$
as follows.  First choose an embedding of $H'$ and start with the
corresponding planar cycle basis for $\mathcal C$.  For every vertex
$v$, consider the induced embedding of $H' - v$ and add to $\mathcal
C$ all cycles in the corresponding planar cycle basis of $H' - v$ that
are not already contained in $\mathcal C$.  It directly follows that
an extended cycle basis has $O(n^2)$ size and can be computed in
$O(n^2)$ time.  Moreover, we get the following lemma.

\begin{lemma}
  \label{lem:extended-cycle-basis}
  Let $H$ be an embedded planar graph and let $\mathcal C$ be an
  extended cycle basis of $H$.  If a vertex $v$ of $H$ is not incident
  to $f$ of $H$, then $\mathcal C$ contains a cycle $C$ (not
  containing $v$) such that $\pos_C(v) \not= \pos_C(f)$.
\end{lemma}
\begin{proof}
  Consider the graph $H - v$ together with the embedding induced by
  the embedding of $H$.  Let $f_v$ be the face of $H - v$ that
  contains $v$ in the embedding of $H$.  As $v$ is not incident to
  $f$, we have $f_v \not= f$.  By Theorem~\ref{thm:cycle-basis}, the
  cycle basis of $H - v$ (and thus $\mathcal C$) must contain a cycle
  $C$ such that $\pos_C(f_v) \not= \pos_C(f)$.
\end{proof}

If we refer to a cycle basis in one of the following sections, we
always assume to actually have an extended cycle basis.

\subsection{Consistent Edge Orderings}
\label{sec:cons-edge-order}

We first assume that the graphs $\1G$ and $\2G$ are biconnected and
then show how to extend our approach to exclusive cutvertices and
simultaneous cutvertices of common degree~3.

\subsubsection{Biconnected Graphs}
\label{sec:biconnected-graphs}

Let $\1G$ and $\2G$ be biconnected planar graphs.  There exists an
instance of \textsc{Simultaneous PQ-Ordering} that has a solution if and
only if $\1G$ and $\2G$ admit embeddings with consistent edge
ordering~\cite{br-spoacep-13}.  This solution is based on the
\emph{PQ-embedding representation}, an instance of \textsc{Simultaneous
  PQ-Ordering} representing all embeddings of a biconnected planar
graph.  We describe this embedding representation and show how to
simplify it for instances that have common P-node degree~3.

For each vertex $\1v$ of $\1G$, the PQ-embedding representation,
denoted by $D(\1G)$, contains the \emph{embedding tree} $T(\1v)$
having a leaf for each edge incident to $\1v$, representing all
possible orders of edges around $\1v$.  For every P-node $\1\mu$ in
the SPQR-tree $\T(\1G)$ that contains $\1v$ in $\skel(\1\mu)$ there is
a P-node in $T(\1v)$ representing the choice to order the virtual
edges in $\skel(\1\mu)$.  Similarly, for every R-node $\1\mu$ of $\1G$
containing $\1v$ in its skeleton, there is a Q-node in $T(\1v)$ whose
flip corresponds to the flip of $\skel(\1\mu)$.  As the orders of
edges around different vertices of $\1G$ cannot be chosen
independently from each other, so called \emph{consistency trees} are
added as common children to enforce Q-nodes stemming from the same
R-node in $\T(\1G)$ to have the same flip and P-nodes stemming from
the same P-node to have consistent (i.e., opposite) orders.  Every
solution of the resulting instance corresponds to a planar embedding
of $\1G$ and vice versa~\cite{br-spoacep-13}.

As we are only interested in the order of common edges, we modify
$D(\1G)$ by projecting each PQ-tree to the leaves representing common
edges.  As $\1G$ and $\2G$ have common P-node degree~3, all P-nodes of
the resulting PQ-trees have degree~3 and can be assumed to be Q-nodes
representing a binary decision.  We call the resulting instance
\emph{Q-embedding representation} and denote it by $D(\1G)$.  

Let $\1\mu$ be an R-node of the SPQR-tree $\T(\1G)$ whose embedding
influences the ordering of common edges around a vertex.  Then the
Q-embedding representation contains a consistency tree consisting of a
single Q-node representing the flip of $\skel(\1\mu)$.  We associate
the binary variable $\ord(\1\mu)$ with this decision.

For a P-node $\1\mu$ we get a similar result.  Let $\1u$ and $\1v$ be
the poles of $\1\mu$.  If the consistency tree enforcing a consistent
decision in the embedding trees $T(\1u)$ and $T(\1v)$ has degree~3,
its flip represents the embedding decision for $\skel(\1\mu)$ and we
again get a binary variable $\ord(\1\mu)$.  Otherwise, this
consistency tree contains two or less leaves and can be ignored.  Then
the choices for the Q-nodes corresponding to $\1\mu$ in $T(\1u)$ and
$T(\1v)$ are independent and we get one binary variable for each of
these Q-nodes.  We denote these variables by $\ord(\1{\mu_u})$ and
$\ord(\1{\mu_v})$.

We call these variables we get for $\1G$ the
\emph{$\circled{1}$-PR-ordering variables}.  The
\emph{$\circled{2}$-PR-ordering} variables are defined analogously.
The \emph{PR-ordering variables} are the union of these two variable
sets.  Let $\ii X$ be the $\circled{i}$-PR-ordering variables and
assume we have a fixed variable assignment $\ii\alpha$.  Then this
variable assignment already determines all edge orderings of the
common graph, i.e., every embedding of $\ii G$ realizing the
assignment $\ii\alpha$ induces the same edge orderings on $G$.  In the
following we describe a set of necessary equations and inequalities on
the PR-ordering variables that ensure that $\1G$ and $\2G$ induce the
same edge orderings on $G$.

For a common vertex $v$ occurring as $\1v$ and $\2v$ in $\1G$ and
$\2G$, respectively, we add a so-called \emph{common embedding tree}
$T(v)$ (consisting of a single P-node) as child of the embedding trees
$T(\1v)$ and $T(\2v)$ in the Q-embedding representations of $\1G$ and
$\2G$; see Figure~\ref{fig:sim-pq-ord-instances}.  Obviously, this
common child ensures that the common edges around $v$ are ordered the
same with respect to $\1G$ and $\2G$.  Adding the common embedding
tree for every common vertex yields the instance $D(\1G, \2G)$.

\begin{figure}[tb]
  \centering
  \includegraphics[page=1]{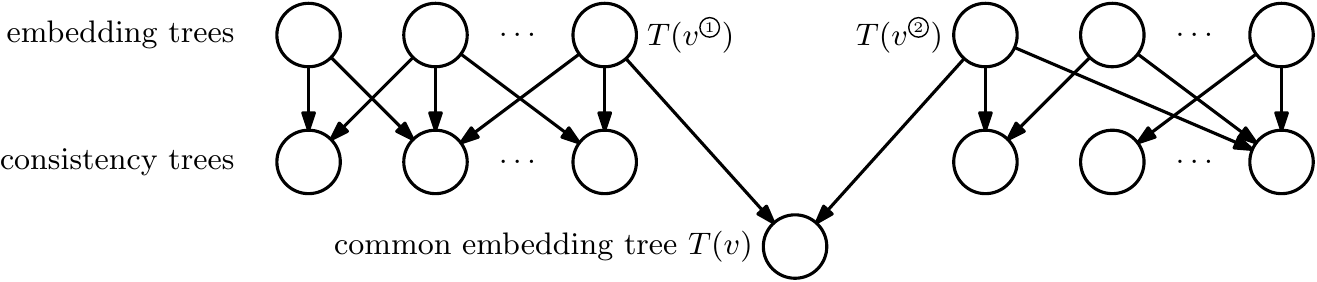}
  \caption{The Q-embedding representations of $\1G$ and $\2G$ together
    with the common embedding tree $T(v)$ of $v$.}
  \label{fig:sim-pq-ord-instances}
\end{figure}  

As the embedding trees (which contain only Q-nodes) are the sources of
$D(\1G, \2G)$, normalizing $D(\1G, \2G)$ yields an equivalent
instances containing no P-nodes~\cite{br-spoacep-13}.  Instances with
this property are equivalent to a set of Boolean equations and
inequalities, containing one variable for each Q-node in each PQ-tree
(and thus includes the PR-ordering variables).  We call this set of
equations and inequalities the \emph{biconnected PR-ordering
  constraints}.  To obtain the following lemma, it remains to proof
the running time.

\begin{lemma}
  \label{lem:ordering-biconnected}
  Let $\1G$ and $\2G$ be two biconnected graphs with common P-node
  degree~3 and let $\alpha$ be a variable assignment for the
  PR-ordering variables.  The graphs $\1G$ and $\2G$ admit embeddings
  that realize $\alpha$ and have consistent edge orderings if and only
  if $\alpha$ satisfies the biconnected PR-ordering constraints.
 
  The biconnected PR-ordering constraints have size $O(n)$ and can be
  computed in $O(n)$ time.
\end{lemma}
\begin{proof}
  Let $D(\1G, \2G)$ be the instances of \textsc{Simultaneous
    PQ-Ordering} as described above.  Clearly, $D(\1G, \2G)$ can be
  constructed in linear time and its size is linear in the size of the
  input graphs.  In general, the equations and inequalities for a
  given instance of \textsc{Simultaneous PQ-Ordering} can be computed
  in quadratic time~\cite{br-spoacep-13}.  In this specific case, it
  can be done in linear time for the following reasons.

  For every arc in $D(\1G, \2G)$ one needs to compute the
  normalization of the child (which takes linear time in the size of
  the parent~\cite{br-spoacep-13}) and a mapping from each inner node
  of the parent to its representative in the child (which takes again
  linear time in the size of the parent~\cite{br-spoacep-13}).  When
  computing the Q-embedding representations for $\1G$ and $\2G$ from
  their SPQR-trees (which can be done in linear time), we can make
  sure that the resulting instances are already normalized and that we
  already know the mapping from the nodes of the embedding trees to
  the consistency trees.  The remaining arcs in $D(\1G, \2G)$ are arcs
  from embedding trees to common embedding trees.  Thus, for every
  PQ-tree $T$ in $D(\1G, \2G)$, it suffices to normalize a single
  outgoing edge, which can be done in time linear in the size of~$T$.
\end{proof}

\subsubsection{Allowing Cutvertices}
\label{sec:allowing-cutvertices}

In the following, we extend this result to the case where we allow
exclusive cutvertices and simultaneous cutvertices of common degree~3.
Let $\1{B_1}, \dots, \1{B_k}$ be the blocks of $\1G$ and let $\2{B_1},
\dots, \2{B_\ell}$ be the blocks of $\2G$.  We say that embeddings of
these blocks have \emph{blockwise consistent edge orderings} if for
every pair of blocks $\1{B_i}$ and $\2{B_j}$ sharing a vertex $v$ the
edges incident to $v$ they share are ordered consistently.  To have
consistent edge orderings, it is obviously necessary to have blockwise
consistent edge orderings.

When composing the embeddings of two blocks that share a cutvertex,
the edges of each of the two blocks have to appear consecutively (note
that this is no longer true for three or more blocks), which leads to
another necessary condition.  Let $v$ be an exclusive cutvertex of
$\1G$.  Then $v$ is contained in a single block of $\2G$ whose
embedding induces an order $\2O$ on all common edges incident to $v$.
Let $\1{B_i}$ and $\1{B_j}$ (for $i, j \in \{1, \dots, k\}$ with $i
\not= j$) be a pair of blocks containing $v$ and let $\2{O_{i,j}}$ be
the order obtained by restricting $\2O$ to the common edges in
$\1{B_i}$ and $\1{B_j}$.  Then the common edges of $\1{B_i}$ must be
consecutive in the order $\2{O_{i, j}}$.  If this is true for every
pair of blocks at every exclusive cutvertex, we say that the
embeddings have \emph{pairwise consecutive blocks}.

\begin{lemma}
  \label{lem:blockwise-cutvertices}
  Two graphs without simultaneous cutvertices admit embeddings with
  consistent edge orderings if and only if their blocks admit
  embeddings that have blockwise consistent edge orderings and
  pairwise consecutive blocks.
\end{lemma}
\begin{proof}
  Let $v$ be a cutvertex in $\1G$ and let $\1{B_1}, \dots, \1{B_k}$ be
  the blocks containing $v$.  Moreover, let $\2O$ be the order of
  common edges around $v$ given by the unique block of $\2G$
  containing $v$.  We only have to show that the embeddings of the
  blocks $\1{B_1}, \dots, \1{B_k}$ can be composed such that they
  induce the order $\2O$ on the common edges.  For $k = 2$ this is
  clear, as we can choose arbitrary outer faces for $\1{B_1}$ and
  $\1{B_2}$ and combine their embeddings by gluing them together at
  $v$.  For $k > 2$ it is easy to see that there must be one block,
  without loss of generality $\1{B_1}$, whose edges appear
  consecutively in $\2O$.  The embeddings of all remaining blocks
  $\1{B_2}, \dots, \1{B_k}$ can be composed by induction such that the
  edges they contain are ordered the same as in $\2O$.  Moreover,
  composing the embedding of $\1{B_1}$ with the resulting embedding of
  $\1{B_2}, \dots, \1{B_k}$ works the same as the composition of two
  blocks.
\end{proof}

To extend Lemma~\ref{lem:ordering-biconnected} to the case where we
allow exclusive cutvertices, we consider the Q-embedding
representations of each block.  Ensuring blockwise consistent edge
orderings works more or less the same as ensuring consistent edge
orderings in the biconnected case.  Moreover, we will see how to add
additional PQ-trees to ensure pairwise consecutive blocks.

Note that the Q-embedding representations again yield PR-ordering
variables.  Fixing these variables determines the edge orderings of
common edges in each block of $\1G$ and in each block of $\2G$.  If we
have no simultaneous cutvertices, every vertex is either not a
cutvertex in $\1G$ or not a cutvertex in $\2G$.  Thus, the PR-ordering
variables actually determine all edge orderings of the common graph.
Thus, although there are new types of embedding choices at the
cutvertices in $\1G$ and at the cutvertices in $\2G$, these choices
are already covered by the PR-ordering variables (at least in terms of
edge orderings).

Let us formally describe the instance of \textsc{Simultaneous
  PQ-Ordering} announced above.  Let $\1{B_1}, \dots \1{B_k}$ be the
blocks of $\1G$ and let $\2{B_1}, \dots, \2{B_\ell}$ be the blocks of
$\2G$.  We start with an instance of \textsc{Simultaneous PQ-Ordering}
containing the Q-embedding representation of each of these blocks.
Let $v$ be a vertex of $G$ that is not a cutvertex.  Then $v$ is
contained in a single block of $\1G$ and in a single block of $\2G$,
let $\1v$ and $\2v$ be the occurrences of $v$ in these blocks.  As
before, there are two embedding trees $T(\1v)$ and $T(\2v)$ describing
the order of edges around $v$ in $\1G$ and $\2G$, respectively.  As
before we can enforce consistent ordering around $v$ by inserting a
common embedding tree as a common child of $T(\1v)$ and $T(\2v)$.

Let $v$ be a cutvertex of $\1G$ (the case that $v$ is a cutvertex of
$\2G$ is symmetric).  Then $v$ occurs in several blocks of $\1G$,
without loss of generality $\1{B_1}, \dots, \1{B_r}$.  We denote the
occurrences of $v$ in these blocks by $\1{v_1}, \dots, \1{v_r}$.  As
$v$ is not a simultaneous cutvertex, it occurs in a single block $\2B$
of $\2G$.  We denote this occurrence by $\2v$.  The embedding tree
$T(\2v)$ contains a leaf for each of common edge incident to $v$, thus
fixing the order of $T(\2v)$ already fixes the order of all common
edges around $v$.  To ensure consistency, we have to enforce that
conditions of Lemma~\ref{lem:blockwise-cutvertices}, i.e., blockwise
consistent edge orderings and pairwise consecutive blocks.

Ensuring blockwise consistent edge orderings is equivalent to
enforcing that the common edges incident to $v$ in $\1{B_i}$ (for $i =
1, \dots, k$) are ordered the same with respect to the Q-embedding
representations of $\1{B_i}$ and $\2B$.  This can be done by inserting
a new child consisting of a single P-node as common child of $T(\2v)$
and $T(\1{v_i})$.  We call this tree the \emph{blockwise common
  embedding tree}.

To ensure pairwise consecutive blocks, consider the two blocks
$\1{B_i}$ and $\1{B_j}$.  We create a PQ-tree $T_{i, j}(\1v)$ that has
a leaf for each common edge incident to $\1v$ that belongs to one of
the two blocks $\1{B_i}$ or $\1{B_j}$.  The structure of $T_{i,
  j}(\1v)$ is chosen such that all common edges in $\1{B_i}$ are
consecutive; see Figure~\ref{fig:pairwise-consec-tree}.  We call $T_{i,
  j}(\1v)$ \emph{pairwise consecutivity trees} and add it as a child
of the embedding tree $T(\2v)$, which has a leaf for every common edge
incident to $v$.

\begin{figure}
  \centering
  \includegraphics[page=1]{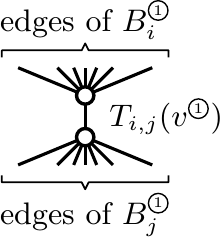}
  \caption{A pairwise consecutivity tree $T_{i, j}(\1v)$.}
  \label{fig:pairwise-consec-tree}
\end{figure}

We denote the resulting instance of \textsc{Simultaneous PQ-Ordering}
by $D(\1G, \2G)$.  As before, all sources in $D(\1G, \2G)$ contain
only Q-nodes.  Thus, normalizing $D(\1G, \2G)$ leads to an instance
containing only Q-nodes~\cite{br-spoacep-13}, which is again
equivalent to a set of equations and inequalities.  We call this set
the \emph{PR-ordering constraints}.

So far, the PR-ordering constraints only ensure blockwise consistent
edge orderings and pairwise consecutive blocks if every cutvertex is
an exclusive cutvertex.  Recall that there is no need for handling
union cutvertices; see Theorem~\ref{thm:no-union-cutvertices}.  Assume
we allow simultaneous cutvertices of common degree~3 and let $v$ be
such a cutvertex.  Then there are two possibilities.  If $v$ does not
separate the three common edges incident to $v$ in one of the graphs
$\ii G$ (for $i \in \{1, 2\}$), then the PR-ordering variables of $\ii
G$ also determine the common edge ordering around $v$ and thus this
simultaneous cutvertex actually behaves like an exclusive cutvertex.
Otherwise, the common edges incident to $v$ are separated by $v$ in
$\1G$ and in $\2G$.  Thus, changing the edge ordering of the common
edges at $v$ in an embedding of $\1G$ has no effect on any other edge
ordering.  As the same holds for $\2G$, we actually choose an
arbitrary edge ordering for the common edges incident to $v$,
independent from all other edge orderings.

Thus, we do not need to add additional constraints for the case that
we allow simultaneous cutvertices of common degree~3. To obtain the
following lemma, it remains to prove the running time.

\begin{lemma}
  \label{lem:ordering-no-sim-cutv}
  Let $\1G$ and $\2G$ have common P-node degree~3 and simultaneous
  cutvertices of common degree at most~3.  Let $\alpha$ be a variable
  assignment for the PR-ordering variables.  The graphs $\1G$ and
  $\2G$ admit embeddings that realize $\alpha$ and have consistent
  edge orderings if and only if $\alpha$ satisfies the PR-ordering
  constraints.

  The PR-ordering constraints have size $O(n^2)$ and can be computed
  in $O(n^2)$ time.
\end{lemma}
\begin{proof}
  The size of the PR-ordering constraints is linear in the size of the
  instance $D(\1G, \2G)$.  Clearly, the Q-embedding representations of
  the blocks have overall linear size.  The common embedding tree for
  a vertex $v$ has size $O(\deg(v))$.  Similarly, the blockwise common
  embedding trees of a vertex $v$ have total size $O(\deg(v))$.
  However, if $\1G$ has a linear number of blocks incident to a vertex
  $v$, then we get a quadratic number of pairwise consecutivity
  trees. 
  

  To get the PR-ordering constraints from the instance $D(\1G, \2G)$
  we have to compute for each arc in $D(\1G, \2G)$ the normalization
  and the mapping from the Q-nodes of the source to their
  representatives in the target.  As in the proof of
  Lemma~\ref{lem:ordering-biconnected}, the arcs of $D(\1G, \2G)$
  belonging to the Q-embedding representation can be computed in
  linear time.  All remaining arcs have an embedding tree as source.
  An embedding tree of a vertex $v$ has only $O(\deg(v))$ common
  embedding trees or blockwise common embedding trees as children.
  Processing all arcs to these children takes $O(\deg(v)^2)$ time and
  thus overall $O(n^2)$ time.

  It remains to deal with arcs of the following type.  The source is
  the embedding tree $T(\2v)$ and the target is the pairwise
  consecutivity tree $T_{i, j}(\1v)$ for a pair of blocks $\1{B_i}
  \not= \1{B_j}$.  Normalizing such an arc would usually take
  $O(\deg(v))$ time, which has to be done for $O(\deg(v)^2)$ pairwise
  consecutivity trees, resulting in the running time $O(\deg(v)^3)$.
  To improve this, we can make use of the fact that the subtree of
  $T_{i, j}(\1v)$ (after the normalization) that represents only the
  edges in $\1{B_i}$ has always the same structure, independent of the
  other block $\2{B_i}$.

  We first compute for each block $\1{B_i}$ the reduction of $T(\2v)$
  with the leaves belonging to $\1{B_i}$, which takes $O(\deg(v))$
  time for each of the $O(\deg(v))$ blocks.  The resulting tree
  contains a single node $\eta_i$ separating the leaves belonging to
  $\1{B_i}$ from all other leaves.  We project this tree to the leaves
  belonging to $\1{B_i}$ to obtain the tree $T_i(\1v)$ with root
  $\eta_i$.  Computing these trees $T_i(\1v)$ together with the
  mapping from the nodes in $T(\2v)$ to their representatives in
  $T_i(\1v)$ takes $O(\deg(v))$ time for each block and thus overall
  $O(\deg(v)^2)$ time.

  The desired (normalized) pairwise consecutivity tree $T_{i, j}(\1v)$
  can be obtained by identifying the roots $\eta_i$ and $\eta_j$ of
  the trees $T_i(\1v)$ and $T_j(\1v)$ with each other.  Extending the
  mapping to the resulting tree $T_{i, j}(\1v)$ can be easily done in
  linear time in the size of $T_{i, j}(\1v)$.  Hence, for the whole
  instance, we get the running time $O(n^2)$.
%
\end{proof}

\subsubsection{Cutvertex-Ordering Variables}
\label{sec:cutv-order-vari}

Let $v$ be an exclusive cutvertex of $\1G$ and let $\1{B_1}, \dots,
\1{B_k}$ be the blocks incident to $v$ that include common edges
incident to $v$.  As mentioned before, the choice of how these blocks
are ordered around $v$ and how they are nested into each other is
determined by the PR-ordering variables of $\2G$.  Consider a pair of
blocks $\1{B_i}$ and $\1{B_j}$.  As the choice of how $\1{B_i}$ and
$\1{B_j}$ are embedded into each other at $v$ may also have an effect
on some relative positions, we would like to have more direct access
to this information (not only via the PR-ordering variables of $\2G$).

Let $\1B \in \{\1{B_1}, \dots, \1{B_k}\}$ be a block of $\1G$ that
contains the common edge $e$ incident to $v$ and let $e_1$ and $e_2$
be two common edges incident to $v$ that are contained in one block
distinct from $\1B$.  We create a variable $\ord(e_1, e_2, \1B)$ to
represent the binary decision of ordering the edges $e_1$, $e_2$, and
$e$ in this order or in its reversed order.  Note that this is
independent from the choice of the edge $e$ of $\1B$ (the blocks are
pairwise consecutive in every embedding).  We create such a variable
for every such triple $e_1$, $e_2$, and $\ii B$ and call them the
\emph{exclusive cutvertex-ordering variables}.


To make sure that the exclusive cutvertex-ordering variables are
consistent with the PR-ordering variables, it suffices to slightly
change the above instance $D(\1G, \2G)$ of \textsc{Simultaneous
  PQ-Ordering}.  Let $\ord(e_1, e_2, \1B)$ be an exclusive
cutvertex-ordering variable for the cutvertex $v$.  Let $e$ be a
common edge in $\1B$ incident to $v$ and let $\2v$ be the occurrence
of $v$ in $\2G$.  Then $D(\1G, \2G)$ contains the embedding tree
$T(\2v)$ that has a leaf for every common edge incident to $v$.  This
includes $e_1$, $e_2$, and $e$.  We create a new PQ-tree $T(e_1, e_2,
e)$ with three leaves corresponding to $e_1$, $e_2$, and $e$ and add
it as child of $T(\2v)$.  In every solution of the resulting instance
of \textsc{Simultaneous PQ-Ordering}, the value of $\ord(e_1, e_2,
\1B)$ then simply corresponds to the orientation chosen for PQ-tree
$T(e_1, e_2, e)$.

Adding this tree for every exclusive cutvertex-ordering variable
establishes the desired connection between these variables and the
PR-ordering variables.  We call the constraints we get from the
resulting instance of \textsc{Simultaneous PQ-Ordering} in addition to
the PR-ordering constraints the \emph{cutvertex-ordering constraints}.

Let $v$ be a simultaneous cutvertex of common degree~3 such that the
common edges incident to $v$ are separated by $v$ in $\1G$ and in
$\2G$.  Recall that this is the unique case, where PR-ordering
variables do not determine the order of the common edges around $v$.
Let $e_1$, $e_2$, and $e_3$ be the common edges incident to $v$.  To
make sure that assigning values to all variables actually determines
all edge-orderings, we add the variable $\ord(e_1, e_2, e_3)$
associated with the order of these three edges.  Recall that changing
this order in $\1G$ or $\2G$ has no effect on any other edge ordering.
Hence, there is no need to add further constraints.  If two of the
edges, without loss of generality $e_1$ and $e_2$, belong to the same
block of $\1G$ and $e_3$ belongs to another block $\1B$, we denote
$\ord(e_1, e_2, e_3)$ also by $\ord(e_1, e_2, \1B)$ to obtain
consistency with the naming of the exclusive cutvertex-ordering
variables.

We call these variables together with the exclusive cutvertex-ordering
variables the \emph{cut\-vertex-ordering variables}.  The PR-ordering
variables together with the cutvertex-ordering variables are simply
called \emph{ordering variables}.  Moreover, the PR-ordering
constraints together with the cutvertex-ordering constraints are
called \emph{ordering constraints}.  To extend
Lemma~\ref{lem:ordering-no-sim-cutv} to incorporate the
cutvertex-ordering variables, it remains to show that the
cutvertex-ordering constraints have $O(n^3)$ size and can be computed
in $O(n^3)$ time.

\begin{lemma}
  \label{lem:ordering-no-sim-cutv-with-cutv-var}
  Let $\1G$ and $\2G$ have common P-node degree~3 and simultaneous
  cutvertices of common degree at most~3.  Let $\alpha$ be a variable
  assignment for the ordering variables.  The graphs $\1G$ and $\2G$
  admit embeddings that realize $\alpha$ and have consistent edge
  orderings if and only if $\alpha$ satisfies the ordering
  constraints.

  The ordering constraints have size $O(n^3)$ and can be computed in
  $O(n^3)$ time.
\end{lemma}
\begin{proof}
  The size of the cutvertex-ordering constraints is clearly linear in
  the number of cutvertex-ordering variables.  For each cutvertex $v$,
  the number of cutvertex-ordering variables is clearly in
  $O(\deg(v)^3)$.  Thus, it remains to show how to get the
  cutvertex-ordering constraints from the resulting instance $D(\1G,
  \2G)$ of \textsc{Simultaneous PQ-Ordering}.

  Let $v$ be a cutvertex in $\1G$ and let $\2v$ be the occurrence of
  $v$ in $\2G$.  For every cutvertex-ordering variable of $v$ we have
  three common edges $e_1$, $e_2$, and $e$ and we need to find the
  node $\eta$ of the embedding tree $T(\2v)$ that separates the leaves
  corresponding to $e_1$, $e_2$, and $e$ from each other.  When
  rooting $T(\2v)$ at $e$, this node $\eta$ is the lowest common
  ancestor of $e_1$ and $e_2$.  Thus, after $O(\deg(v))$ preprocessing
  time, we can get the cutvertex-ordering constraint for every
  cutvertex-ordering variable that includes $e$ in constant time per
  variable~\cite{ht-fafnca-84}.  We have to spend this $O(\deg(v))$
  preprocessing time at most once for each common edge incident to
  $v$, yielding a total preprocessing time of $O(\deg(v)^2)$.  As we
  have $O(\deg(v)^3)$ cutvertex-ordering variables for $v$, the
  running time is dominated by the constant time LCA-queries, which
  yield the running time $O(\deg(v)^3)$.  For the whole instance, this
  gives the claimed $O(n^3)$ running time.
\end{proof}

\subsection{Common-Face Constraints}
\label{sec:common-face-constr}

Recall from Section~\ref{sec:spec-bridg-comm-face-constr} that we can
assume that there are no bridges that are block-local and exclusive
one-attached if we in return solve \textsc{Sefe} with block-local
common-face constraints.  In this section, we show how to handle these
additional constraints.  To this end, we show that satisfying
block-local common-face constraints in a given instance of
\textsc{Sefe} is equivalent to satisfying a set of equations and
inequalities.  The union of these constraints with the constraints
from the previous section thus enforce that the embeddings of each
common connected component are consistent and satisfy given
block-local common-face constraints.

Let $B$ be a block of the common graph.  Let $\mu$ be an R-node of
$B$.  Then we introduce the binary variable $\ord(\mu)$ where
$\ord(\mu) = 0$ indicates that $\skel(\mu)$ is embedded according to
its reference embedding and $\ord(\mu) = 1$ indicates that
$\skel(\mu)$ is flipped.  In case $\mu$ is a P-node of $B$, we can
assume by Theorem~\ref{thm:all-link-graphs-are-connected} that the
union-link graph of $\mu$ is connected.  Recall that this implies that
the embedding of $\skel(\mu)$ is fixed up to a flip
(Lemma~\ref{lem:link-graph-adjacent}).  Thus, we also get a reference
embedding for $\skel(\mu)$ and can describe the embedding choice for
$\skel(\mu)$ with a binary variable $\ord(\mu)$.  We call these
variables the \emph{common PR-node variables}.

It follows directly from
Section~\ref{sec:simult-embedd-union-bridge-constr} that common-face
constraints for $B$ are equivalent to a set of equations and
inequalities on the variables $\ord(\mu)$ (we basically get the
consistency and union-bridge constraints from Step~\ref{step:2} and
Step~\ref{step:4}).

Note that the constraints from Section~\ref{sec:cons-edge-order}
enforcing consistent edge orderings do not contain common PR-node
variables.  They only contain PR-ordering variables determining the
embeddings of $\1G$ and $\2G$ (Lemma~\ref{lem:ordering-no-sim-cutv}).
As fixing the PR-node variables for $\1G$ fixes the embedding of $G$,
it also fixes the values for all common PR-node variables.  It remains
to show that this dependency of the common PR-node variables from the
PR-ordering variables in $\1G$ can be expressed using a set of
equations and inequalities.

To this end, consider a common vertex $v$ in the common block $B$ and
let $\1B$ be the block of $\1G$ containing $B$.  Let $T(v)$ be the
embedding tree of $v$ in $B$, i.e., the PQ-tree describing the
possible edge orderings of common edges incident to $v$.  Note that
each inner node of $T(v)$ is actually a Q-node and that $T(v)$ has one
inner node for each P- or R-node whose embedding affects the edge
ordering around $v$.  Let $\1v$ be the occurrence of $v$ in $\1B$ and
let $T(\1v)$ be the embedding tree of $\1v$ in $\1B$ projected to the
common edges incident to $v$.  Then $T(\1v)$ describes all orders of
common edges incident to $v$ that can be induced by an embedding of
$\1B$.

Clearly, $T(\1v)$ is more restrictive than $T(v)$ in the sense that
every order represented by $T(\1v)$ is also represented by $T(v)$.
Thus, $T(\1v)$ is a reduction of $T(v)$.  It is not hard to
see~\cite{br-spoacep-13} that every Q-node in $T(v)$ has a unique
Q-node in $T(\1v)$ (called its \emph{representative}) that determines
its flip.  Thus, for every P- or R-node $\mu$ of $G$, we find at least
one vertex $v$ such that $\ord(\mu)$ corresponds to the flip of a
Q-node in $T(v)$, which corresponds to a flip of a Q-node in $T(\1v)$,
which corresponds to a PR-ordering variable $\ord(\1\mu)$ (or
$\ord(\1{\mu_v})$) of $\1G$.  Thus, $\ord(\mu) = \ord(\1\mu)$ or
$\ord(\mu) \not= \ord(\1\mu)$ gives us the desired connection between
the common PR-node variables and the PR-ordering variables.

We call the set of all equations and inequalities described in this
section the \emph{common-face constraints}.  With the results form
Section~\ref{sec:simult-embedd-union-bridge-constr} (and with standard
PQ-tree operations), we can compute the common-face constraints in
linear time.  This yields the following lemma where $n$ is the total
input size, i.e., the size of the two graphs plus the size of the
common-face constraints.

\begin{lemma}
  \label{lem:common-face-constraints}
  Let $(\1G, \2G)$ be an instance of \textsc{Sefe} with common-face
  constraints and let $\alpha$ be a variable assignment for the
  PR-ordering and the common PR-node variables.  Every embedding of
  $\1G$ realizing $\alpha$ satisfies the common-face constraints if
  and only if $\alpha$ satisfies the common-face constraints.

  The common-face constraints have $O(n)$ size and can be computed in
  $O(n)$ time.
\end{lemma}

\subsection{Consistent Relative Positions}
\label{sec:cons-relat-posit}

Let $H$ and $H'$ be two connected components of the common graph $G$.
To represent the relative position $\pos_{H'}(H)$ of $H$ with respect
to $H'$, we use the relative positions $\pos_C(H)$ of $H$ with respect
to the cycles $C$ in an extended cycle basis of $H'$.  With $\pos_C(H)
= 0$ and $\pos_C(H) = 1$, we associate the cases that $H$ lies to the
right of $C$ and to the left of $C$, respectively.  We call these
variables the \emph{component position variables}.  Note that fixing
all position variables determines all relative positions of common
connected components with respect to each other
(Theorem~\ref{thm:cycle-basis}).  Thus, fixing the PR-ordering
variables (which also fixes the cutvertex-ordering variables) and the
position variables completely determines the embedding of the common
graph $G$.

In this section, we give a set of necessary equations and inequalities
on the position variables of two graphs $\1G$ and $\2G$ that enforce
consistent relative positions on their common graph $G$.  As fixing
the PR-ordering variables may also determine some position variables,
we also have to make sure that these two types of variables are
consistent with each other.

Let $\1G$ be a connected planar graph containing $G$, let $C$ be a
cycle in $G$ (a cycle from the extended cycle basis), and let $H$ be a
connected component of $G$ not containing $C$.  Depending on how $C$
and $H$ are located in $\1G$, different embedding choices of $\1G$
determine the relative position $\pos_C(H)$~\cite{br-drpse-15}.  We
quickly list these embedding choices here and describe the constraints
arising from them in the following sections.

Let $\1\mu$ be an R-node of $\1G$ such that $C$ induces a cycle
$\kappa$ in $\skel(\1\mu)$.  If $\skel(\1\mu)$ has a virtual edge
$\eps$ that is not part of $\kappa$ such that the expansion graph
$\expan(\eps)$ includes a vertex of $H$, then the embedding of
$\skel(\1\mu)$ determines the relative position $\pos_C(H)$.  We say
that $\pos_C(H)$ is \emph{determined by the R-node $\1\mu$}.

Let $\1\mu$ be a P-node of $\1G$ such that $C$ induces a cycle
$\kappa$ in $\skel(\1\mu)$.  Then $\kappa$ consists of two virtual
edges $\eps_1$ and $\eps_2$.  The relative position $\pos_C(H)$ is
determined by the embedding of $\skel(\1\mu)$ if $H$ is contained in
the expansion graph of a virtual edge $\eps$ with $\eps \not= \eps_i$
for $i \in \{1, 2\}$.  We say that $\pos_C(H)$ is \emph{determined by
  the P-node $\1\mu$}.

If $C$ and $H$ belong to the same block $\1B$ of $\1G$, then
$\pos_C(H)$ is either determined by an R-node or by a P-node of
$\1B$~\cite{br-drpse-15}.  Otherwise, let $v$ be the cutvertex in
$\1G$ that separates $C$ from $H$ and belongs to the block of $C$.  If
$v$ is not contained in $C$, then we introduce the variable
$\pos_C(v)$ corresponding to the decision of embedding $v$ to the
right or to the left of $C$.  We call such a variable the
\emph{cutvertex position variables}.  Clearly, $H$ and $v$ lie on the
same side in each embedding of $\1G$.  Moreover, the relative position
of $v$ with respect to $C$ is determined by an R-node or by a P-node
$\1\mu$ (the belong to the same block of $\1G$).  In this case, we
also say that both variables, $\pos_C(v)$ and $\pos_C(H)$, are
determined by the R-node or P-node $\1\mu$.  Moreover, we also say
that $\pos_C(H)$ is \emph{determined by $\pos_C(v)$}.

If the cutvertex $v$ is contained in $C$, then the relative position
$\pos_C(H)$ is determined by the embedding choices made at the
cutvertex.  We distinguish two cases.  Let $\1S$ be the split
component with respect to the cutvertex $v$ that contains $H$.  If
$\1S$ includes a common edge incident to $v$, we say that
$\pos_C(H)$ is \emph{determined by the common cutvertex $v$}.
Otherwise, we say that $\pos_C(H)$ is \emph{determined by the
  exclusive cutvertex $v$} (note that $v$ might still be a
cutvertex of the common graph in this case).  These two cases are in
so far different as changing $\pos_C(H)$ affects the edge ordering of
the common graph in the former case, whereas it does not in the latter
case.

The component position variables together with the cutvertex position
variables are simply called \emph{position variables} In the following
sections, we describe for each of the four cases different constraints
in the form of equations and inequalities on position variables,
PR-ordering variables, and cutvertex-ordering variables.

\subsubsection{Relative Positions Determined by R-Nodes}

We start with the simplest case, where $\pos_C(H)$ is determined by an
R-node $\1\mu$ of $\1G$.  If $\pos_C(H)$ is also determined by a
cutvertex position variable $\pos_C(v)$, we simply set $\pos_C(H) =
\pos_C(v)$ to make sure that the cutvertex $v$ and the component
$H$ are on the same side of $C$.

Otherwise, $C$ induces a cycle $\kappa$ in $\skel(\1\mu)$ and $H$
shares a vertex with the expansion graph of a virtual edge $\eps$ not
belonging to $\kappa$.  The PR-ordering variable $\ord(\1\mu)$
determines whether $\skel(\1\mu)$ is embedded according to its
reference embedding ($\ord(\1\mu) = 0$) or whether it is flipped
($\ord(\1\mu) = 1$).

Assume $\eps$ lies to the right of $\kappa$ in the reference embedding
of $\skel(\1\mu)$.  Then $\ord(\1\mu) = 0$ implies that $\eps$ lies to
the right of $\kappa$, which implies that $H$ lies to the right of
$C$, i.e., $\pos_C(H) = 0$.  Moreover, flipping $\skel(\1\mu)$ brings
$\eps$ to the left of $\kappa$.  Thus, $\ord(\1\mu) = 1$ implies
$\pos_C(H) = 1$, which yields $\ord(\1\mu) = \pos_C(H)$ as necessary
condition.  If $\eps$ lies to the left of $\kappa$ in the reference
embedding, we obtain $\ord(\1\mu) \not= \pos_C(H)$ instead.

Analogously, we set $\ord(\1\mu) = \pos_C(v)$ or $\ord(\1\mu) \not=
\pos_C(v)$ for every cutvertex position variable $\pos_C(v)$
determined by $\1\mu$.

\paragraph{Sufficiency.}

We call the constraints defined in this section the \emph{R-node
  constraints}.  The following lemma states the more or less obvious
necessity and sufficiency of the R-node constraints.


\begin{lemma}
  \label{lem:r-node-constr-suff}
  Let $\1\mu$ be an R-node of $\1G$ and let $X$ be the set of position
  variables determined by $\1\mu$ together with the PR-ordering
  variable $\ord(\1\mu)$.  A variable assignment $\alpha$ of $X$ can
  be realized by an embedding of $\1G$ if and only if $\alpha$
  satisfies the R-node constraints.
%
\end{lemma}



\subsubsection{Relative Positions Determined by P-Nodes}

Let $\1\mu$ be a P-node of $\1G$ with poles $u$ and $v$ and let
$\eps_1, \dots, \eps_k$ be the virtual edges of $\skel(\1\mu)$.  Let
$C$ be a cycle in $G$, that induces a cycle $\kappa$ in
$\skel(\1\mu)$.  Without loss of generality, $\kappa$ consists of the
virtual edges $\eps_1$ and $\eps_2$.  Let $H$ be a common connected
component whose relative position with respect to $C$ is determined by
$\1\mu$.

As for R-nodes, we first consider the case that a relative position
$\pos_C(H)$ is determined by a cutvertex position variable
$\pos_C(v)$.  In this case we simply set $\pos_C(H) = \pos_C(v)$.
In the following we assume that $\pos_C(H)$ is determined by the
P-node $\1\mu$ but not by a cutvertex position variable.  All
cutvertex position variables $\pos_C(v)$ are handled analogously.

As $\pos_C(H)$ is determined by $\1\mu$, the common connected
component $H$ is contained in a virtual edge $\eps \in \{\eps_3,
\dots, \eps_k\}$ not belonging to $\kappa$.  Note that embedding
$\eps$ to the right or to the left of $\kappa$ determines not only the
relative position of $H$ but the position of every connected component
that is contained in the expansion graph of $\eps$.  We make sure that
these relative positions fit to each other by introducing a new
variable $\pos_\kappa(\eps)$ with the interpretation that
$\pos_\kappa(\eps) = 0$ if and only if $\eps$ is embedded to the right
of $\kappa$.  Clearly, $\pos_\kappa(\eps) = \pos_C(H)$ is a necessary
condition for every connected component $H$ contained in the expansion
graph of $\eps$.

If there is another common cycle $C'$ inducing the same cycle $\kappa$
in $\skel(\mu)$, we use the same variable $\pos_\kappa(\eps)$ to
determine on which side of $\kappa$ the virtual edge $\eps$ is
embedded.  If $C'$ induces the same cycle oriented in the opposite
direction, we use the negation of $\pos_\kappa(\eps)$ instead.

The above constraints are not sufficient for two reasons.  First,
changing the embedding of $\skel(\1\mu)$ may change edge orderings in
the common graph.  In this case, there are PR-ordering variables
partially determining the embedding of $\skel(\1\mu)$ and we have to
make sure that their values and the values of the position variables
determined by $\1\mu$ fit to each other.  Second, if different common
cycles induce different cycles in $\skel(\1\mu)$, then not every
combination of relative positions with respect to these cycles can
actually be achieved by embedding the skeleton.

\paragraph{Connection to Ordering Variables.}

As mentioned before, we have to add additional constraints to ensure
consistency between the position variables and the ordering variables.
Let $\kappa$ be the cycle induced by the cycle $C$ in $\skel(\1\mu)$
and assume without loss of generality that $\kappa$ consists of the
virtual edges $\eps_1$ and $\eps_2$.  Embedding another virtual edge
$\eps$ to the right or to the left of $\kappa$ changes the edge
ordering at the poles $u$ or $v$ if and only if the expansion graph of
$\eps$ includes common edges incident to $u$ or $v$, respectively.  In
this case, we have a PR-ordering variable determining this embedding
choice and we can make sure that the edge ordering and the relative
positions fit to each other using an equation or an inequality.

To make this more precise, we define the following five types of
virtual edges.  Each virtual edge $\eps \in \{\eps_1, \dots, \eps_k\}$
is of exactly one of the following five types.
\begin{enumerate}[\bf{Type }1.]
\item $\expan(\eps)$ includes a common path from $u$ to $v$.
\item $\expan(\eps)$ has common edges incident to $u$ and to $v$ but
  is not of Type~1.
\item $\expan(\eps)$ has a common edge incident to $u$ but none
  incident to $v$.
\item $\expan(\eps)$ has a common edge incident to $u$ but none
  incident to $v$.
\item $\expan(\eps)$ has no common edges incident to $u$ or to $v$.
\end{enumerate}

As the cycle $\kappa$ consists of $\eps_1$ and $\eps_2$, they must
both be of Type~1.  Choosing the relative position of $\eps$ with
respect to $\kappa$ affects the edge ordering at $u$ or at $v$ if and
only if $\eps$ is not of Type~5.  If $\eps$ is of Type~1 or Type~2,
then there is a PR-ordering variable $\ord(\1\mu)$ determining the
order of the edges $\eps_1, \eps_2, \eps$.  If the ordering
corresponding to $\ord(\1\mu) = 0$ has $\eps$ to the right of
$\kappa$, we set $\ord(\1\mu) = \pos_\kappa(\eps)$.  Otherwise, we set
$\ord(\1\mu) \not= \pos_\kappa(\eps)$.

If $\eps$ is of Type~3, the edge ordering of $\eps_1, \eps_2, \eps$ is
determined by the PR-ordering variable $\ord(\1{\mu_s})$.  As before,
we get either the equation $\ord(\1{\mu_s}) = \pos_\kappa(\eps)$ or
the inequality $\ord(\1{\mu_s}) \not= \pos_\kappa(\eps)$.  The case
that $\eps$ is of Type~4 is analogous, except that we have the
PR-ordering variable $\ord(\1{\mu_t})$ instead of $\ord(\1{\mu_s})$.

\paragraph{Multiple Cycles.}

If there are common cycles inducing different cycles in the P-node
$\1\mu$, then at least three virtual edges must be of Type~1 (i.e.,
their expansion graph includes a common path between the poles $u$ and
$v$).  As we assume that $\1\mu$ has common P-node degree~3, three
virtual edges are of Type~1 and all remaining virtual edges are of
Type~5.  Let $\eps_1$, $\eps_2$, and $\eps_3$ be the virtual edges of
Type~1 and let $\eps$ be another virtual edge of $\skel(\1\mu)$.
Denote the cycle consisting of $\eps_i$ and $\eps_j$ by $\kappa_{i,
  j}$ ($i, j \in \{1, 2, 3\}$ and $i < j$).  To simplify the notation,
we use $\pos_{i, j}(\eps)$ as short form for the relative position
$\pos_{\kappa_{i, j}}(\eps)$.

Let $\eps \in \{\eps_4, \dots, \eps_k\}$ be another virtual edge of
$\skel(\1\mu)$.  Then we are interested in the three position
variables $\pos_{1, 2}(\eps)$, $\pos_{1, 3}(\eps)$, and $\pos_{2,
  3}(\eps)$, which are not independent from each other.  Moreover,
which combinations of relative positions can actually be realized
depends on the ordering of $\eps_1$, $\eps_2$, and $\eps_3$.  This
ordering is determined by the PR-ordering variable $\ord(\1\mu)$.  In
the remainder of this section, we first figure out which combinations
of values for $\ord(\1\mu)$, $\pos_{1, 2}(\eps)$, $\pos_{1, 3}(\eps)$,
and $\pos_{2, 3}(\eps)$ are actually possible and then show that
restricting the variables to these combinations is equivalent to a set
of equations and inequalities.

\begin{figure}
  \centering
  \includegraphics[page=1]{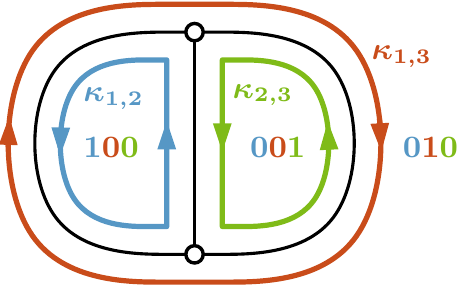}
  \caption{Three virtual edges of Type~1 in a P-node with the three
    different virtual cycles $\kappa_{1, 2}$, $\kappa_{1, 3}$, and
    $\kappa_{2, 3}$.  For each face, the variable assignment
    corresponding to this face is given.}
  \label{fig:4-constr-in-p-node}
\end{figure}

When fixing the order variable $\ord(\1\mu)$ (without loss of
generality to 0), we get the situation shown in
Figure~\ref{fig:4-constr-in-p-node}.  From the set of eight
combinations for the three position variables, only the three
combinations 100, 010 and 001 are possible.  When changing the order
of the edges (setting $\ord(\1\mu)$ to~1), every bit is reversed.
Thus, for the tuple $(\ord(\1\mu), \pos_{1, 2}(\eps),\allowbreak \pos_{1,
  3}(\eps),\allowbreak \pos_{2, 3}(\eps))$ we get the possibilities 0100, 0010,
0001 and their complements 1011, 1101, 1110.  A restriction equivalent
to this is called \emph{P-node 4-constraint} (where equivalent means
that an arbitrary subset of variables may be negated).

We note that, in the short version of this paper~\cite{bkr-se-13}, we
missed the fact that the combinations 1000 and 0111 are not possible.
This lead to the wrong assumption that a combination is feasible if
and only if there is an odd number of 1s, which can be expressed as a
linear equation over $\mathbb F_2$.  As a matter of fact, the P-node
4-constraint allows six different combinations, which is not a power
of two and can thus not be the solution space of a system of linear
equations over $\mathbb F_2$.  Thus, a P-node 4-constraint is in
particular not equivalent to a set of equations or inequalities.  We
resolve this issue with the following lemma in conjunction with the
new results from Section~\ref{sec:union-separ-pairs}.

\begin{lemma}
  \label{lem:p-node-4-const-simple}
  If two variables of a P-node 4-constraint are known to be equal or
  unequal, the P-node 4-constraint is equivalent to a set of equations
  and inequalities.
\end{lemma}
\begin{proof}
  We basically have the three possibilities 0100, 0010, and 0001 and
  their complements for the variables $abcd$.  No pair of variables is
  equal in all three possibilities and no pair of variables in unequal
  in all three possibilities.  Thus requiring equality or inequality
  for one of the pairs eliminates exactly one or two of these three
  possibilities.  If exactly one possibility and its complement
  remains, this is obviously equivalent to a set of equations and
  inequalities.

  If 0100 and 0010 (and their complements) remain, this is equivalent
  to $a = d$ and $b \not= c$.  If 0100 and 0001 remain, this is equivalent
  to $a = c$ and $b \not= d$.  Finally, if 0010 and 0001 remain, this
  is equivalent to $a = b$ and $c \not= d$.
\end{proof}

In the following, we show that for every P-node 4-constraint, we
always find an equation or inequality between a pair of variables,
turning all P-node 4-constraints into a set of equations and
inequalities.

Consider the union graph $G^\cup$.  If the poles $\{u, v\}$ are a
separating pair in $G^\cup$, then each split component is the union of
the expansion graphs of several virtual edges of $\skel(\1\mu)$.  As
the expansion graphs of $\eps_1$, $\eps_2$, and $\eps_3$ have common
$uv$-paths, we can assume that they are not separated
(Theorem~\ref{thm:only-special-union-separating-pairs}).  Moreover,
having common P-node degree~3 implies that none of the other expansion
graphs has a common edge incident to $u$ or to $v$ (they are all of
Type~5).  Thus, again by
Theorem~\ref{thm:only-special-union-separating-pairs}, we can assume
that $\{u, v\}$ is not a separating pair in $G^\cup$.

It follows that there must be a path $\pi$ in $G^\cup$ that connects a
vertex of $\expan(\eps)$ with a vertex of (without loss of generality)
$\expan(\eps_1)$ that does not pass through $u$ or $v$ or vertices of
the expansion graphs of $\eps_2$ and $\eps_3$.  It follows that the
relative position of $\eps$ with respect to $\kappa_{2, 3}$ must be
the same as the relative position of any internal vertex of
$\expan(\eps_1)$ with respect to $\kappa_{2, 3}$.  As this relative
position is determined by the order of $\eps_1$, $\eps_2$, and
$\eps_3$, we obtain either $\ord(\1\mu) = \pos_{2, 3}(\eps)$ or
$\ord(\1\mu) \not= \pos_{2, 3}(\eps)$.

\paragraph{Sufficiency.}

We call the constraints defined in this section the \emph{P-node
  constraints}.  The following lemma follows directly from the
previous considerations.

\begin{lemma}
  \label{lem:p-node-constr-suff}
  Let $\1\mu$ be a P-node of $\1G$ and let $X$ be the set of position
  variables determined by $\1\mu$ together with the PR-ordering
  variables of $\1\mu$.  A variable assignment $\alpha$ of $X$ can be
  realized by an embedding of $\1G$ if and only if $\alpha$ satisfies
  the P-node constraints.
\end{lemma}

\subsubsection{Relative Positions Determined by Common Cutvertices}

Recall that the relative position $\pos_C(H)$ is determined by a
common cutvertex $v$ of $\1G$ if $C$ contains $v$ and $H$ lies in a
split component $\1S$ (with respect to $v$) different from the split
component containing $C$ such that $\1S$ has a common edge incident to
$v$.

First note that the whole split component $\1S$ has to be embedded on
one side of $C$.  Thus, for every common connected component in $\1S$,
we would get the same set of constraints.  To reduce the number of
constraints, we introduce the variable $\pos_C(\1S)$ representing the
decision of embedding $\1S$ either to the right or to the left of $C$.
Clearly, $\pos_C(H) = \pos_C(\1S)$ for every common connected
component $H$ in~$\1S$ is a necessary condition.

Note that this condition is very similar to the first type of
constraints we required for P-nodes (connected components in the
expansion graph of the same virtual edge have the same relative
positions).  As for the P-nodes, we have to address two potential
issues.  First, embedding the split component $\1S$ to one side or
another of $C$ changes the edge ordering around the cutvertex $v$.
Second, if there are multiple cycles through $v$, then the relative
positions of $\1S$ with respect to all these cycles must be
consistent.

\paragraph{Connection to Ordering Variables.}

Let $\1B$ be the block of $\1S$ containing $v$ and let $e_1$ and $e_2$
be the two edges of $C$ incident to $v$.  Moreover, let $e$ be a
common edge of $\1B$ incident to $v$.  Recall that the cyclic order of
$e_1$, $e_2$, and $e$ is described by the cutvertex-ordering variable
$\ord(e_1, e_2, \1B)$.

Assume without loss of generality that $e_1$ is oriented towards $v$
and $e_2$ is oriented away from $v$ (in the cycle $C$).  Then the
(clockwise) cyclic order $e_1, e, e_2$ forces the block $\1B$, and
thus the whole split component $\1S$, to lie left of $C$.  The
opposite cyclic order forces $\1S$ to the right of $C$.  Thus,
depending on the orientation of $C$, we either get $\ord(e_1, e_2,
\1B) = \pos_C(\1S)$ or $\ord(e_1, e_2, \1B) \not= \pos_C(\1S)$ as
necessary conditions.

\paragraph{Multiple Cycles.}

Assume multiple common cycles $C_1, \dots, C_k$ contain the cutvertex
$v$ and assume that these cycles are already embedded.  We have to
make sure that every assignment to the variables $\pos_{C_i}(\1S)$ for
$i = 1, \dots, k$ actually corresponds to a face of $C_1 \cup \dots
\cup C_k$ that is incident to $v$.

We cannot directly express this requirement as a set of equations and
inequalities on the position variables.  However, if we assume that a
given variable assignment for the cutvertex-ordering variables of $v$
can be realized by an embedding of $\1G$ (which is ensured by the
constraints from Section~\ref{sec:cons-edge-order}), then the above
constraints establishing the connection between the cutvertex-ordering
variables and the position variables make sure that the corresponding
values for the position variables are also realized.

\paragraph{Sufficiency.}

We call the constraints from this section the \emph{common cutvertex
  constraints}.  Let $\alpha$ be a variable assignment for the
cutvertex-ordering variables of a cutvertex $v$.  We say that $\alpha$
is \emph{order realizable}, if $\1G$ admits an embedding realizing
$\alpha$.  We obtain the following lemma.

\begin{lemma}
  \label{lem:common-cutvertex-constr}
  Let $X$ be the set of position variables that are determined by the
  common cutvertex $v$ in $\1G$ and let $Y$ be the cutvertex-ordering
  variables of $v$.  A variable assignment $\alpha$ of $X \cup Y$ can
  be realized by an embedding of $\1G$ if and only if $\alpha$
  satisfies the common cutvertex constraints and $\rst{\alpha}{Y}$ is
  order realizable.
\end{lemma}

\subsubsection{Relative Positions Determined by Exclusive Cutvertices}

As in the previous section, let $\1S$ be the split component with
respect to the cutvertex $v$ that contains the connected component
$H$.  As before, every common connected component of $\1S$ has to be
embedded on the same side of $C$.  However, in this case, we need
slightly stronger constraints.

Let $H_v$ be the connected component of the common graph that includes
the cutvertex $v$.  Let further $B_1^\cup, \dots, B_k^\cup$ be the
union bridges of $H_v$ (note that this is the first time, where the
second graph $\2G$ comes into play).  As the union bridge $B_i^\cup$
(for $i = 1, \dots, k$) has to be completely embedded into a single
face of $H_v$, every common connected component in $B_i^\cup$ lies on
the same side of $C$.  As before for the split components, we
represent the decision of putting $B_i^\cup$ to the right or to the
left of $C$ using the variable $\pos_C(B_i^\cup)$.  Then the
constraint $\pos_C(B_i^\cup) = \pos_C(H)$ for every common connected
component $H$ in $B_i^\cup$ is clearly necessary.  Note that the
resulting constraints are strictly stronger than setting $\pos_C(\1S)
= \pos_C(H)$ for every common connected component $H$ in $\1S$, as
$\1S$ is contained in a single union bridge.

Recall that (in contrast to the previous section) $\1S$ does not
contain a common edge incident to $v$.  It follows that the decision
of putting $H$ to the right or to the left of $C$ in an embedding of
$\1G$ has no influence on the edge ordering at $v$.  Thus, there is no
need for further constraints to ensure consistency between edge
orderings and relative positions.  Moreover, we will see that there is
no need for additional constraints to make sure that the relative
positions actually describe a face (in case $v$ is contained in
multiple cycles).


\paragraph{Sufficiency.}

We call the constraints defined in this section the \emph{exclusive
  cutvertex constraints}.  Assume that $H_v$ is already embedded.
Then we can choose to embed $\1S$ into an arbitrary face of $H_v$
incident to $v$, which determines the relative positions of the
components in $\1S$ with respect to cycles through $v$ without
affecting any other embedding choice.  It only remains to make sure
that the relative positions of $H$ with respect to all cycles through
$v$ actually describe a face incident to $v$.  Unfortunately, the
exclusive cutvertex constraints do not guarantee this property and we
are not able to give additional constraints enforcing it.

However, we can prove the following lemma by exploiting the fact that
$\pos_C(H)$ is in $\2G$ determined by an R-node, by a P-node, or by a
common cutvertex.  The R-node, P-node, and common cutvertex
constraints of $\2G$ then help to prove the existence of the desired
face.

We call the union of all R-node, all P-node, all common cutvertex, and
all exclusive cutvertex constraints the \emph{position constraints}. 

\begin{figure}
  \centering
  \includegraphics[page=1]{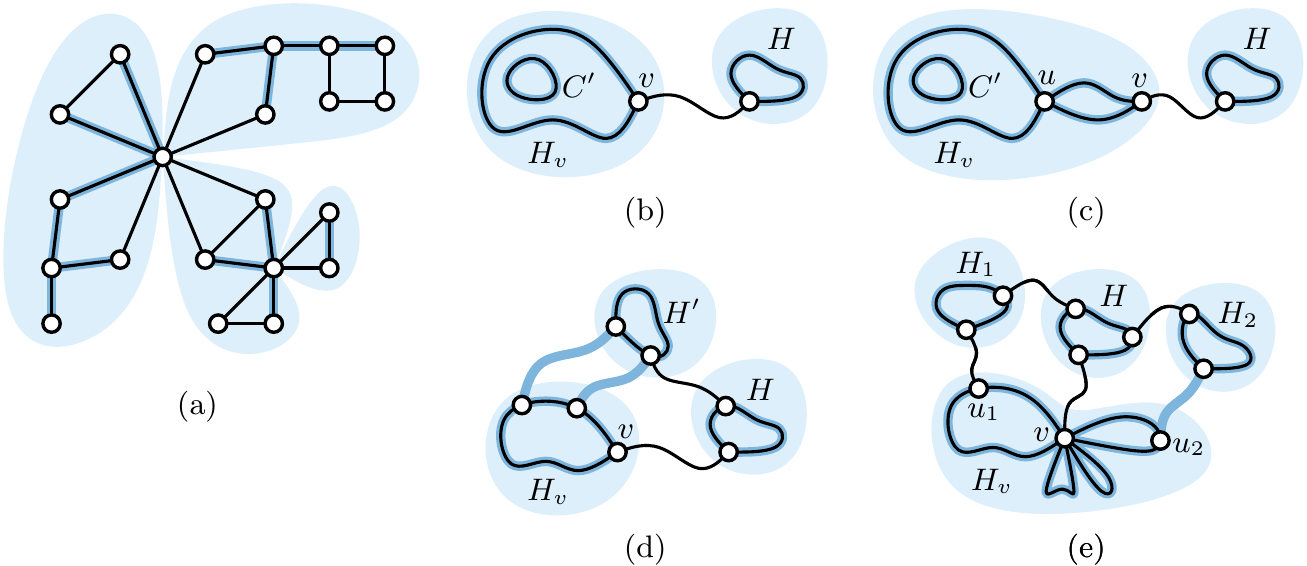}
  \caption{(a)~The extended blocks of the graph $\1G$.  (b)~The
    cutvertex $v$ and the cycle $C'$ are contained in the same block
    of $H_v$.  (c)~$v$ and $C'$ are in different blocks.  (d)~The
    exclusive bridge containing $H'$ has two attachments in $H_v$.
    Note that $H'$ and $H$ belong to the same union bridge.  (e)~The
    union bridge containing $H$ (and $H_1$ and $H_2$) has attachments
    in different blocks of $H_v$, i.e., it is not block-local.}
  \label{fig:rel-pos-excl-cutv}
\end{figure}

\begin{lemma}
  \label{lem:exclusive-cutvertex-constraints-wrap-up}
  Let $\1G$ and $\2G$ have common P-node degree~3 and simultaneous
  cutvertices of common degree at most~3.  Let $\alpha$ be a variable
  assignment for the ordering and position variables satisfying the
  ordering and position constraints with respect to $\1G$ and $\2G$.
  Then $\1G$ admits an embedding that realizes~$\alpha$.
\end{lemma}
\begin{proof}
  Let the blocks of $\1G$ be partitioned into a maximum number of
  partitions such that two blocks that each have a common edge
  incident to a cutvertex belong to the same partition; see
  Figure~\ref{fig:rel-pos-excl-cutv}a.  Let $\1{B_1}, \dots, \1{B_k}$ be
  the blocks of one such partition and let $\1P = \1{B_1}\cup \dots
  \cup \1{B_k}$.  Note that $\1P$ is a maximal subgraph of $\1G$ such
  that every split component with respect to a cutvertex $v$ includes
  a common edge incident to $v$.  We call $\1P$ an \emph{extended
    block} of $\1G$.  By
  Lemmas~\ref{lem:ordering-no-sim-cutv-with-cutv-var},~\ref{lem:r-node-constr-suff},~\ref{lem:p-node-constr-suff},
  and~\ref{lem:common-cutvertex-constr}, $\1P$ (and every other
  extended block) admits an embedding that realizes $\alpha$.

  Analogously, we define the extended blocks of $\2G$ and choose an
  embedding for each extended block in $\2G$ that realizes $\alpha$.

  To get an embedding of $\1G$ realizing $\alpha$, it remains to
  combine the embeddings of the extended blocks at the cutvertices
  separating them.  Let $\1{P_v}$ be the extended block of $\1G$ that
  includes a common edge incident to the cutvertex $v$ and let $\1{P}$
  be another extended block containing~$v$.  Let $H_v$ be the
  connected component of the common graph containing~$v$.  Note that
  $H_v$ is completely contained in $\1{P_v}$ and thus its embedding is
  already fixed.  Note further that $\1{P}$ is part of a split
  component with respect to $v$ and thus part of a single union bridge
  of $H_v$.  Thus, the exclusive cutvertex constraints make sure that
  the relative positions with respect to common cycles through~$v$ are
  the same for all common connected components in $\1{P}$.  Hence, the
  embeddings of $\1{P_v}$ and $\1{P}$ can be combined such that the
  resulting embedding of $\1{P_v} \cup \1{P}$ realizes $\alpha$ if and
  only if the relative positions of one common connected component $H$
  in the union bridge containing $\1{P}$ with respect to the common
  cycles through $v$ describe a face of $H_v$ that is incident to $v$.
  The latter follows immediately from the following two claims.

  \begin{claim}
    \label{claim:1}
    If the relative positions of $H$ with respect to common cycles
    through $v$ describe a face of $H_v$, then this face is incident
    to $v$.
  \end{claim}

  \begin{claim}
    \label{claim:2}
    The relative positions of $H$ with respect to all common cycles
    through $v$ describe a face of $H_v$.
  \end{claim}

  We start with the proof of Claim~\ref{claim:1}.  Let $f$ be a face
  of $H_v$ such that the relative positions of $H$ with respect to
  cycles through~$v$ (as given by $\alpha$) are the same as the
  relative positions of $f$ with respect to these cycles.  Then $f$ is
  incident to the cutvertex $v$ for the following reason.  If $f$ is
  not incident to $v$, then there exists a common cycle $C'$ (not
  containing $v$) in the connected component $H_v$ of $v$ that
  separates $v$ from a connected component $H$.  By
  Lemma~\ref{lem:extended-cycle-basis} we can assume that $C'$ is part
  of the extended cycle basis.

  There are two possibilities.  If $v$ belongs in $\1G$ to the same
  block as $C'$ (see Figure~\ref{fig:rel-pos-excl-cutv}b), the relative
  position $\pos_{C'}(H)$ is determined by the relative position of
  $v$ with respect to $C'$, as $v$ separates $H$ from $C'$ in $\1G$
  and $C'$ does not contain $v$.  Thus, we have the equation
  $\pos_{C'}(v) = \pos_{C'}(H)$ in this case.  Otherwise, $\1G$ has
  another cutvertex $u$ separating $v$ and $H$ from $C'$; see
  Figure~\ref{fig:rel-pos-excl-cutv}c.  Then $v$ and $H$ are in the same
  split component with respect to this cutvertex.  In this case we
  also have the requirement that $v$ and $H$ are on the same side of
  $C'$.  Hence, the cycle $C'$ cannot separate $v$ from $H$, which
  proves Claim~\ref{claim:1}.

  It remains to prove Claim~\ref{claim:2}.  Let $B^\cup$ be the union
  bridge of $H_v$ that contains $H$.  Recall that the exclusive
  cutvertex constraints require $\pos_C(H') = \pos_C(B^\cup)$ for
  every cycle $C$ of $H_v$ and every common connected component $H'$
  of $B^\cup$.  Thus, showing that the relative positions of $H'$
  describe a face of $H_v$ for one common connected component of
  $B^\cup$ shows this fact for all common connected components of
  $B^\cup$ (and thus in particular for~$H$).

  By Theorem~\ref{thm:no-block-local-exclusive-one-attaced}, we can
  assume one of the following is true.  The common connected component
  $H_v$ is a cycle; the union bridge $B^\cup$ is not block-local; or
  $B^\cup$ is not exclusive one-attached.

  If $H_v$ is a cycle, then there is only a single cycle through $v$
  and both sides of this cycle form a face of $H_v$.  Thus, there is
  nothing to show in this case.

  Assume the union bridge $B^\cup$ is not exclusive one-attached.
  Then there exists without loss of generality a $\circled{2}$-bridge
  $\2B$ that belongs to the union bridge $B^\cup$ such that $\2B$ has
  two attachments in $H_v$; see Figure~\ref{fig:rel-pos-excl-cutv}d
  (the case of a $\circled{1}$-bridge $\1B$ belonging to $B^\cup$ and
  having two attachments in $H_v$ is analogous).  Let further $H'$ be
  a common connected component contained in $\2B$.  Then $H'$ belongs
  in $\2G$ to a block that contains at least one block of $H_v$.
  Thus, the extended block containing $H'$ completely contains $H_v$.
  It follows that the relative positions of $H'$ with respect to
  cycles in $H_v$ describe a face of $H_v$.

  Finally, assume $B^\cup$ is not block-local.  Then there are two
  $\circled{i}$-bridges $\ii{B_1}$ and $\ii{B_2}$ (with $i \in \{1,
  2\}$) belonging to the union bridge $B^\cup$ with attachments in
  different blocks of $H_v$.  If one of these bridges has an
  attachment vertex in a block of $H_v$ not containing the cutvertex
  $v$, then the relative positions of this bridge with respect to any
  common cycle containing $v$ is determined by an R-node, by a P-node,
  or by a common cutvertex.  Thus, the relative positions correspond
  to a face of $H_v$ in this case.  It remains to consider the case
  that the attachment vertices of $\ii{B_1}$ and $\ii{B_2}$ belong to
  blocks of $H_v$ incident to $v$; see
  Figure~\ref{fig:rel-pos-excl-cutv}e.

  Let $S_1, \dots, S_k$ be the split components of the common
  component $H_v$ with respect to the cutvertex $v$.  Assume without
  loss of generality that $\ii{B_1}$ and $\ii{B_2}$ have their
  attachment vertices $u_1$ and $u_2$ in $S_1$ and $S_2$,
  respectively.  Let $H_1$ and $H_2$ be common connected components in
  $\ii{B_1}$ and $\ii{B_2}$, respectively.  The relative position of
  $H_1$ with respect to a cycle through $v$ that is not contained in
  $S_1$ is determined (in $\ii G$) by the common cutvertex $v$.  Thus,
  the relative positions of $H_1$ with respect to $S_2 \cup \dots \cup
  S_k$ describe a face of $S_2 \cup \dots \cup S_k$.  Moreover, this
  face contains the whole split component $S_1$.  Thus, if the
  relative positions of $H_1$ with respect to cycles in $S_1$ describe
  a face of $S_1$, then the relative positions with respect to cycles
  in $H_v = S_1 \cup \dots \cup S_k$ describe a face of $H_v$.
  Clearly, this is true as $H_1$ and $H_2$ have the same relative
  positions (they are in the same union bridge $B^\cup$) and the
  relative positions of $H_2$ with respect to cycles in $S_1$ describe
  a face of $S_1$ (one can use a symmetric argument to the one above).
  This concludes the proof.
\end{proof}

\subsubsection{Computing the Constraints}
\label{sec:cump-constr}

Recall from Section~\ref{sec:cons-edge-order}
(Lemma~\ref{lem:ordering-no-sim-cutv-with-cutv-var}) that we have
potentially $O(n^3)$ cutvertex-ordering variables.  Moreover, there
are $O(n^2)$ cycles in the extended cycle basis $\mathcal C$ and thus
$O(n^3)$ component position variables.  Thus, our aim is to compute
the position constraints described in the previous sections in
$O(n^3)$ time.

Let $C \in \mathcal C$ be a cycle.  For the relative positions with
respect to $C$ that are determined by R-nodes or P-nodes, we need to
know for every R- and P-node $\mu$ of $\1G$ and of $\2G$, whether it
induces a cycle $\kappa$ in $\skel(\mu)$.  If so, we also need to know
the cycle $\kappa$.  This can clearly be done in linear time for each
cycle, yielding a total running time of $O(n^3)$ (note that techniques
from~\cite{br-drpse-15} can be used to compute this information for
multiple cycles in linear time).

Similarly, in $O(n^2)$ time, we can compute for every virtual edge
$\eps$ in $\skel(\mu)$, which common connected components are
contained in the expansion graph of $\eps$.  Assume $\mu$ is an R-node
and let $X$ be the set of ordering variables determined by $\mu$.
With the above information,  one can easily compute the R-node
constraints for $\mu$ in $O(|X| + |\skel(\mu)|)$ time.  As each
relative position is determined by at most one R-node, the sets $X$
are disjoint for different R-nodes.  Thus, we get a total running time
of $O(n^3)$ for computing the R-node constraints.

Computing the P-node constraints of a P-node $\mu$ can be done
analogously (yielding $O(n^3)$ running time in total), except for the
case where we have to handle a P-node 4-constraint.  Recall that we
get P-node 4-constraints if three virtual edges $\eps_1$, $\eps_2$,
and $\eps_3$ of $\skel(\mu)$ include common paths between the poles.
For every other virtual edge $\eps$, we then get a P-node
4-constraint, which makes $O(|\skel(\mu)|)$ P-node 4-constraints for
$\mu$.  For the P-node 4-constraint corresponding to the virtual edge
$\eps$, we have to check whether the union graph $G^\cup$ has a path
$\pi$ from $\expan(\eps)$ to $\expan(\eps_i)$ (for $i \in \{1, 2,
3\}$) that is disjoint from the expansion graphs $\expan(\eps_j)$ for
$j \in \{1, 2, 3\}$ with $i \not= j$.  This can clearly be done in
$O(n)$ time for each edge $\eps$ of $\skel(\mu)$.  It follows that we
can compute the P-node constraints in $O(n^3)$ time.

For a cutvertex $v$ of $\1G$, consider the relative positions $X$
determined by the common cutvertex $v$.  For every split component
$\1S$ and every common connected component $H$ in $\1S$, we have the
constraint $\pos_C(H) = \pos_C(\1 S)$.  These constraints can be
easily computed in $O(n + |X|)$ time.  As the sets $X$ are disjoint
for every cutvertex, this yields a total running time of $O(n^3)$.
Moreover, we have to compute constraints of the type $\ord(e_1, e_2,
\1B) = \pos_C(\1S)$ connecting the relative positions to the cutvertex
ordering variables.  Clearly, for each variables $\pos_C(\1S)$ this
constraint can be added in constant time, which yields a running time
in $O(|X|)$.  Hence, the common cutvertex constraints can be computed
in $O(n^3)$ time.

Finally, consider the relative positions $X$ determined by the
exclusive cutvertex $v$ and let $H_v$ be the common connected
component containing $v$.  For every union bridge $B^\cup$, every
common connected component $H$ in $B^\cup$, and every common cycle
through $v$, we have to add the constraint $\pos_C(B^\cup) =
\pos_C(H)$.  We can first (in $O(n)$ time) partition the common
connected components according to their union bridges.  Then, adding
these constraints for one cycle $C$ can be done in $O(|X_C|)$ time,
where $X_C$ is the set of relative positions with respect to $C$ in
$X$.  Thus, we get the exclusive cutvertex constraints for $v$ in
$O(|X|)$ time, which yields a total running time of $O(n^3)$.

\begin{lemma}
  \label{lem:running-time-position-constraints}
  The position constraints can be computed in $O(n^3)$ time.
\end{lemma}

\subsection{Putting Things Together}
\label{sec:putt-things-togeth}

Assume $(\1G, \2G)$ is a \textsc{Sefe} instances such that $\1G$ and
$\2G$ are connected graphs, $\1G$ and $\2G$ have common P-node
degree~3, and every simultaneous cutvertex has common degree~3.

We first used Theorem~\ref{thm:no-union-cutvertices} to get rid of all
union cutvertices.  This helped to ensure consistent edge orderings in
Section~\ref{sec:cons-edge-order}.  Actually, without union
cutvertices, we know for each vertex $v$ that it is either not a
cutvertex in one of the graphs $\1G$ or $\2G$, which makes
representing the possible edge orderings much simpler, or that it has
common degree~3, which also makes the ordering simple.

To ensure consistent relative positions of two common connected
components with respect to each other, we first showed in
Section~\ref{sec:relat-posit-with} that it suffices to ensure
consistent relative positions of each common connected component with
respect to the cycles of a cycle basis in the other component.
Unfortunately, setting relative positions with respect to cycles does
not necessarily lead to an embedding (e.g., if a cycle $C_1$ lies
``inside'' $C_2$, and $C_2$ lies ``inside'' $C_3$, then $C_3$ cannot
lie ``inside'' $C_1$).

This leads to difficulties, when one component $H_1$ can be
potentially embedded into several faces of another component $H_2$,
which is the case when $H_1$ is attached to $H_2$ via only two
vertices that are a separating pair of $H_2$, or when $H_1$ and $H_2$
are separated by a cutvertex.  For the former case, it helped to
assume that split components of union separating pairs have a very
special structure
(Theorem~\ref{thm:only-special-union-separating-pairs}).  For the
latter case, it helped to assume that there are no union bridges that
are block-local and exclusively one-attached
(Theorem~\ref{thm:no-block-local-exclusive-one-attaced}).

Using Theorem~\ref{thm:no-block-local-exclusive-one-attaced} comes at
the cost that we have to satisfy some common-face constraints.
However, in Section~\ref{sec:common-face-constr} we showed that this
can be done easily (Lemma~\ref{lem:common-face-constraints}).

The set of equations and inequalities we obtain has size $O(n^3)$, can
be computed in $O(n^3)$ time, and can be solved in linear time in its
size~\cite{apt-ltatt-79}.  This lets us conclude with the following
theorem.

\begin{theorem}
  \label{thm:sefe-for-comm-deg-3}
  \textsc{Sefe} can be solved in $O(n^3)$ time for two connected graphs
  with common P-node degree~3 and simultaneous cutvertices of common
  degree at most~3.
\end{theorem}

\section{Conclusion}
\label{sec:conclusion}

In this paper, we have shown how to combine techniques for ensuring
consistent relative positions~\cite{br-drpse-15} with
known~\cite{adfpr-tsegi-12,br-spoacep-13} and newly developed tools
for ensuring consistent edge orderings.  This lead to an efficient
algorithm solving \textsc{Sefe} for two connected graphs with common
P-node degree~3 and simultaneous cutvertices of common degree at
most~3.  Together with the linear-time algorithm for decomposing a
given instance into equivalent instances in which each 2-component is
a cycle, this gives an efficient algorithm if each connected component
of the common graph is biconnected, has maximum degree~3, or is
outerplanar with maximum degree~3 cutvertices.

We note that all techniques developed in
Section~\ref{sec:graph-with-common} extend to the sunflower case,
where we have multiple graphs pairwise intersecting in the same common
graph.  Actually, the two graphs $\1G$ and $\2G$ are only considered
together if $\2G$ restricts the embedding choices of the common graph
in $\1G$ in a way that makes it possible to formulate certain
constraints.  Thus, more graphs intersecting in the same common graph
can only help.  Moreover, the preprocessing algorithms from
Section~\ref{sec:prepr-algor} also directly extend to the sunflower
case when adapting the definition of impossible P-nodes
(Lemma~\ref{lem:no-impossible-P-nodes}) in a straightforward manner.

Besides solving this fairly general set of \textsc{Sefe} instances,
our results and in particular the preprocessing algorithms give some
new structural insights that may help in further research.  For
example, Theorem~\ref{thm:all-link-graphs-are-connected}, stating that
one can assume all union-link graphs to be connected, not only helps
in later chapters but also shows that the decision of ordering virtual
edges in P-nodes of the common graph is fairly easy.

What remains poorly understood are the edge orderings at cutvertices.
We were basically able to handle cutvertices if the choices boil down
to binary decision.  This is for example the case if the cutvertex has
only common degree~3.  Although less obvious, this is also the case if
the instance has common P-node degree~3.  For a cutvertex in $\1G$,
this basically means that the other graph $\2G$ hierarchically groups
the common edges incident to the cutvertex such that there are at most
three groups on each level, yielding binary decisions.

We feel that the understanding of edge orderings at cutvertices is a
major bottleneck for solving more general instances of {\sc Sefe}.  To
get a better understanding, we believe that it would be helpful to
first study cutvertices in the context of constrained planarity
problems, which are often a bit more simple.  A good starting point
could be the constrained planarity problems planarity with partitioned
PQ-constrains or variants like partitioned
full-constraints~\cite{br-npcpcep-14}.

\bibliographystyle{plain}
\bibliography{se}

\end{document}